\documentclass[12pt]{article}
\usepackage{amsmath,amssymb,fullpage,graphicx}
\usepackage{enumerate}
\RequirePackage[round]{natbib}
\RequirePackage[colorlinks,citecolor=blue,urlcolor=blue]{hyperref}

\usepackage{url} 

\usepackage{multirow}
\usepackage{subfigure}
\usepackage{sectsty}
\usepackage{longtable}
\usepackage{algorithmic}

\usepackage{xcolor}

\let\hat\widehat
\let\tilde\widetilde

\newtheorem{theorem}{Theorem}
\newtheorem{theorem2}{Theorem}
\newtheorem{lemma}[theorem]{Lemma}

\newtheorem{proposition}[theorem]{Proposition}
\newtheorem{assumption}[theorem2]{Assumption}

\newenvironment{proof}{{\bf Proof.}}{$\Box$}

\newcommand{\E}{\mbox{$\mathbb{E}$}}
\newcommand{\R}{\mbox{$\mathbb{R}$}}

\usepackage{array}

\newcommand{\blind}{1}

\begin{document}

	\def\spacingset#1{\renewcommand{\baselinestretch}%
		{#1}\small\normalsize} \spacingset{1}

	\if1\blind
	{
		\title{\bf A statistical framework for analyzing activity pattern from GPS data}
	\author{Haoyang Wu,
			Yen-Chi Chen\thanks{
				Contact: \url{yenchic@uw.edu}}, Adrian Dobra\thanks{
				The authors gratefully acknowledge the support of \textit{NSF DMS-2310578, DMS-2141808 and NIH U24-AG072122.}\hspace{.2cm}}\\
		}
		\maketitle
	} \fi
	
	\if0\blind
	{
		\bigskip
		\bigskip
		\bigskip
		\begin{center}
			{\LARGE\bf A statistical framework for analyzing activity pattern from GPS data}
		\end{center}
		\medskip
	} \fi
	
	\bigskip
	\begin{abstract}
		We introduce a novel statistical framework for analyzing the GPS data of a single individual. Our approach models daily GPS observations as noisy measurements of an underlying random trajectory, enabling the definition of meaningful concepts such as the average GPS density function. We propose estimators for this density function and establish their asymptotic properties.
		To study human activity patterns using GPS data, we develop a simple movement model based on mixture models for generating random trajectories. Building on this framework, we introduce several analytical tools to explore activity spaces and mobility patterns. We demonstrate the effectiveness of our approach through applications to both simulated and real-world GPS data, uncovering insightful mobility trends.
		
	\end{abstract}
	
	\noindent%
	\vfill
	
	
	\newpage
	\spacingset{1.9} 
	
	\section{Introduction}	\label{sec::intro}

	The rapid advancement of Global Positioning System (GPS) technology has revolutionized various domains, ranging from transportation and urban planning to environmental studies and public health \citep{papoutsis2021relaxing, yang2022identifying, mao2023temporal}. GPS data provides high-resolution spatial and temporal information, enabling precise tracking of movement and location patterns. With the increasing availability of GPS-equipped devices, such as smartphones, vehicles, and wearable deices, vast amounts of trajectory data are generated every day.
	
	In scientific research, GPS data play a crucial role in understanding human mobility, tracing
	activity spaces, and analyzing travel behaviors \citep{barnett2020inferring, andrade2020identifying, elmasri2023predictive}. Researchers leverage GPS data to study commuting patterns, urban dynamics, wildlife tracking, and disaster response. By examining spatial-temporal GPS records, scientists gain insights into mobility patterns of individuals and the dynamics of population movements. Despite the rich information embedded in the GPS data, little is known about how to properly analyze the data from a statistical perspective. 
	
	
	While various methods exist for processing GPS data, a comprehensive statistical framework for analysis is still lacking. 
	Traditional approaches often rely on heuristic methods or machine learning techniques \citep{newson2009hidden, early2020smoothing, chabert2023auto, burzacchi2024generating}
	that may lack rigorous statistical foundations. 
	A major challenge in modeling GPS data is that records cannot be viewed as independently and identically distributed (IID)
	random variables. 
	Thus, a conventional statistical model does not work for GPS records.
	The absence of a reliable statistical framework hinders cross-disciplinary applications and reduces the potential for scientific applications.
	
	The objective of this paper is to introduce a general statistical framework for analyzing GPS data.
	We introduce the concept of random trajectories and model GPS records as a noisy measurement of points on a trajectory. 
	Specifically, we assume that  on each day, the individual has a true trajectory that is generated from an IID process.
	Within each day, the GPS data are  locations on the trajectory with additive measurement noises recorded at particular timestamps.
	The independence across different days allow us to use daily-average information to infer 
	the overall patterns of the individual.
	

	\subsection{Related work}

	Our procedure for estimating activity utilizes a weighted kernel density estimator (KDE; \citealt{chacon2018multivariate, wasserman2006all, silverman2018density, scott2015multivariate}). The weights assigned to each observation account for the timestamp distribution while adjusting for the cyclic nature of time, establishing connections to directional statistics \citep{mardia2009directional, meilan2021nonparametric, pewsey2021recent, ley2017modern}.
	
	Our analysis of activity space builds on the concept of density level sets \citep{chacon2015population, tsybakov1997nonparametric, cadre2006kernel, chen2017density} and level sets indexed by probability \citep{polonik1997minimum, 10.1214/19-AOAS1311, chen2019generalized}. Additionally, the random trajectory model introduced in this paper is related to topological and functional data analysis \citep{chazal2021introduction, wang2016functional, wasserman2018topological}.
	
	Recent statistical research on GPS data primarily focuses on recovering latent trajectories \citep{early2020smoothing, duong2024using, huang2023reconstructing, burzacchi2024generating} or analyzing transportation patterns \citep{mao2023temporal, elmasri2023predictive}. In contrast, our work aims to develop a statistical framework for identifying an individual's activity space and anchor locations \citep{10.1214/19-AOAS1311}. Activity spaces are the spatial areas an individual visits during their activities of daily living. Anchor locations represent mobility hubs for an individual: they are locations the individual spends a significant amount of their time that serve as origin and destination for many of the routes they follow \citep{schonfelder-axhausen-2004}.
	
	%
	%

	
	%
	%
	%
	%
	%
	%
	%
	%
	%
	%
	%
	%
	%
	%
	%
	%
	%
	
	\subsection{Outline}
	
	This paper is organized as follows. 
	In Section \ref{sec::gps}, we formally introduce our statistical framework for modeling GPS data. 
	This framework describes the underlying statistical model on how GPS data are generated. 
	Section \ref{sec::estimation} describes estimators
	of the key parameters of interest introduced in Section \ref{sec::gps}. Section \ref{sec::asymp} presents the associated asymptotic theory. To study human activity patterns, we introduce a simple movement model in Section \ref{sec::smm}
	and present some statistical analysis based on this model. 
	We apply our methodology to a real GPS data in Section \ref{sec:real-data} and explain the key insights about human mobility gained from our analysis.
	All codes and data are available on \url{https://github.com/wuhy0902/Mobility_analysis}.
	
	\section{Statistical framework of  GPS data}	\label{sec::gps}
	
	\begin{figure}
		\center
		\includegraphics[width=1.5in]{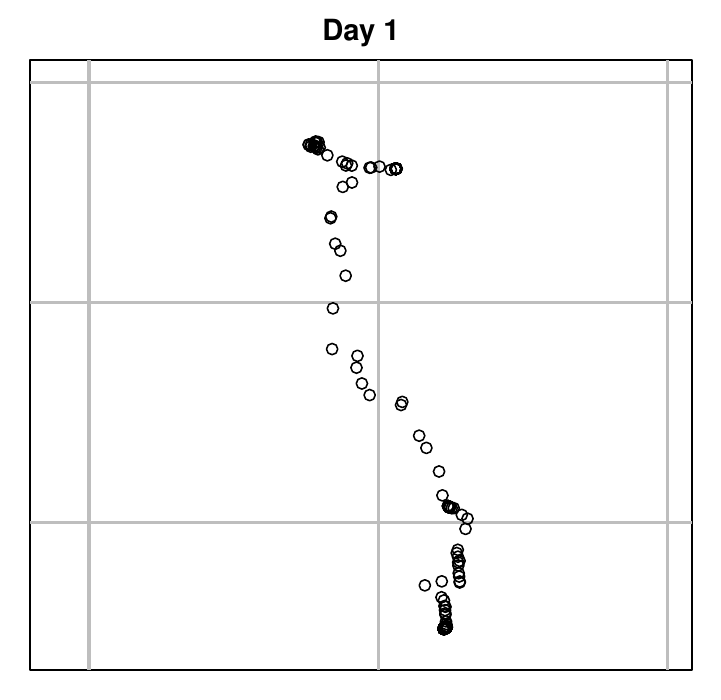}
		\includegraphics[width=1.5in]{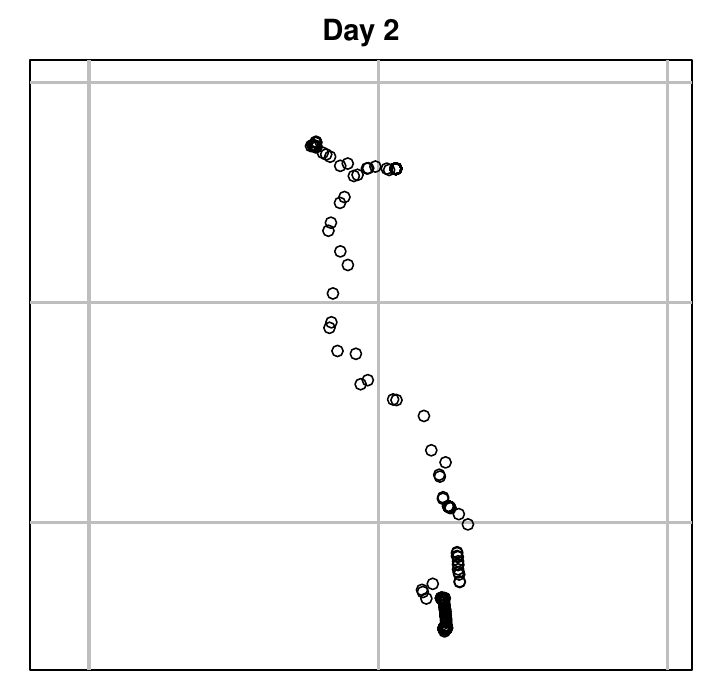}
		\includegraphics[width=1.5in]{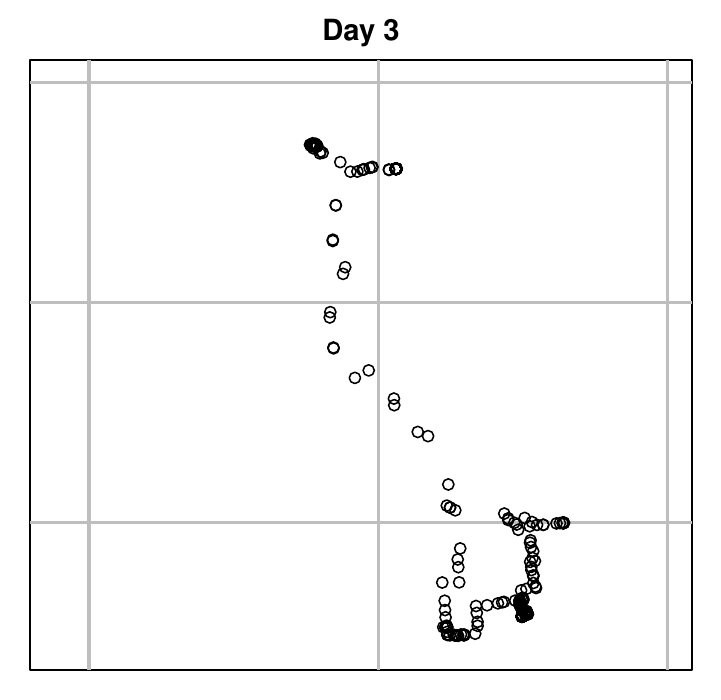}
		\includegraphics[width=1.5in]{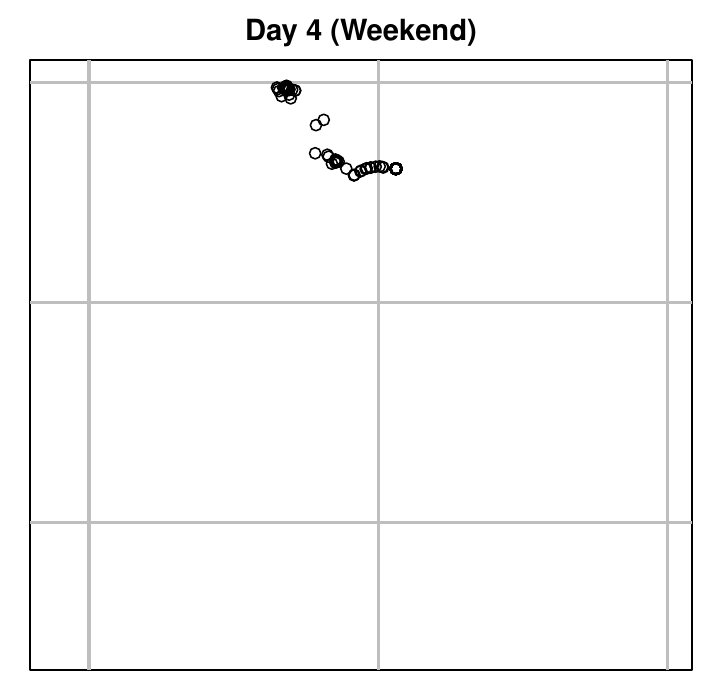}
		\caption{An example of real GPS data from an example individual (without timestamps).
			The four panels show the GPS records of this individual in four different days.
			The first three panels show weekdays whereas the last panel shows a weekend day. 
			This is part of the GPS data that 
			we analyze in Section~\ref{sec:real-data}.}
		\label{fig::example}
	\end{figure}
	
	We consider GPS data recorded from a single individual over several non-overlapping time periods of the same length. In what follows we will consider a time period to be one day, but our results extend to any other similar time periods that can be longer (e.g., weeks, months, weekdays, weekends) or shorter (e.g., daytime hours between 7:00am and 5:00pm).

	The GPS data can be represented as random variables
	$$
	(X_{i,j}, t_{i,j}) \in \mathbb{R}^2 \times [0,1],
	$$
	where $i=1,\ldots, n$ denotes the day in which the observation was recorded and 
	$j=1,\ldots, m_i$ denotes the sequence of observations for $i$-th day. Here $X_{i,j}$ stands for the spatial location of this GPS observation (latitude and longitude), and $t_{ij}$ represents the time when the observation was recorded.
	Note that  for each day $i$,
	$
	t_{i,1}<t_{i,2}<\ldots<t_{i,m_i}.
	$
	Here $m_i$ is the total number of observations in day $i$.
	We assume that the timestamps $t_{i,j}$ are non-random (similar to a fixed design). 
	Figure~\ref{fig::example} provides an
	example of real GPS data recorded from an individual in four different days (ignoring timestamps).

	
	A challenge in modeling GPS data is that observations within the same day are \emph{dependent}
	since each observation can be viewed as a realization of the individual's trajectory 
	on discrete time point (with some measurement errors).  Thus, it is a non-trivial task to formulate an adequate statistical model that accurately captures the nature of GPS data.

	\subsection{Latent trajectory model for GPS records}
	
	To construct a statistical framework for GPS records, we consider a random trajectory model 
	as follows.
	For $i$-th day, 
	there exists a continuous mapping
	$
	S_i: [0,1]\rightarrow \mathbb{R}^2
	$
	such that $S_i(t)$ is the actual location of the individual on day $i$ and time $t$.
	We call $S_i$ the (latent) trajectory of $i$-th day.

	The observed data is 
	$
	X_{i,j} = S_i(t_{i,j}) + \epsilon_{i,j},
	$
	where $\epsilon_{i,j}$ is independent noises (independent across both $i$ and $j$)
	with a bounded PDF that are often assumed to be Gaussian.
	We also write $X_{R,i,j} = S_i(t_{i,j})$  to denote the actual location of the individual at time $t_{i,j}$.
	
	We assume that the latent trajectories 
	$
	S_1,\ldots, S_n\sim P_S,
	$
	where $P_S$ is a distribution of random trajectory.
	Namely, on each day, the individual follows
	a random trajectory independent of all other trajectories
	from $P_S$.
	In Section \ref{sec::smm}, we provide a realistic example
	on how to generate a random trajectory.

	Figure~\ref{fig::DGP}
	presents an example of the latent trajectory model. 
	The first column presents a map of important locations such as home and office,
	and the corresponding road networks inside this region. 
	The second column shows random trajectories of two days.
	The actual trajectory is shown in black lines. The red arrows indicate the direction of the trajectory. 
	Note that the individual might remain at the same location for an extended period of time as indicated as the solid black dots  in the second column. The third column presents the observed locations of the GPS records.
	We note that GPS records mostly follow the individual's trajectory but may occasionally deviate from it due to measurement errors.

	\begin{figure}
		\center
		\includegraphics[height=1.5in]{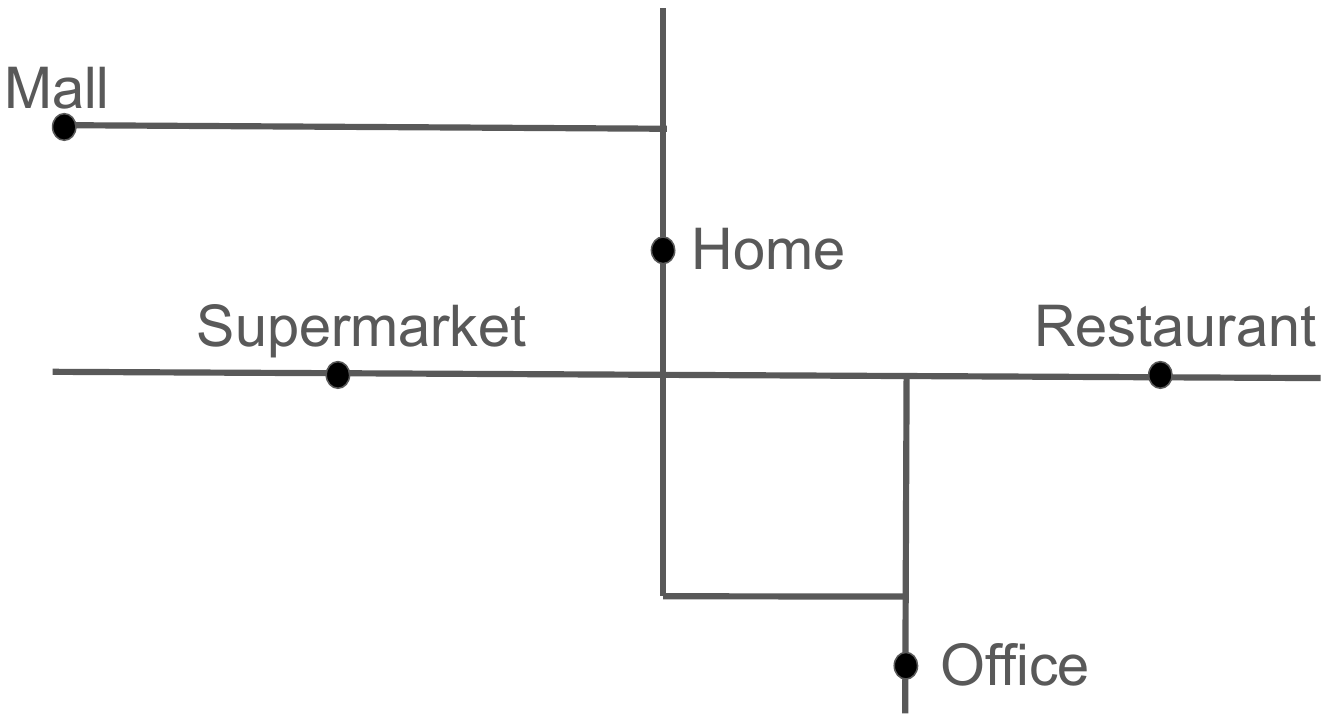}
		\includegraphics[height=1.5in]{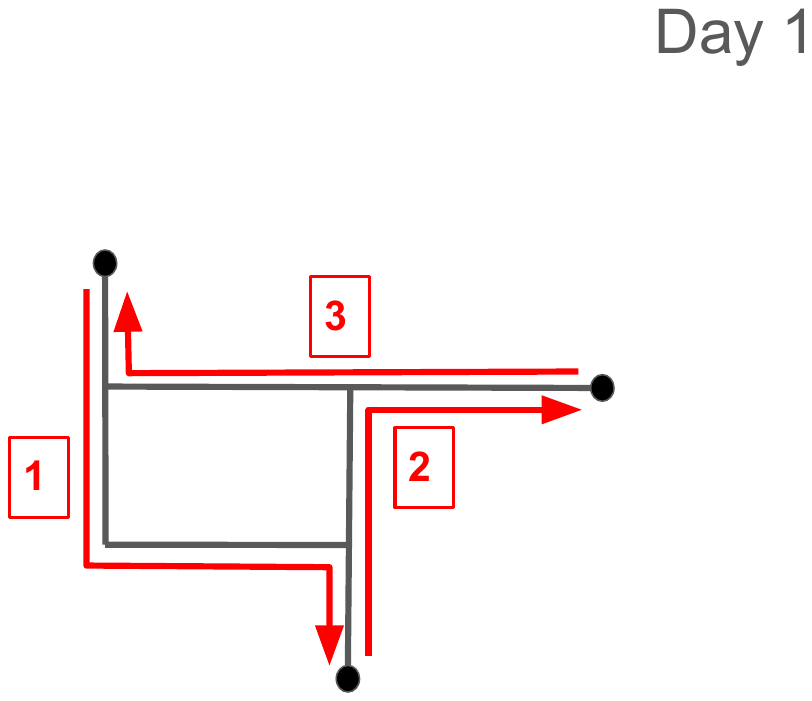}
		\includegraphics[height=1.5in]{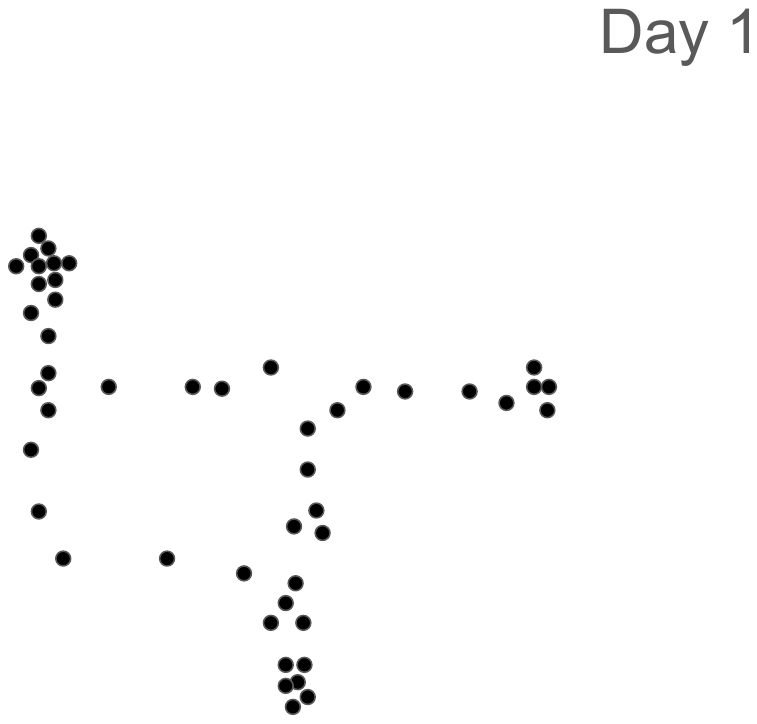}\\
		\includegraphics[height=1.5in]{figures/GPS_ex1_1.pdf}
		\includegraphics[height=1.5in]{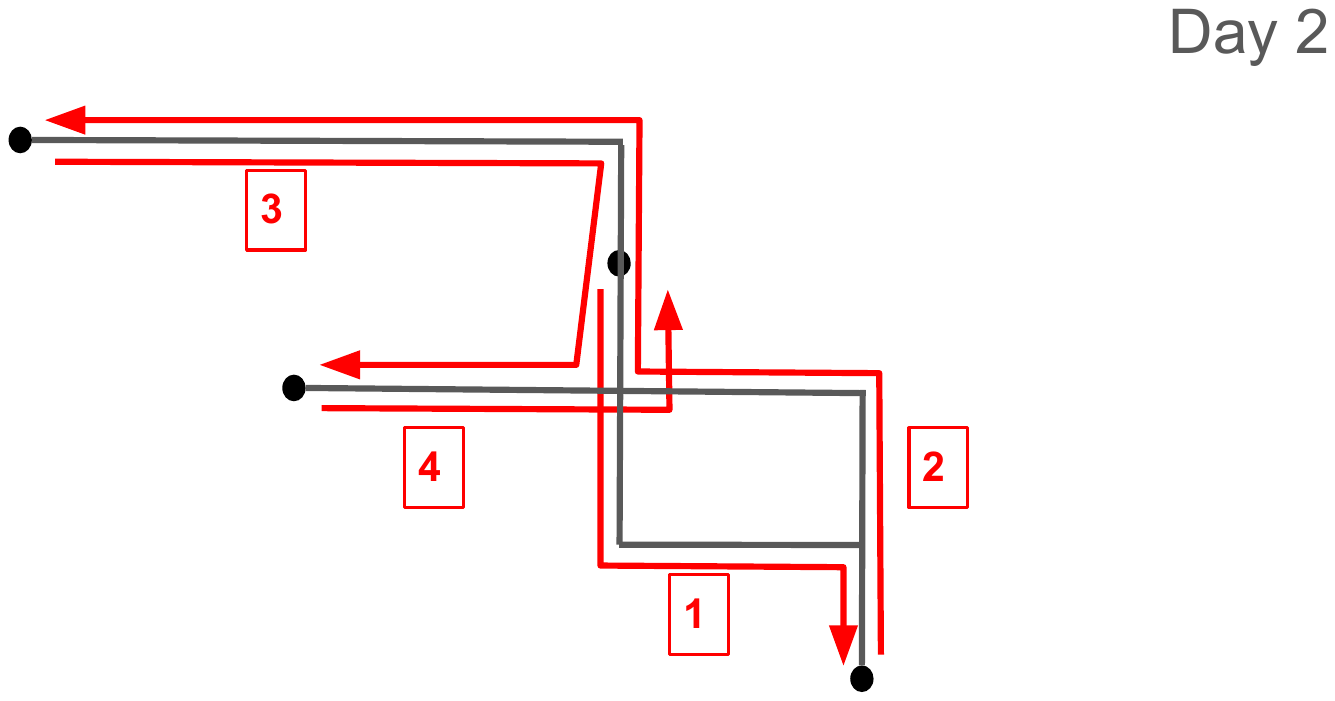}
		\includegraphics[height=1.5in]{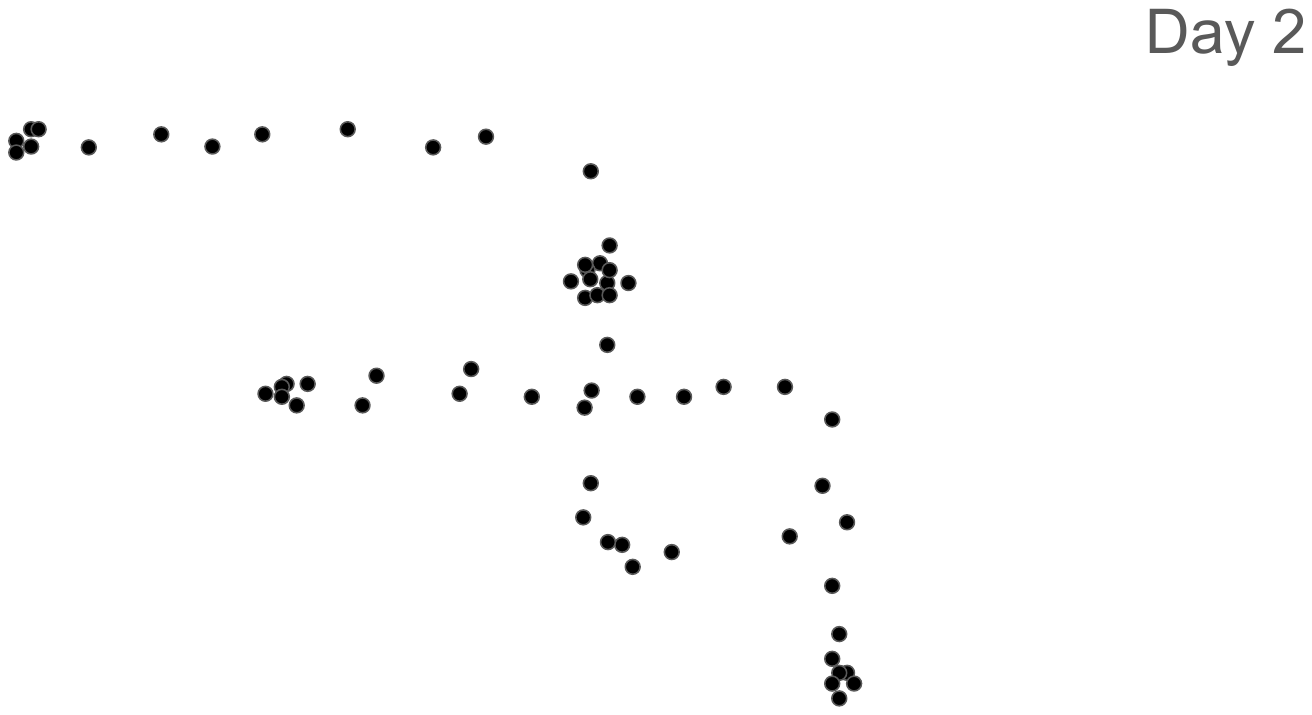}
		\caption{An example of data generating process for GPS records. The first column shows the map of some important locations and the roads visited by an individual during their activities of daily living. The second column shows two possible trajectories. Black lines represent the true trajectory followed by the individual. Red arrows display movement along these trajectory. The third  column shows the observed GPS data for each day.}
		\label{fig::DGP}
	\end{figure}
	
	\subsection{Activity patterns}
	
	The individual's activity pattern is
	characterized by the distribution of the random trajectory $S(t)$.
	Let 
	$A\subset \R^2$
	be a spatial region.
	The probability
	$
	P(S(t)\in A)
	$
	is the probability that the individual is in the region $A$ at time $T=t$.
	The manner in which distribution of $S(t)$ changes over time provides us with information about the dynamics of the most likely locations visited by this individual. Let $U\sim {\sf Uniform}[0,1]$  be a random variable independent of $S(t)$.
	The probability
	$$
	P(S(U) \in A) =  \int_0^1 P(S(t)\in A)dt.
	$$
	can be interpreted the \emph{proportion} of time that the individual spends in the region $A$. 
	Thus, the distribution of $S(U)$ tells us the overall activity patterns of the individual.


	\subsection{GPS densities}
	
	The distributions of $S(t)$ or $S(U)$
	are difficult to estimate
	from the GPS data
	because of the measurement error $\epsilon$.
	This is known as the deconvolution problem
	and is a notoriously difficult task.
	Thus, our idea is to 
	use the distribution of $X_t = S(t)+\epsilon$
	and $X^* = S(U)+\epsilon$
	as surrogates.


	Naively, one may consider the 
	the marginal PDF of $X_{i,j}$
	\begin{equation}
		f_X(x) = \lim_{r\rightarrow 0} \frac{P(X_{i,j} \in B(x,r))}{\pi r^2}
		\label{def:GPS density}
	\end{equation}
	as an estimate of the distribution of $X^*$.
	This geometric approach \citep{Introduction_geo} is a convenient way to define density in this case. 
	However, this PDF is NOT the PDF of $X^*$
	because
	it depends on the distribution of the timestamps $T_{i,j}$.
	Even if the individual's latent trajectory remains
	the same, changing the distribution of timestamps $\{T_{i,j}: j=1,\ldots, m_i\}$
	will change the distribution of $f_X$.

	
	{\bf Average GPS density.}
	To solve this problem, 
	we consider a time-uniform density of the observation.
	Let $U\sim [0,1]$ be a uniform random variable for the time interval $[0,1]$. Recall that
	$$
	X^* = S(U) + \epsilon,
	$$
	where $S\sim P_S$ and $\epsilon \sim P_\epsilon$ are independent random quantities.
	$X^*$ can be viewed as the random latent position $S(U)$ corrupted by the measurement error $\epsilon$.
	The average behavior of $X^*$
	describes the expectation of a random trajectory (with measurement error)
	that is invariant to the timestamps distribution. 
	With this, we define
	the \emph{average} GPS density is the PDF of $X^*$:
	\begin{equation}
		f_{GPS}(x) =  \lim_{r\rightarrow 0} \frac{P(X^* \in B(x,r))}{\pi r^2}.
	\end{equation}
	The quantity $f_{GPS}$ describes the expectation of a random trajectory with consideration to the measurement error
	and is invariant to the distribution of timestamps.
	What influences $f_{GPS}$ is the distribution of the latent trajectory $P_S$
	and the measurement errors (which are generally small). 
	Thus, $f_{GPS}(x)$ is a good quantity to characterize the activity pattern based on GPS records.

	{\bf Interval-specific GPS density.}
	Sometimes we may be interested in the activity pattern that takes place within a specific time-interval, e.g., the morning or evening rush hours. Let $A\subset [0,1]$ be an interval of interest.
	The activity patterns within $A$ can be described by the following interval-specific GPS density
	\begin{equation}
		\begin{aligned}
			f_{GPS, A}(x) &= \lim_{r\rightarrow 0} \frac{P(X^*_A \in B(x,r))}{\pi r^2},\\
			X^*_A &= S(U_A)+\epsilon,\\
			U_A&\sim {\sf Uniform}(A),
		\end{aligned}
		\label{def:GPS density in area}
	\end{equation}
	where ${\sf Uniform}(A)$ is the uniform distribution over the interval $A\subset[0,1]$.

	{\bf Conditional GPS density.}
	When the interval $A$ is shrinking to a specific point $t$, 
	we obtain the \emph{conditional GPS density}:
	\begin{equation}
		\begin{aligned}
			f_{GPS}(x|t) &= \lim_{r\rightarrow 0} \frac{P(X^*_t\in B(x,r))}{\pi r^2},\\
			X^*_t &= S(t)+\epsilon.
		\end{aligned}
	\end{equation}
	The conditional GPS can be used  for predicting individual's location
	at time $t$ when we have no information about the individual's latent trajectory.
	The conditional GPS density and the previous two densities are linked via the following proposition.
	
	\begin{proposition}
		\label{prop::GPS}
		The above three GPS densities are linked by the following equalities:
		\begin{align*}
			f_{GPS}(x) = \int_0^1 f_{GPS}(x|t)dt,
			\qquad f_{GPS,A}(x) = \int_A f_{GPS}(x|t)dt.
		\end{align*}
	\end{proposition}

	

	\subsection{Identification of activity patterns from GPS densities}

	Even if there are measurement errors, the GPS densities still provide
	useful bound on the distribution of $S(U)$.
	For simplicity, we assume that the measurement error $\epsilon\sim N(0, \sigma^2\mathbb{I}_2)$.
	Namely, the measurement error is an isotropic Gaussian with variance $\sigma^2$.
	Let $F_{\chi^2_2}(t)$ be the CDF of the $\chi^2_2$ distribution.
	
	\begin{proposition}[Anchor point identification]
		\label{prop::AP}
		Let $U\sim {\sf Uniform}[0,1]$ be a uniform random variable.
		Suppose $a\in \R^2$ is a location such that 
		$
		P(S(U)=a) \geq \rho_a.
		$
		Then 
		$
		f_{GPS}(a) \geq   \frac{\rho_a}{2\pi\sigma^2}
		$
		and 
		$
		\int_{B(a, r)} f_{GPS}(x)dx \geq \rho_a \cdot F_{\chi^2_2}\left(\frac{r^2}{\sigma^2}\right).
		$

	\end{proposition}

	Proposition \ref{prop::AP} shows that if the location $a$ has a probability mass,
	the average GPS density $f_{GPS}$ will also be a non-trivial number. 
	A location with a high probability mass of $S(U)$ 
	means that on average, the individual has spent a substantial amount of 
	time staying there.
	This is generally a location of interest and
	is called an anchor location if the individual has spent a lot of time there on
	a daily basis.
	The home and work place are good examples of anchor locations.

	Proposition \ref{prop::AP} can be generalized to a region. 
	Let $C\oplus r = \{x\in\R^2: d(x,C) \leq r\}$
	be the $r$-neighboring region around the set $C$.
	Then we have the following result.

	\begin{proposition}[High activity region]
		\label{prop::HA}
		Let $U\sim {\sf Uniform}[0,1]$ be a uniform random variable.
		Suppose $C\subset \R^2$ is a regions where
		$
		P(S(U)\in C) \geq \rho_C.
		$
		Then 
		$
		\int_{C\oplus r} f_{GPS}(x)dx \geq \rho_C \cdot F_{\chi^2_2}\left(\frac{r^2}{\sigma^2}\right).
		$
		
	\end{proposition}
	
	Proposition \ref{prop::HA} shows that if the random trajectories $P_S$ 
	had spent some time in the region $C$,
	the average GPS density will also be a non-small number around $C$.
	The CDF of a $\chi^2_2$ distribution is due to the measurement error being an
	isotropic two-dimensional Gaussian.
	If the measurement error has a different distribution,
	this quantity has to be modified.

	Finally, a region with a high average GPS probability will imply its neighborhood has 
	more activities.
	
	\begin{proposition}[GPS density to activity]
		\label{prop::GA}
		Let $U\sim {\sf Uniform}[0,1]$ be a uniform random variable.

		Suppose $C\subset \R^2$ is a regions where
		$
		\int_{C} f_{GPS}(x)dx  \geq G_C.
		$
		Then 
		$$
		P(S(U)\in C\oplus r) \geq \frac{G_C - (1- F_{\chi^2_2}(\frac{r^2}{\sigma^2}))}{F_{\chi^2_2}(\frac{r^2}{\sigma^2})} = 1- \frac{1-G_C}{F_{\chi^2_2}(\frac{r^2}{\sigma^2})}
		$$
		for any $r>0$.
		
	\end{proposition}
	The lower bound in Proposition \ref{prop::AP} is meaningful only if $F_{\chi^2_2}(\frac{r^2}{\sigma^2})\geq 1-G_C$. 
	This proposition shows that if the density $f_{GPS}(x)$ is high,
	we will expect some high probability around of $X(U)$ around $x$.

	Due to Propositions \ref{prop::AP}, \ref{prop::HA}, and \ref{prop::GA},
	the average density $f_{GPS}$ captures
	high density regions of the latent trajectory $S(t)$.
	As a result, inference on $f_{GPS}$
	provides us valuable information about the individual's actual activities.

	\section{Estimating the GPS densities}	\label{sec::estimation}
	
	In this section, we study the problem of estimating the GPS densities.
	Recall that our data are 
	$
	(X_{i,j}, t_{i,j}) \in \mathbb{R}^2 \times [0,1],
	$
	where $i=1,\ldots, n$ is the indicator of different days and
	$j=1,\ldots, m_i$ is the indicator of the time at each day. 
	Our method is based on the idea of kernel smoothing. 
	We will show that after properly reweighting the observations according to the timestamp,
	a 2D kernel density estimator (KDE)
	can provide a consistent estimator of the GPS densities. 
	
	Naively, one may drop the timestamps
	and apply the kernel density estimation, which leads to 
	\begin{equation}
		\hat f_{naive}(x) = \frac{1}{N_n h^2}\sum_{i=1}^n \sum_{j=1}^{m_i} K\left(\frac{X_{i,j}-x}{h}\right),
		\label{eq::naive}
	\end{equation}
	where $h>0$ is the smoothing bandwidth
	and $K(\cdot)$ is the kernel function such as a Gaussian
	and $N_n = \sum_{i=1}^n m_i$ is the total number of observations.
	One can easily see that $\hat f_{naive}(x)$
	is a consistent estimator of $f_X$ (marginal density of $X$), not $f_{GPS}$.
	So this estimator is in general inconsistent.

	\subsection{Time-weighted estimator}
	
	To estimate $f_{GPS}$, we consider a very simple estimator by weighting 
	each observation by its time difference. 
	For observation $X_{i,j}, t_{i,j}$, 
	we use the two mid-points to the before and after timestamps: $\frac{t_{i,j-1}+t_{i,j}}{2}$ and $\frac{t_{i,j}+t_{i,j+1}}{2}$. 
	The time difference between these two mid-points is 
	$
	W_{i,j} = \frac{t_{i,j+1} - t_{i,j-1}}{2},
	$
	which will be used as the weight of observation $X_{i,j}$. 
	
	For the first time $t_{i,1}$ and the last time $t_{i,m_i}$,
	we use the fact that the time at $t=0$ and $t=1$ are the same.
	So we set $t_{i,0} = t_{i, m_i}-1$ as if the last timestamp occurs before the first timestamp
	and $t_{i,m_i+1} = 1+t_{i,1}$ as if the first timestamp occurs after the last timestamp. 
	This also leads to the summation $\sum_{j=1}^{m_i} W_{i,j} = 1$. 
	%
	%
	We then define the time-weighted estimator as 
	\begin{equation}
		\begin{aligned}
			\hat f_{w}(x) &= \frac{1}{nh^2}\sum_{i=1}^n \sum_{j=1}^{m_i}  W_{i,j} \cdot K\left(\frac{X_{i,j}-x}{h}\right).
		\end{aligned}
		\label{eq::naive}
	\end{equation}
	This is essentially a weighted 2D KDE with the weight $W_{i,j}$
	representing the amount of time around $t_{i,j}$.
	While this estimator is just simply weighting each observation by the time,
	it is a consistent estimator of $f_{GPS}$ (Theorem~\ref{thm::TW}). 
	

	If we want to estimate the interval-specific GPS density,
	we can restrict observations only to the those within the time interval of interest
	and adjust the weights accordingly.

	\subsection{Conditional GPS density estimator}
	
	To estimate the conditional GPS density, we 
	can simply use the conditional kernel density estimator (KDE):
	\begin{equation}
		\hat f(x|t) = \frac{\hat f(x,t)}{\hat f(t)} = \frac{\sum_{i=1}^n\frac{1}{m_i} \sum_{j=1}^{m_i} K\left(\frac{X_{i,j}-x}{h_X}\right)K_T\left(\frac{d_T(t_{i,j},t)}{h_T}\right)}{
			h_X^2 \cdot \sum_{i'=1}^n \frac{1}{m_{i'}}\sum_{j'=1}^{m_{i'}} K_T\left(\frac{d_T(t_{i',j'},t)}{h_T}\right)},
		\label{eq::ckde}
	\end{equation}
	where $h_X>0$ and $h_T>0$ are the smoothing bandwidth and $K_T$ is a kernel function for the time
	and 
	$$
	d_{T}(t_1,t_2) = \min\left\{|t_1-t_2|, 1-|t_1-t_2|\right\}
	$$
	measures the time difference correcting for the fact that the time $t=0$ and $t=1$
	are the same time.
	%
	%
	The time interval $[0,1]$ is not just simply an interval but its two ends $t=0$ and $t=1$ are connected,
	making it a cyclic structure. 
	Thus, $d_T$ measures the distance in terms of time.
	Alternatively, one may map the time $[0,1]$ into a ``degree" in $[0, 2\pi]$ as if it is on a clock
	and use a directional kernel to preserve the cyclic structure of the time \citep{nguyen2023adaptive, mardia2009directional,meilan2021nonparametric}.

	While this estimator seem to be similar to the naive estimator $\hat f_{naive}$, 
	it turns out that this estimator is consistent for $f_{GPS}(x|t)$. 
	Moreover, the conditional KDE can be used to estimate the average GPS density 
	by integrating over $t$:
	\begin{equation}
		\begin{aligned}
			\hat f_c(x) = \int_0^1 \hat f(x|t) dt &= \frac{1}{nh^2}\sum_{i=1}^n \sum_{j=1}^{m_i}  \tilde{W}_{i,j} \cdot K\left(\frac{X_{i,j}-x}{h}\right),\\
			\tilde{W}_{i,j} &=n\cdot  \int_0^1\frac{K_T\left(\frac{d_T(t_{i,j},t)}{h_T}\right)}{\sum_{i'=1}^n\frac{1}{m_{i'}}\sum_{j'=1}^{m_{i'}}K_T\left(\frac{d_T(t_{i',j'},t)}{h_T}\right)}dt.
		\end{aligned}
		\label{eq::int_ckde}
	\end{equation}
	Similar to $\hat f_w$, the estimator $\hat f_c$ is also 
	a weighted 2D KDE.
	This integrated conditional KDE $\hat f_c$ is a consistent estimator of $f_{GPS}$.
	Note that we can convert $\hat f(x|t)$ to
	an interval-specific GPS density estimator by
	integrating over the interval of interest. 
	In Appendix \ref{sec::ckde::comp},
	we discuss some numerical techniques for quickly computing the estimator $\hat f_c$.

	{\bf Remark (daily average estimator).}
	Here is an alternative estimator for the conditional GPS by averaging
	daily conditional GPS density estimators.
	Let 
	$$
	\tilde f_i(x|t) =  \frac{ \sum_{j=1}^{m_i} K\left(\frac{X_{i,j}-x}{h_X}\right)K\left(\frac{d_T(t_{i,j},t)}{h_T}\right)}{
		h^2 \cdot \sum_{j'=1}^{m_{i'}} K\left(\frac{d_T(t_{i',j'},t)}{h_T}\right)}
	$$
	be a conditional KDE for day $i$.
	One can see that when $m_i\rightarrow\infty$ and $h_X,h_T$ are properly shrinking to $0$, 
	this estimator will be a consistent estimator of conditional GPS density at day $i$.
	We can then simply average over all days to get an estimator of $f_{GPS}(x|t):$
	$
	\tilde f(x|t) = \frac{1}{n}\sum_{i=1}^n \tilde f_i(x|t).
	$
	This estimator is also consistent 
	but we experiments show that its performance is not as good as $\hat f(x|t)$.

	\section{Asymptotic theory}	\label{sec::asymp}
	
	In this section, we study the asymptotic theory of these estimators. 
	Recall that our data consist of the collections
	$
	(X_{i,j}, t_{i,j})
	$
	for $i=1,\ldots, n$ and $j=1,\ldots, m_i$
	such that 
	$
	X_{i,j} = S_i(t_{i,j})+\epsilon_{i,j},
	$
	where $S_1,\ldots, S_n \sim F_S$
	are IID random trajectories and $\epsilon_{i,j}$
	are IID mean $0$ noises.
	We assume a fixed design scenario where the timestamps
	$\{t_{i,j}\}$
	are non-random. All relevant assumptions are given in Appendix \ref{sec::assum}.

	For simplicity, we assume that the number of observation of each day $m_i$
	is a fixed and non-random quantity.
	Since random trajectories are IID across different days, 
	the observations in day $i_1$ are
	independent to the observations in day $i_2$.

	{\bf Special case: even-spacing design.}
	When comparing our results to a conventional nonparametric estimation rate,
	we will consider an even-spacing design such that $m_i = m$ 
	and
	$
	t_{i,j} = \frac{2j-1}{2m}.
	$
	Namely, all timestamps are evenly spacing on the interval $[0,1]$. Under this scenario we will re-write our convergence rate in terms
	of $n,m,h$, making it easy to compare to a conventional rate.
	The analysis under the even-spacing design allows us to 
	investigate the effect of smoothing bandwidth in
	a clearer way.

	\subsection{Time-weighted estimator}
	
	We first derive the convergence rate of the time-weighted estimator $\hat f_w(x)$.

	\begin{theorem}
		\label{thm::TW}
		Under Assumptions \ref{ass:K-basic}, \ref{ass:XR-2-derivative}, and \ref{ass:measurement deri bounded}. For any $x\in \R^2$, we have
		\begin{align*}
			\hat f_{w}(x)-f_{GPS}(x) =& O(h^2)+  O\left(\frac{1}{nh^3}\sum_{i=1}^{n}\sum_{j=0}^{m_i}(t_{i,j+1}-t_{i,j})^2\right)+ O_P\left(\frac{1}{nh}
			\sqrt{\sum_{i=1}^{n}\sum_{j=1}^{m_i} W_{i,j}^2}\right)+O_P\left(n^{-\frac{1}{2}}\right)
		\end{align*} 
		as $h\rightarrow 0, \max_{i,j}W_{i,j}\rightarrow0$ and $n\rightarrow \infty$.
	\end{theorem}
	Theorem \ref{thm::TW} describes the convergence rate of the time-weighted estimator. 
	The first term $O(h^2)$ is the conventional smoothing bias.
	The second term $O\left(\frac{1}{nh^3}\sum_{i=1}^{n}\sum_{j=1}^{m_i}(t_{i,j+1}-t_{i,j})^2\right)$
	is a bias due to the timestamps differences.
	The third quantity 
	$O_p\left(\frac{1}{nh}\sqrt{\sum_{j=1}^{n}\sum_{j=1}^{m_i} W_{i,j}^2}\right)$
	is a stochastic errors due to the fluctuations in weights $W_{i,j} = \frac{t_{i,j+1} - t_{i,j-1}}{2}$.
	The last quantity $O_p\left(n^{-\frac{1}{2}}\right)$
	is the usual convergence rate due to averaging over $n$ independent day's data.

	While Theorem~\ref{thm::TW} shows a pointwise convergence rate,
	it also implies the rate for mean integrated square error (MISE): 
	\begin{equation}
		\begin{aligned}
			\int\E\left[\left(    \hat f_{w}(x)-f_{GPS}(x)\right)^2\right]dx& = O(h^4)+  O\left(\frac{1}{n^2h^6}\left(\sum_{i=1}^{n}\sum_{j=1}^{m_i}(t_{i,j+1}-t_{i,j})^2\right)^2\right)\\
			&+     O\left(\frac{1}{n^2h^2}{\sum_{i=1}^{n}\sum_{j=1}^{m_i} W_{i,j}^2}\right)+O\left(n^{-1}\right)
		\end{aligned}
	\end{equation}
	when the support of $X$ is a compact set.
	We show the connection between this convergence rate and the rate of usual nonparametric estimators
	in the even-spacing design scenario.

	\subsubsection{Even-spacing design and bandwidth selection} 
	To make a meaningful comparison to conventional rate
	of a KDE, we consider the even-spacing design with $m_i = m$. 
	In this case, 
	$$
	t_{i,j+1}-t_{i,j} = \frac{1}{m},\qquad W_{i,j} = \frac{1}{m}.
	$$
	Thus,
	\begin{equation}
		\begin{aligned}
			\hat f_{w}(x)-f_{GPS}(x) &= O(h^2) + O\left(\frac{1}{mh^3}\right)  + O_P\left(\sqrt{\frac{1}{nmh^2}}\right)+O_p\left(n^{-\frac{1}{2}}\right)\\
			&= O(h^2) + O\left(\frac{1}{mh^3}\right) + O_P\left(\sqrt{\frac{1}{nmh^2}}\right).
		\end{aligned}
		\label{eq::TW::rate}
	\end{equation}
	In the case of MISE, 
	equation \eqref{eq::TW::rate}
	becomes
	\begin{equation}
		\int\E\left[\left(    \hat f_{w}(x)-f_{GPS}(x)\right)^2\right]dx = 
		O(h^4) + O\left(\frac{1}{m^2h^6}\right) + O\left(\frac{1}{nmh^2}\right).
		\label{eq::TW::mise}
	\end{equation}

	This is a reasonable rate.
	The first quantity $O(h^4)$
	is the conventional smoothing bias in a KDE.
	The second quantity $O\left(\frac{1}{m^2h^6}\right)$
	is the bias due to the `resolution' of each day's observation. 
	The last term $O\left(\frac{1}{nmh^2}\right)$
	is the variance (stochastic error) of a 2D KDE when the effective sample size is $nm$.
	In our estimator, we are averaging over $n$ days and each day has $m$ observations,
	so the total number of random elements we are averaging is $n m$, which is the effective sample size here.

	Choosing the optimal bandwidth is a non-trivial problem
	because $h=h_{n,m}$ and there are three terms in equation \eqref{eq::TW::mise}.
	We denote
	\begin{align*}
		(B) = O\left(\frac{1}{m^2h^6}\right),\qquad
		(S) = O\left(\frac{1}{nmh^2}\right),
	\end{align*}
	so that the term (B) represents the temporal resolution bias
	and the term (S) represents the variance.

	\begin{proposition}
		\label{prop::TW::bw}
		The optimal smoothing bandwidth for equation \eqref{eq::TW::mise}
		is 
		$$
		h^*_{n,m} \asymp \begin{cases}
			m^{-1/5},\qquad &\mbox{if  } m^{1/5}=  o(n)\\
			(nm)^{-1/6},\qquad &\mbox{if  } n= o (m^{1/5})
		\end{cases}
		$$
		and the optimal rate under equation \eqref{eq::TW::rate} is:
		$$
		\int\E\left[\left(    \hat f_{w}(x)-f_{GPS}(x)\right)^2\right]dx= \begin{cases}
			O\left(m^{-4/5}\right) ,\qquad &\mbox{if  } m^{1/5}=  o(n)\\
			O\left((nm)^{-2/3}\right),\qquad &\mbox{if  } n= o (m^{1/5})
		\end{cases}
		$$
	\end{proposition}

	Proposition \ref{prop::TW::bw}
	shows that the critical rate to balance between the errors 
	is $n\asymp m^{1/5}$. 
	When $m^{1/5}=o(n)$, 
	the temporal resolution bias (B) is larger than
	the stochastic error (S)
	and we obtain a convergence rate of $O_P\left(m^{-2/5}\right)$,
	which is dominated by  the resolution bias. 
	On the other hand, $n=o(m^{1/5})$
	is the case where (S) dominates (B),
	so the convergence rate is $O_P\left((nm)^{-1/3}\right)$,
	which is the stochastic error.

	\subsection{Conditional GPS density estimator}
	
	In this section, we study the convergence of the conditional GPS density estimator $\hat f(x|t)$.

	\begin{theorem}
		\label{thm::ckde}
		Under Assumptions \ref{ass:K-basic} to \ref{ass:K1,K2}, 
		for any $(x,t)\in \R^2\times [0,1]$, we have
		\begin{align*}
			\hat f(x|t) - f(x|t) = &O(h_X^2)+O\left(\sum_{i=1}^{n}\alpha_i(t)\sum_{j=1}^{m_i} \frac{K\left(\frac{d_T(t_{i,j},t)}{h_T}\right)}{\sum_{j'=1}^{m_i} K_T\left(\frac{d_T(t_{i,j'},t)}{h_T}\right)}d_T(t_{i,j},t)\right)\\
			&+O_p\left( \frac{1}{h_X}\sqrt{\sum_{i=1}^{n} \alpha_i(t)^2\sum_{j=1}^{m_i} \left[\frac{K\left(\frac{d_T(t_{i,j},t)}{h_T}\right)}{\sum_{j'=1}^{m_i} K_T\left(\frac{d_T(t_{i,j'},t)}{h_T}\right)}\right]^2}\right)
			+O_p\left(\sqrt{\sum_{i=1}^{n} \alpha_i(t)^2}\right),
		\end{align*}
		where $\alpha_i(t)$ is the kernel weight of day $i$ for the time $t$:
		$$
		\alpha_i(t) = \frac{\hat p_{T,i}(t)}{ \sum_{i'=1}^n \hat p_{T,i'}(t)} = \frac{\frac{1}{m_i}\sum_{j=1}^{m_i} K\left(\frac{d_T(t_{i,j},t)}{h_T}\right)}{
			\sum_{i'=1}^n \frac{1}{m_{i'}}\sum_{j'=1}^{m_{i'}} K\left(\frac{d_T(t_{i',j'},t)}{h_T}\right)}.
		$$
		
	\end{theorem}
	
	The corresponding MISE rate will be the square of the rate in Theorem \ref{thm::ckde}.
	Each term in Theorem~\ref{thm::ckde}
	has a correspondence to the term in Theorem~\ref{thm::TW}.
	The first term is the smoothing bias.
	The second term is the resolution bias due to the fact that observed $t_{i,j}$
	may be away from the time point of interest $t$. 
	A high level idea on how the linear term $d_T(t_{i,j}, t)$ occurs is that the difference $\|S_i(t_{i,j})-S_i(t)\| =  O(d_T(t_{i,j}, t))$
	when the velocity of the trajectory is bounded from the above.
	The third term is the stochastic error of the estimator. 
	The last term is the stochastic error due to different weights of each day;this quantity corresponds to 
	the term $O_P(n^{-1/2})$.

	\subsubsection{Even-spacing design and bandwidth selection} 
	
	We now investigate the scenario of even-spacing, which provides us insights
	into the convergence rate and bandwidth selection problem.
	When $m_i=m$ is the same and the timestamps are evenly spacing, 
	we immediate obtain
	$\alpha_i(t) = \frac{1}{n}$
	for every $i=1,\ldots, n$.
	
	Moreover, 
	the summation
	$
	\sum_{j'=1}^{m_{i'}} K\left(\frac{d_T(t_{i',j'},t)}{h_T}\right)
	$
	is associated with the KDE of the time $t$:
	$$
	\hat p_{T,i}(t) = \frac{1}{mh_T}\sum_{j'=1}^{m} K\left(\frac{d_T(t_{i',j'},t)}{h_T}\right) = p_{T,i}(t) +\Delta_{i, 1}(t) = 1 + \Delta_{i,1}(t),
	$$
	where 
	$
	\Delta_{i,1}(t) = O(h_T^2) + O\left(\sqrt{\frac{1}{mh_T}}\right)
	$ 
	is the approximation error. 
	This implies that 
	\begin{equation}
		\sum_{j'=1}^{m_{i'}} K\left(\frac{d_T(t_{i',j'},t)}{h_T}\right) = m_i\cdot  h_T ( 1 + \Delta_{i,1}(t)). 
		\label{eq::ckde::2}
	\end{equation}
	Similarly, 
	the summation in the numerator of the second term of equation \eqref{eq::ckde}
	\begin{equation}
		\begin{aligned}
			\sum_{j=1}^{m_i}K\left(\frac{d_T(t_{i,j},t)}{h_T}\right)d_T(t_{i,j},t)
			& =  m_i \cdot \frac{1}{m_i }\sum_{j=1}^{m_i}K\left(\frac{d_T(t_{i,j},t)}{h_T}\right)d_T(t_{i,j},t)\\
			& = m_i \cdot \E\left(K\left(\frac{U-t}{h_T}\right)|U-t|\right) +O\left(\sqrt{\frac{1}{mh_T}}\right)\\
			&= c_K \cdot m_i \cdot h^2_T+\Delta_{i,2}(t),
		\end{aligned}
		\label{eq::ckde::3}
	\end{equation}
	where $U$ is a uniform random variable over $[0,1]$
	and $c_K= \int |t| K(t)dt$ is a constant
	and 
	$
	\Delta_{i,2}(t)= O(h_T^2) + O\left(\sqrt{\frac{1}{mh_T}}\right)
	$
	is another approximation error.
	Note that in the second equality of equation \eqref{eq::ckde::3},
	we replace $d_T(t_{i,j},t)$ with $|t_{i,j}-t|$ because when $h\rightarrow0$,
	only those $d_T(t_{i,j},t)\approx0$ matters and in this case, $d_T(t_{i,j},t) = |t_{i,j}-t|$.

	Now, putting equations \eqref{eq::ckde::2} and \eqref{eq::ckde::3}
	into Theorem~\ref{thm::ckde}, we obtain 
	$$
	\hat f(x|t) - f(x|t) = O(h_X^2) + O(h_T) + O\left(\sqrt{\frac{1}{mh_T}}\right)+ O_P\left(\sqrt{\frac{1}{ nm h_X^2}}\right) + O_P\left(n^{-1/2}\right)
	$$
	and the corresponding MISE rate is 
	\begin{equation}
		\int\E\left[\left(    \hat f(x|t)-f_{GPS}(x|t)\right)^2\right]dx 
		= O(h_X^4) + O(h_T^2) + O\left({\frac{1}{mh_T}}\right)+ O\left({\frac{1}{ nm h_X^2}}\right) + O\left(n^{-1}\right).
		\label{eq::ckde::mise}
	\end{equation}
	The first quantity is the conventional smoothing bias.
	The second bias $O(h_T^2)$ is the resolution bias. 
	In contrast to the time-weighted estimator,
	this resolution bias is in a different form
	and is due to the fact that observations with a high weights
	will be within $O(h_T)$ neighborhood of $t$ in terms of time.
	In the worst case, the actual location on the trajectory can differ
	by the same order, leading to this rate. 
	The third quantity is the bias when we approximate a uniform integral with a dense even-spacing grid. 
	The fourth term is the conventional variance part of the estimator
	and  the last one is just the variance due to $n$ distinct days.

	
	Equation \eqref{eq::ckde::mise} also 
	provides the convergence rate of the estimator $\hat f_c(x)$
	because  $\hat f_c(x) = \int_0^1 \hat f(x|t)dt$
	is a finite integral. 
	So the MISE convergence rate of $\hat f_c$
	is the same as in equation \eqref{eq::ckde::mise}.
	
	The bandwidth selection under equation \eqref{eq::ckde::mise}
	is straightforward.
	The optimal choice will be of the rate
	\begin{equation}
		\begin{aligned}
			h^*_T = m^{-1/3},\qquad
			h^*_X = (nm)^{-1/6},
		\end{aligned}
	\end{equation}
	which corresponds to the convergence rate 
	$$
	\int\E\left[\left(    \hat f(x|t)-f_{GPS}(x|t)\right)^2\right]dx  = O\left((nm)^{-2/3}\right)+O\left(m^{-2/3}\right).
	$$
	One can clearly see that the convergence rate is dominated by the second term $O\left(m^{-1/3}\right)$
	because of the resolution bias $O(h^2_T)$ and the approximation bias $O(\frac{1}{mh_T})$ at each day.
	Since both quantity are not stochastic errors, increasing the number of days $n$ will not affect them.

	Under the optimal smoothing bandwidth, 
	we also have 
	$$
	\int\E\left[\left(    \hat f_c(x)-f_{GPS}(x)\right)^2\right]dx  = O\left((nm)^{-2/3}\right)+O\left(m^{-2/3}\right) = O\left(m^{-2/3}\right).
	$$
	Comparing this to the time-weighted estimator in Proposition \ref{prop::TW::bw}, 
	the two rates are the same when the number of days $n$ is not much larger than the number of observation per day $m$,
	which is the typical case. 
	 When we have $m^{1/5} = o(n) $, 
	$\hat f_w(x)$ has a MISE of rate $O\left(m^{-4/5}\right)$ when $n$ is large,
	which is slightly faster than $\hat f_c(x)$.
	This is because the time-weighted estimator does not require an estimator on the conditional density,
	so there is no smoothing on the time, making the estimator more efficient.

	\section{Simple movement model}
	\label{sec::smm}
	
	All the above analysis can be applied to GPS data recorded from a human being but also from other living creatures whose spatial trajectories are studied in the literature (e.g., elk, whales, turtles, fish). However, human activity patterns are very different from rest because they are generally very regular \citep{hagerstrand-1963,hagerstrand-1970}.
	To tailor our analysis on humans, we study a special class of random trajectories that we call 
	\emph{simple movement model (SMM)}. In spite of its simplicity, it provides a model that generates a trajectory similar
	to an actual trajectory of a person. 
	Our simulation studies are all based on this model.
	
	Let
	$
	\mathcal{A} = \{a_1,\ldots, a_K\} \subset \R^2
	$
	be a collection of locations that an individual visits very often. 
	The set $\mathcal{A}$ includes the so-called anchor locations such as the home, office location
	as well as other places such as restaurants and grocery stores.
	Let $B$ be a random element that takes the form of 
	$$
	B = (A_1,R_1,A_2,R_2,\ldots, A_{k}, R_{k}, A_{k+1})
	$$
	such that $A_1,\ldots, A_{k+1} \in \mathcal{A}$
	are the location where the individual stay stationary. 
	The vector $B$ is called an action vector since it describes the actions being taken in a single day.
	The quantity $R_j$ is a road/path connecting $A_{j} $ to $A_{j+1}$,
	which describes the road that the individual takes to move from location
	$A_{j}$ to $A_{j+1}$.

	\begin{figure}
		\center
		\includegraphics[width=5in]{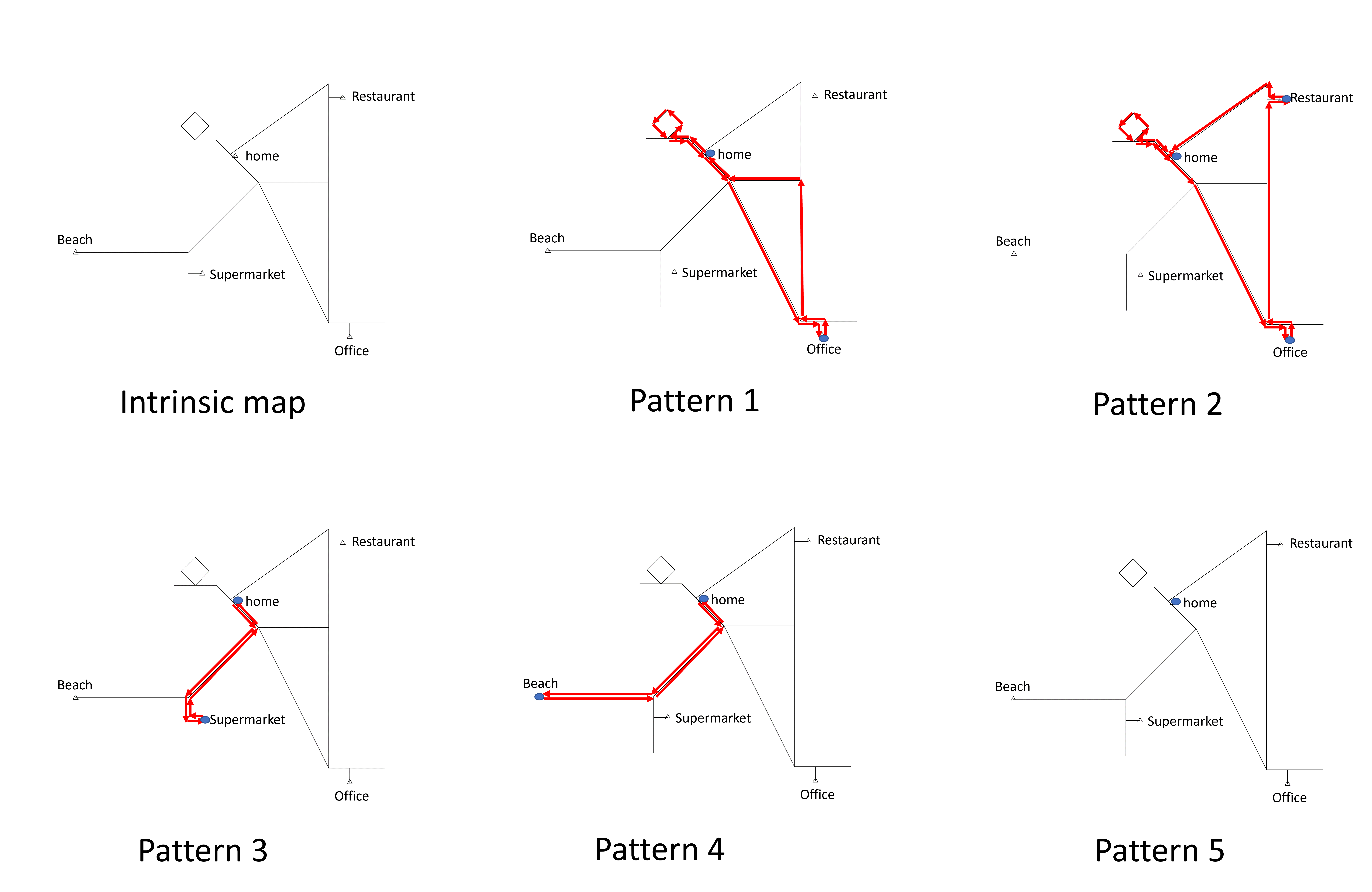}
		\caption{An example of the SMM. The top left panel shows
			the possible places that the individual visits regularly, as well as the road networks connecting these locations. The other five panels show the possible states that
			the action vector can be. More details of the actions
			can be found in Table \ref{Table:pattern_detail} in the appendix.}
		\label{fig::map_pattern}
	\end{figure}

	Figure~\ref{fig::map_pattern}
	provides an example of the SMM.
	In this case, there are five possible patterns that the action vector $B$ can be.
	The first pattern is the following action vector 
	$$
	B = b_1 = \{\texttt{home}, R_{1,1}, \texttt{office}, R_{1,2}, \texttt{home}, R_{1,3}, \texttt{home}\}. 
	$$
	The road $R_{1,1}$ is from home directly to the office and $R_{1,2}$
	is a road from office to home.
	The road $R_{1,3}$ is where the individual goes to the park in the top region
	and then goes home without staying.
	Pattern 3 is the simplest one as it involves only a trip to the local grocery store:
	$$
	B = b_3 = \{\texttt{home}, R_{31}, \texttt{supermarket}, R_{3,2}, \texttt{home}\}. 
	$$
	
	
	%

	The randomness of the action vector $B$ reflects that the individual
	may follow different activity patterns in a day.
	Also, even if all the locations $A_j$ are the same, the individual may take a different road to move between them. 
	In this construction, we would assume that all possible outcomes of $B$
	is a finite set, meaning that the the distribution of $B$ 
	can be described by a probability mass function.

	The action vector $B$ only describes the actions (stationary at a location or moving from one to another) of the individual.
	It does not include any information about the time. 
	Thus, we introduce another random variable 
	$Z$ conditioned on $B$ such that 
	$Z = (Z_1,\ldots, Z_{2k+1})$ is a random partition of the interval $[0,1]$
	such that for $j=1,\ldots, k+1$, $Z_{2j-1}$ is the amount of time that the individual spend on $A_j$
	while $Z_{2j}$ is the amount of time that the individual is on the road $R_j$.
	One can think of $Z$ as a (2k+1)-simplex. 
	When the individual is on the road, we assume that the individual moves between locations with constant velocity.  
	One may use the Dirichlet distribution or truncated Gaussian distribution to generate $Z$. 
	
	To illustrate the idea, we consider again the example in Figure~\ref{fig::map_pattern}. 
	When $B=b_1$, there are a total of $7$ elements in the action vector. 
	So $Z|B=b_1$ is a 7-simplex. 
	Suppose $Z = (\frac{9}{24}, \frac{1}{24}, \frac{8}{24}, \frac{1}{24}, \frac{2}{24}, \frac{1}{24}, \frac{2}{24})$.
	This means that the individual leaves home at 9 am and spends 1 hour on the road to office.
	The individual then spends 8 hours in the office and drives home for 1 hour.
	After a 2-hour stay at home, the individual goes to the park and takes a round trip back home for 1-hour
	and stays at home for the rest of the day.
	Note that the vector $Z|B=b_3$ will be a 5-simplex since pattern 3 only has 5 possible actions.

	Once we have randomly generate $B,Z$, we can easily obtain the corresponding random trajectory $S(t)$.
	Since this model consider a simplified movement of an individual, we 
	call it \emph{simple movement model}.
	Although it is a simplified model, it does capture the primary activity of an individual
	that can be captured by the GPS data. 
	Under SMM, the distribution $P_S$ can be characterized by 
	a PMF of $B$ and the PDF of $Z|B$. 
	Since the action vector $B$ only have finite number of possible outcome,
	the implied trajectories can be viewed as generating from
	a mixture model where the action vector $B$
	acts as the indicator of the component.

	\subsection{Recovering anchor locations}
	\label{sec::AL}
	
	In the SMM, 
	we may define the
	the anchor locations as 
	$\mathcal{A}_\lambda = \{a\in\mathcal{A}: P(S(U) = a)>\lambda\}$
	for some given threshold $\lambda$.
	The threshold $\lambda$ is the proportion of the time that the individual spend at the location $a$ in a day. 
	For instance, the threshold $\lambda = \frac{8}{24}$ should give us the home location
	as we expect people to spend more than 8 hours in the home in a day.
	$\lambda = \frac{4}{24}$ may correspond to both the home and the work place.

	Suppose that the measurement noise is isotropic Gaussian with variance $\sigma^2$ that is known to us, 
	Proposition \ref{prop::AP} show that 
	for any $a$ with 
	$P(S(U)=a) \geq \rho_a,$
	we have
	$
	f_{GPS}(a) \geq   \frac{\rho_a}{2\pi\sigma^2}.
	$
	This implies that any anchor location $a \in \mathcal{A}_\lambda$
	will satisfy 
	$
	f_{GPS}(a) \geq   \frac{\lambda}{2\pi\sigma^2}.
	$
	
	Let $\hat f(x)$ be an estimated GPS density, either from the time-weighted approach
	or the conditional GPS method. 
	We may use the estimated upper level set 
	\begin{equation}
		\hat{\mathcal{L}}_\lambda = \left\{x: \hat f(x)\geq \frac{\lambda}{2\pi\sigma^2} \right\}
		\label{eq::level}
	\end{equation}
	as an estimator for possible locations of $\mathcal{A}_\lambda$.
	As a result,
	the level sets of $\hat f(x)$
	provide useful information on the anchor locations.

	In addition to the level sets, the density local modes
	of $\hat f(x)$
	can be used as an estimator of the anchor locations. 
	Under SMM, anchor locations are where there are probability mass
	of $S(U)$. 
	On a regular road $R_i$, there is no probability mass (only a 1D probability density
	along the road).
	This property implies that the GPS density $f_{GPS}(x)$
	will have  local modes around each anchor point (when the measurement noise is not very large).
	Thus, the local modes within  $\hat{\mathcal{L}}_\lambda$
	can be used to recover the anchor location:
	$$
	\hat{\mathcal{M}}_\lambda =\{x\in \hat{\mathcal{L}}_\lambda: \nabla \hat f(x) =0, \,\, \lambda_1(\nabla\nabla \hat f(x)) <0\},
	$$
	where $\lambda_1(M)$ is the largest eigenvalue of the matrix $M$ and $\nabla\nabla f(x)$ is the Hessian matrix of $f(x)$.

	\subsection{Detecting activity spaces}
	\label{sec::AS}
	
	The activity space is a region where the individual spends most of her/his time
	in a regular day.
	This region is unique to every individual
	and can provide key information such as how the built-environment influences the individual \citep{sherman-et-2005}. We show how GPS data can be used to determine an individual's acvity space.

	Suppose we know the true GPS density $f_{GPS(x)}$, 
	a simple approach to quantify the activity space is via
	the level set indexed by the probability \citep{polonik1997minimum},
	also known as the density ranking \citep{10.1214/19-AOAS1311,chen2019generalized}:
	\begin{equation}
		\begin{aligned}
			Q_\rho &= \min_{\lambda}\left\{ \mathcal{L}^*_\lambda: \int_{x\in \mathcal{L}^*_\lambda}f_{GPS}(x)dx \geq \rho \right\}\\
			\mathcal{L}^*_\lambda& = \left\{x:  f_{GPS}(x)\geq\lambda \right\}.
		\end{aligned}
		\label{eq::Q1}
	\end{equation}
	Because the level set $ \mathcal{L}^*_\lambda$ are nested when changing $\lambda$,
	$Q_\rho$ can be recovered by varying $\lambda$ until we have an exactly $\rho$ amount of
	probability within the level set.
	
	The set $Q_\rho$ has a nice interpretation: it is the smallest region where the individual
	spends at least a proportion of $\rho$ time in it. 
	As a concrete example, $Q_{0.9}$ is where the individual has spent $90\%$
	of the time on average. 
	Thus, $Q_\rho$ 
	can be used as statistical entity of the activity space.
	Interestingly, $Q_\rho$
	can also be interpreted as
	a prediction region of the individual.

	\subsubsection{Estimating the activity space}
	Estimating $Q_\rho$
	from the data can be done easily
	by the plug-in estimator approach. 
	A straightforward idea is to use 
	\begin{equation*}
		\begin{aligned}
			\hat Q^*_\rho &= \min_{\lambda}\left\{ \hat{\mathcal{L}}^*_\lambda: \int_{x\in \hat{\mathcal{L}}^*_\lambda}\hat f(x)dx \geq \rho \right\},\qquad
			\hat{ \mathcal{L}}^*_\lambda& = \left\{x: \hat {f}(x)\geq\lambda \right\}.
		\end{aligned}
	\end{equation*}
	While this estimator is valid, it may be expansive to compute it
	because of the integral in the estimator. 
	
	Here is an alternative approach to estimating $Q_\rho$ without the evaluation of the integral. 
	First, the integral $ f_{GPS}(x)dx = dF_{GPS}(x)$, where 
	$
	F_{GPS}(x) = P(X(U)+\epsilon\leq x)
	$
	is the corresponding CDF of $f_{GPS}$. 
	When the timestamps are even-spacing, $F_{GPS}$
	can be estimated by the empirical distribution of $X_{i,j}$. 
	When the timestamps are not even-spacing,
	we can simply weight each observation by the idea of time-weighting estimator. 
	In more details, a time-weighted estimator of $F_{GPS}$ is
	$$
	\hat F_{w} (x) = \frac{1}{n}\sum_{i=1}^n \sum_{j=1}^{m_i} W^\dagger_{i,j} I(X_{i,j}\leq x),
	$$
	where 
	$
	W^\dagger_{i,j} = \frac{t_{i,j+1} - t_{i,j-1}}{2T_i},
	$
	if we are using $\hat f_w$
	and 
	$$
	\qquad W^\dagger_{i,j} = \int_0^1\frac{K_T\left(\frac{d_T(t_{i,j},t)}{h_T}\right)}{\sum_{i'=1}^n\frac{1}{m_{i'}}\sum_{j'=1}^{m_{i'}}K_T\left(\frac{d_T(t_{i',j'},t)}{h_T}\right)}dt
	$$
	if we are using $\hat f_c$.
	$\hat F_{w} (x)$ can be viewed as a time-weighted EDF. 
	A feature of $\hat F_w(x)$ is that 
	when integrating over any region $A$, we have 
	$$
	\int I(x\in A) d\hat F_w(x) = \frac{1}{n}\sum_{i=1}^n\sum_{j=1}^{m_i} W^\dagger_{i,j} I(X_{i,j}\in A).
	$$
	It is the time-weighted proportion of $X_{i,j}$ inside the region $A$. 
	With this, 
	we obtain the estimator
	\begin{equation}
		\begin{aligned}
			\hat Q_\rho &= \min_{\lambda}\left\{ \hat{ \mathcal{L}}^*_\lambda:  \frac{1}{n}\sum_{i=1}^n \sum_{j=1}^{m_i} W^\dagger_{i,j} I(X_{i,j}\in  \hat{ \mathcal{L}}^*_\lambda ) \geq \rho \right\},\\
			\hat{ \mathcal{L}}^*_\lambda& = \left\{x: \hat {f}(x)\geq\lambda \right\};
		\end{aligned}
		\label{eq::Q2}
	\end{equation}
	note that $\hat f$ can be either $\hat f_c$ or $\hat f_w$.
	We provide a fast numerical method for computing $\hat Q_\rho$ in Appendix \ref{sec::num::as}.

	\subsection{Clustering of trajectories}
	
	In the SMM,
	the action vector $B$ is assumed to take a finite number of possible outcomes.
	While the actual time of each movement under this model may vary a bit according to the random vector $Z$,
	resulting in a slightly different trajectory with the same $B$,
	the main structure of the trajectory remain similar. 
	This implies that the individual's trajectories
	can be clustered into groups according to the action vector $B$. 
	
	While we do not directly observe $B$,
	we may use the estimated GPS densities at each day 
	as a noisy representation of $B$
	and cluster different days according to the estimated GPS densities.
	
	Here we describe a simple approach to cluster trajectories.
	Let 
	$$
	\hat f_{c,i}(x) = \frac{1}{nh^2} \sum_{j=1}^{m_i} \tilde{ W}_{i,j} \cdot K\left(\frac{X_{i,j}-x}{h}\right)
	$$
	be the estimated time-weighted GPS density of day $i$,
	where $\tilde{ W}_{i,j}$ is defined in equation \eqref{eq::int_ckde}.
	While we may use the pairwise distance matrix among different days $\int (\hat f_{w,a}(x) - \hat f_{w,b}(x))^2dx$
	as a distance measure, 
	we recommend using a scaled log-density, i.e.,
	\begin{equation}
		D_{a,b} = \int \left[\log (\hat f_{c,a}(x)+\xi)- \log(\hat f_{c,b}(x)+\xi)\right]^2dx,
		\label{eq::clu2}
	\end{equation}
	for some small constant $\xi$ to stabilize the log-density.
	Our empirical studies find that equation \eqref{eq::clu2}  works much better in practice
	because the GPS density is often very skew due to the fact that
	when the individual is not moving, the density is highly concentrated at a location.
	
	With the distance matrix $D$,
	we can perform cluster analysis based on hierarchical cluster or spectral clustering \citep{von2007tutorial}
	to obtain a cluster label for every day.
	We may also detect outlier activity by examining the distance matrix.

	\subsubsection{Recovering dynamics after clustering}

	Once we have done clustering of the data,
	we obtain a label $G_i$ for every day, indicating
	what cluster the day belongs to.
	In SMM, trajectories in the same cluster are likely to be
	from the same action vector $B$. 
	If we focus on the data of a cluster,
	we may perform inference on anchor locations or activity space
	(Sections \ref{sec::AL} and \ref{sec::AS})
	specifically for the corresponding action vector.

	Under SMM,
	trajectories in the same cluster are likely to have the same action vector $B$
	(but different temporal distribution $Z$).
	We will expect that 
	the conditional density on the days from the same clusters 
	to be similar.
	The average behavior of these conditional densities in the same cluster
	provides us useful information about the latent action vector.
	
	Let  $G_i\in\{1,\ldots, M\}$ for every $i=1,\ldots, n$ be the cluster label. 
	For a cluster $g \in \{1,\ldots, M\}$, we 
	compute the cluster conditional GPS density
	\begin{equation}
		\hat f_{g}(x|t) =  \frac{\sum_{i: G_i=g}\frac{1}{m_i} \sum_{j=1}^{m_i} K\left(\frac{X_{i,j}-x}{h_X}\right)K_T\left(\frac{d_T(t_{i,j},t)}{h_T}\right)}{
			h_X^2\sum_{i': G_{i'}=g} \frac{1}{m_{i'}}\sum_{j'=1}^{m_{i'}} K_T\left(\frac{d_T(t_{i',j'},t)}{h_T}\right)}.
		\label{eq::g::cond}
	\end{equation}
	How the conditional distribution $\hat f_{g}(x|t)$ changes over time $t$
	informs us the movement of this individual under this action vector.

	{\bf Estimating the conditional center.}
	Since trajectories in the same cluster are similar,
	we expect the cluster conditional GPS density $\hat f_{g}(x|t)$
	to be uni-modal and the mean of this density will be a good representation of its center.
	Thus, a simple way to track how this distribution moves over time is to study 
	the movement of the mean location, which can be written as
	\begin{equation}
		\hat \mu_{g}(t) =  \frac{\sum_{i: G_i=g}\frac{1}{m_i} \sum_{j=1}^{m_i} X_{i,j}K_T\left(\frac{d_T(t_{i,j},t)}{h_T}\right)}{
			\sum_{i': G_{i'}=g} \frac{1}{m_{i'}}\sum_{j'=1}^{m_{i'}} K_T\left(\frac{d_T(t_{i',j'},t)}{h_T}\right)} \in \R^2.
		\label{eq::g::mu}
	\end{equation}
	Note that this is like a kernel regression with two `outcome variables' (coordinates).
	For two time points $t_1<t_2$, the difference in the center 
	$
	\hat \mu_g(t_2) - \hat \mu_g(t_1)
	$
	can be interpreted as the average movement in the distribution $\hat f_g(x|t)$
	within the time interval $[t_1,t_2]$.
	Thus, the function $\hat \mu_g(t)$
	describes the overall dynamics of the trajectories.
	Note that we may use a local polynomial regression method \citep{fan2018local}
	to estimate the velocity as well.

	\section{Real data analysis}\label{sec:real-data}
	

We analyze GPS data recorded for a single individual from a longitudinal study of exposure to a cognitive-impairing substance. Participants in this study have been recruited from several cities in a major metropolitan area in the US. The  survey data collected comprise one month of GPS tracking with a GPS-enabled smartphone carried by each study participant. The protocol of the study required that the smartphones recorded point locations (latitude, longitude, timestamp) on the spatiotemporal trajectories followed by study participants every minute using a commercially available software. The study participant whose GPS locations are used  to illustrate our proposed methodology has been selected at random from the study cohort which comprises more than 150 individuals with complete survey data. In order to preserve any disclosure of private, sensitive information related to this study participant, the latitude and longitude coordinates have been shifted and scaled. The timestamps have been mapped to the unit interval.

\begin{figure}
	\centering
	\includegraphics[width=0.25\linewidth]{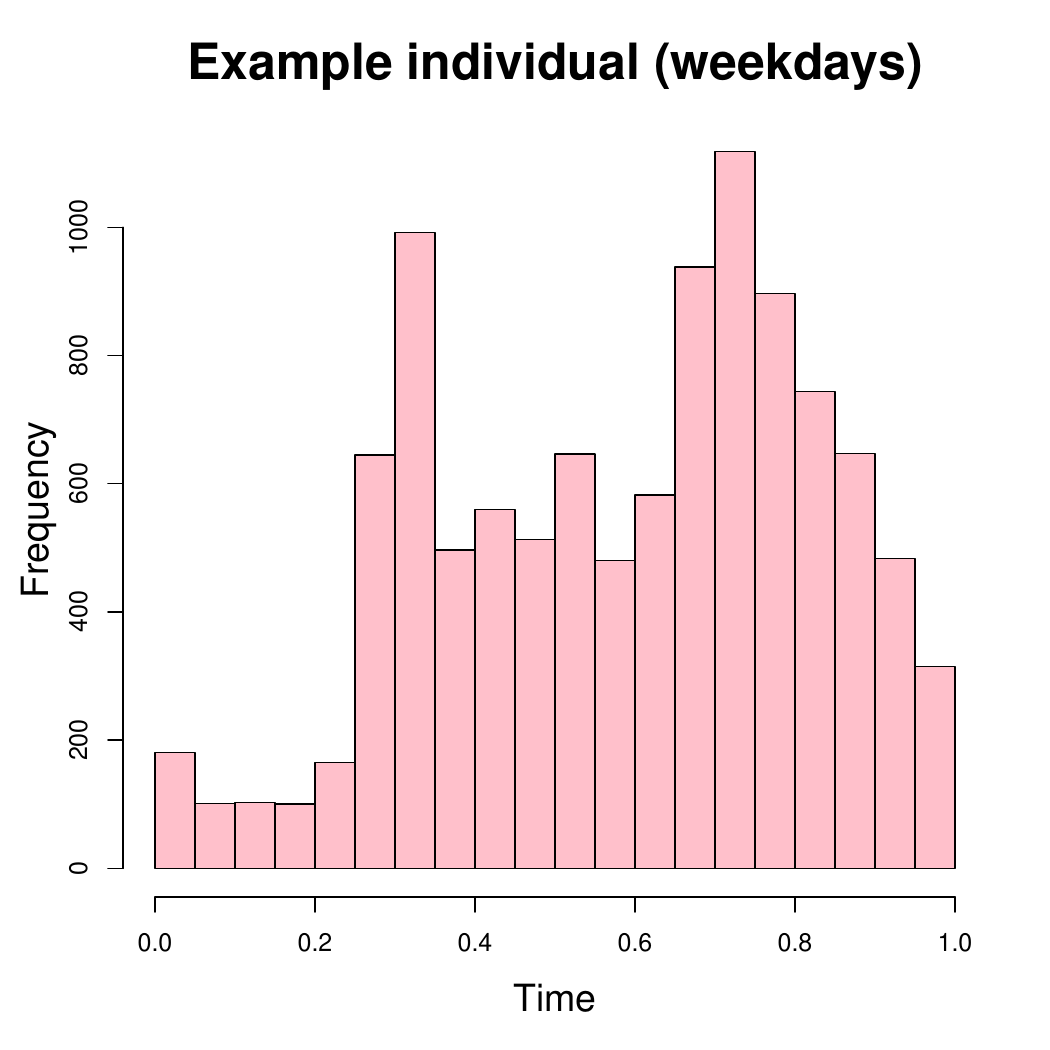}
	\includegraphics[width=0.25\linewidth]{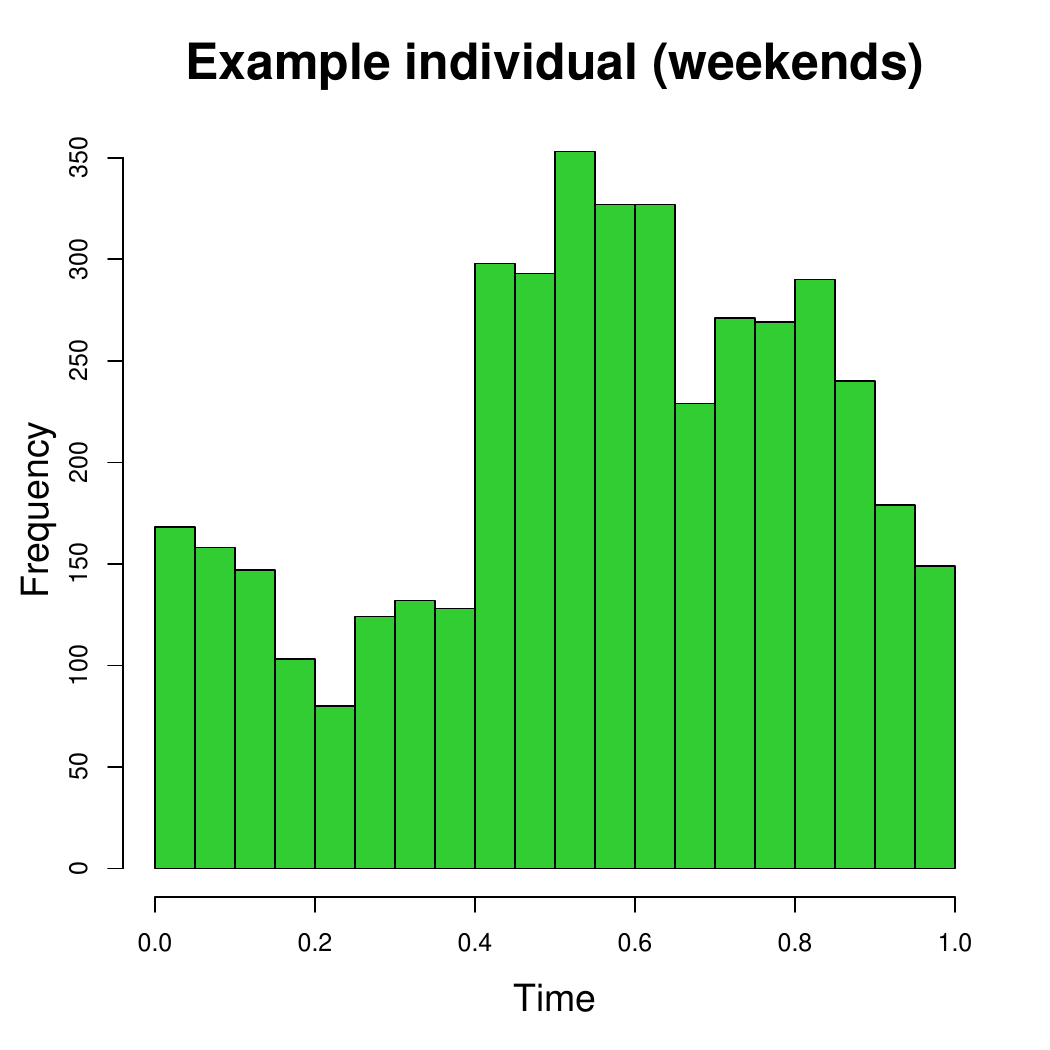}
	\includegraphics[width=0.25\linewidth]{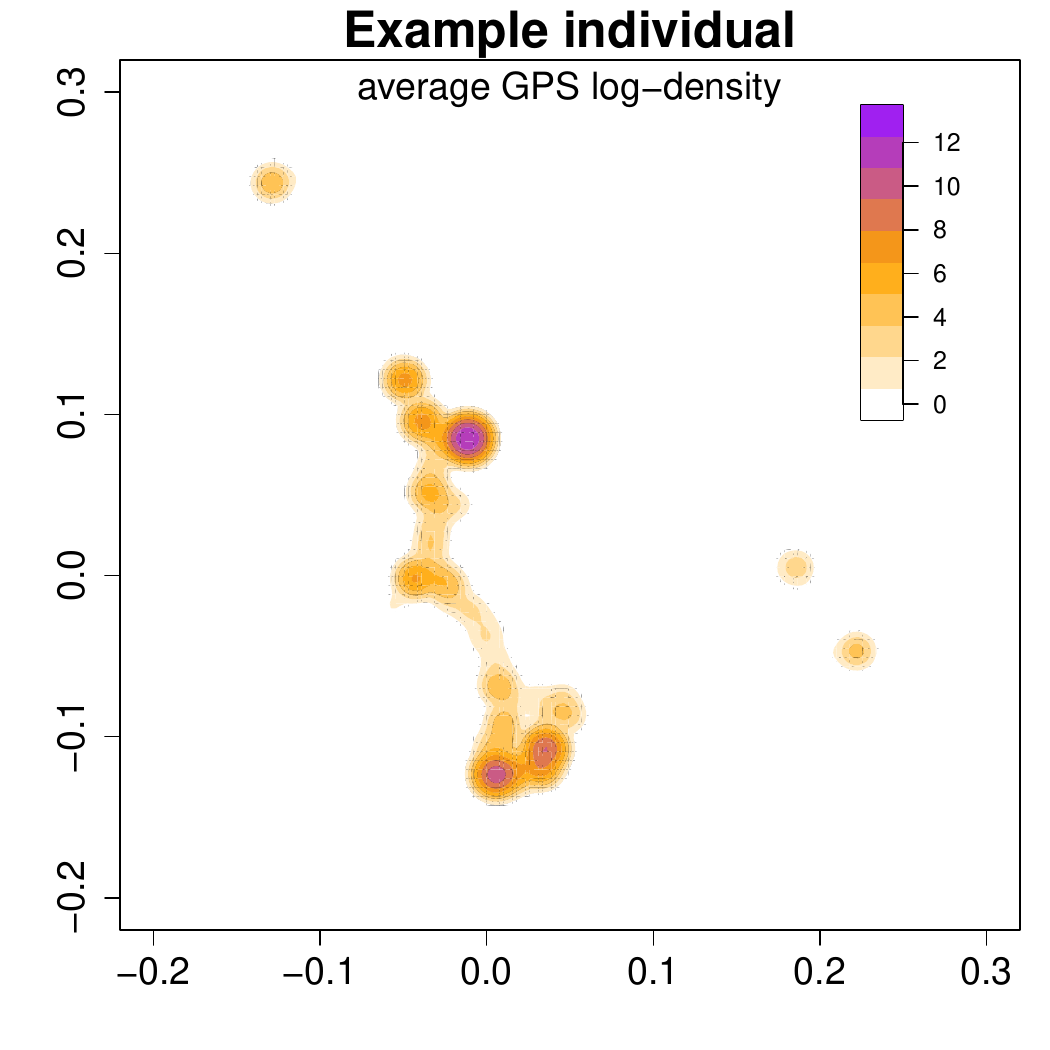}
	\caption{Summary information of example individual. The first two panels display
		the distribution of timestamps in weekdays versus weekends. The right panel display
		the average GPS log-density.}
	\label{fig:p125_summary}
\end{figure}

	\begin{figure}
		\centering
		\includegraphics[width=0.25\linewidth]{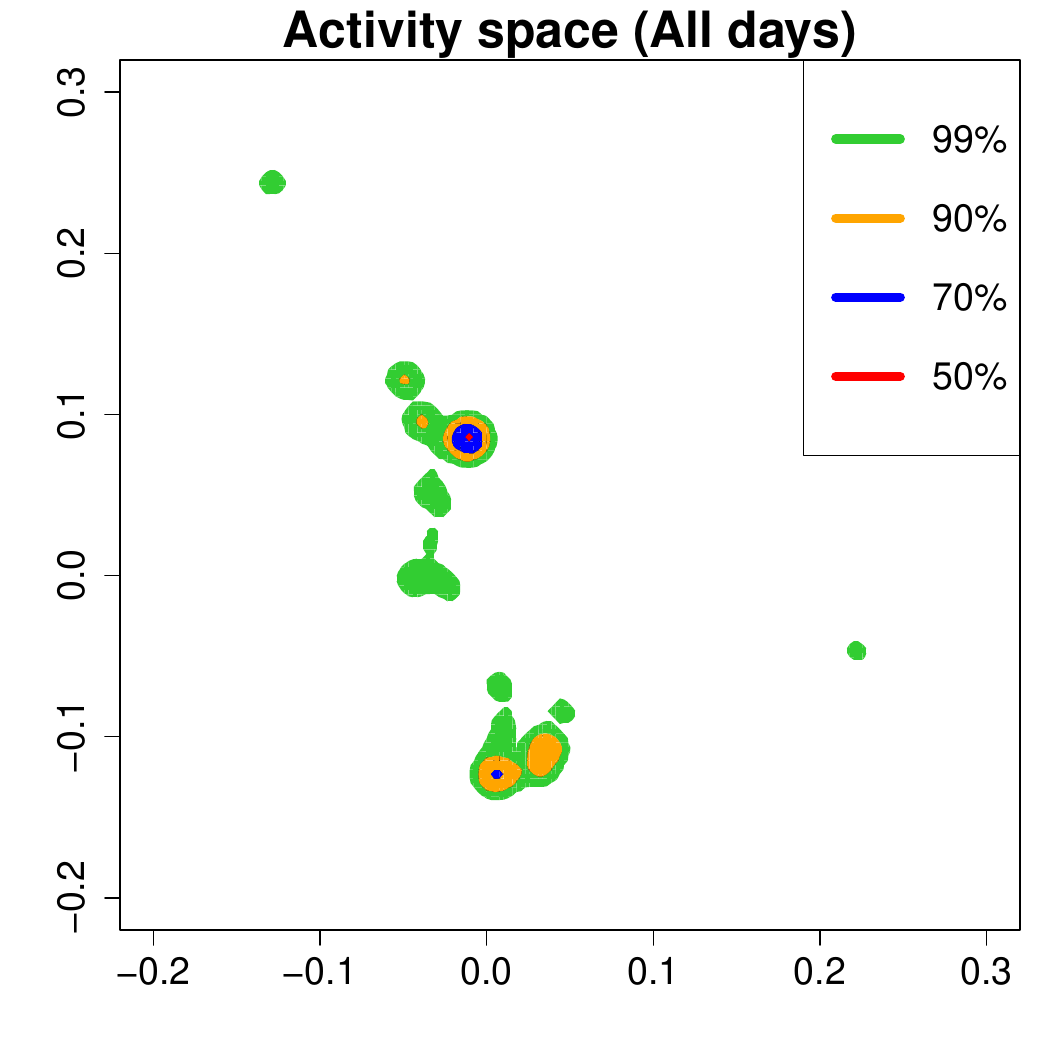}
		\includegraphics[width=0.25\linewidth]{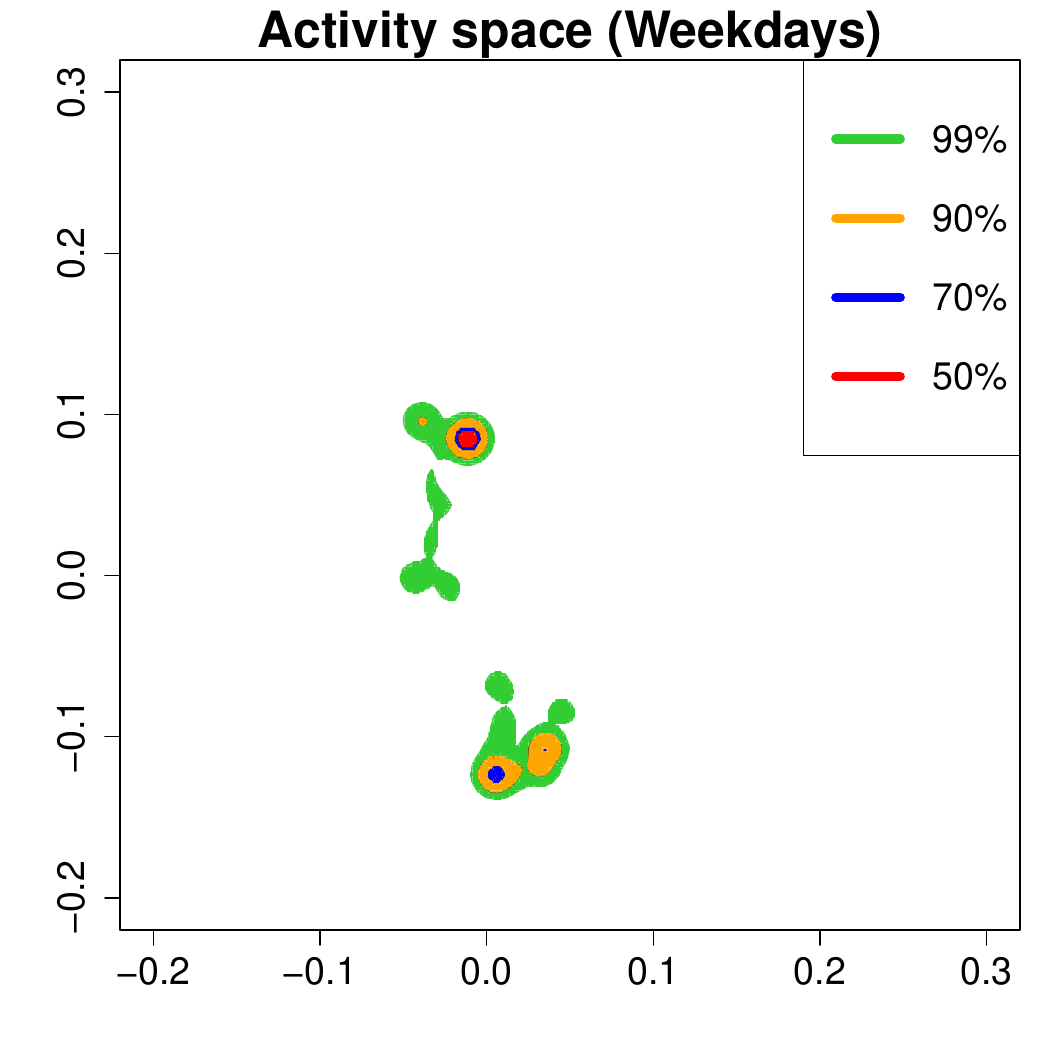}
		\includegraphics[width=0.25\linewidth]{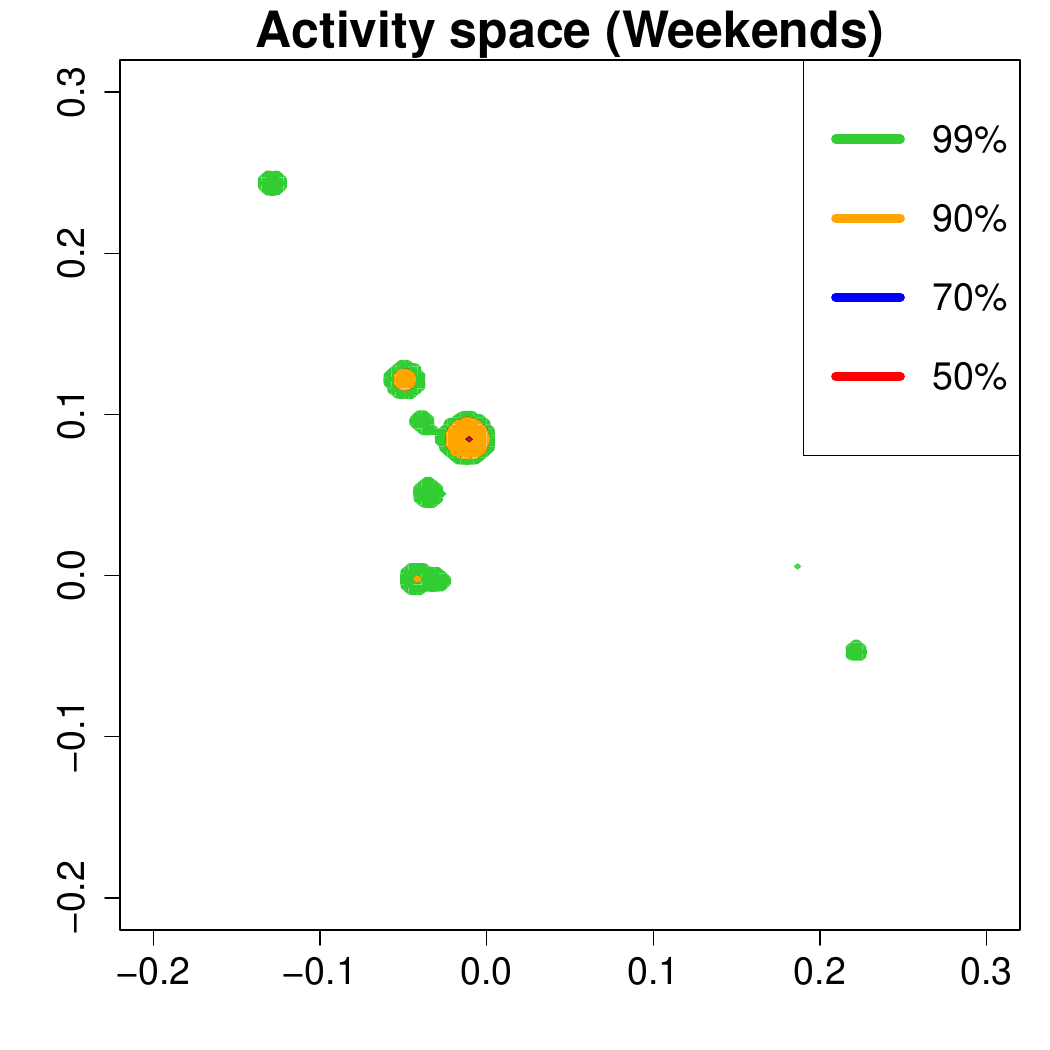}
		\caption{Activity space of the example individual. The first panel shows the activity space across from all days. 
			Second panel uses only the weekday's data and third panel focuses on the weekend. 
			The 90\% region means that this region covers 90\% of the time that the individual spends in a day (on average). }
		\label{fig:p125AS}
	\end{figure}

The GPS data for the selected example individual consists of  a total of 14,972 observations recorded over $n=30$ days, averaging approximately $m=500$  
	observations per day. Although the study protocol specified that a new GPS observation should be recorded every minute, the actual recorded timestamps are not regularly distributed over the entire observation window, as illustrated in the first two panels of  Figure~\ref{fig:p125_summary}. Consequently, a naive estimation approach would be biased, so we should use either
	$\hat f_w$ or $\hat f_c$
	for density estimation. We choose $\hat f_c$ because it allows for conditional density estimation at any given time.
	The smoothing bandwidths are set to $h_X = 0.005$ (longitude/latitude degrees) and $h_T = 0.02$ (approximately 30 minutes). 
	The third panel in Figure~\ref{fig:p125_summary} displays the GPS log-density $\log (\hat f_c(x)+1)$.
	Several high-density regions emerge, suggesting potential anchor locations for the individual.

	{\bf Activity space analysis.}
	Figure~\ref{fig:p125AS} presents the individual's activity space under three scenarios: all days, weekdays, and weekends. The contour region 
	$Q_\rho$  represents the area where the individual spends $\rho$ proportion of their time.
	For instance, 
	$Q_{0.99}$ (green) covers 99\% of their time.
	More than 50\% of the time is spent at home, as it includes nighttime. 
	On weekdays (middle panel), 90\% of the activities (orange contour) are concentrated at home or the office. 
	To further analyze mobility, we separate the data into weekdays and weekends and examine hourly density variations.

	{\bf Analysis on the weekdays patterns.}
	Weekday Patterns
	Figure~\ref{fig:p125} depicts GPS density distributions across four weekday time windows:
	1-2 AM (sleep time),
	7-8 AM (morning rush hour),
	12-1 PM (noon), and
	5-6 PM (evening rush hour).
	The top row shows the logarithm of the hourly GPS density to mitigate skewness. The bottom row presents the activity space contours for 99\% (green), 90\% (orange), 70\% (blue), and 50\% (red) of the time.
From Figure~\ref{fig:p125}, we identify five likely anchor locations:
one main living place (highest density region in first column) and other two close locations the person may also visit at night (two high-density regions at bottom in the first column) and
two office locations (identified in the third panel as high-density regions during midday).
By further cluster analysis, we find that there is a main living location (corresponding to the highest density region at night), and an alternative living place close to a park that the example individual sometimes walks inside. 
The second and fourth columns capture rush-hour mobility, illustrating connectivity between anchor locations.
	
		\begin{figure}
		\centering
		\includegraphics[width=0.2\linewidth]{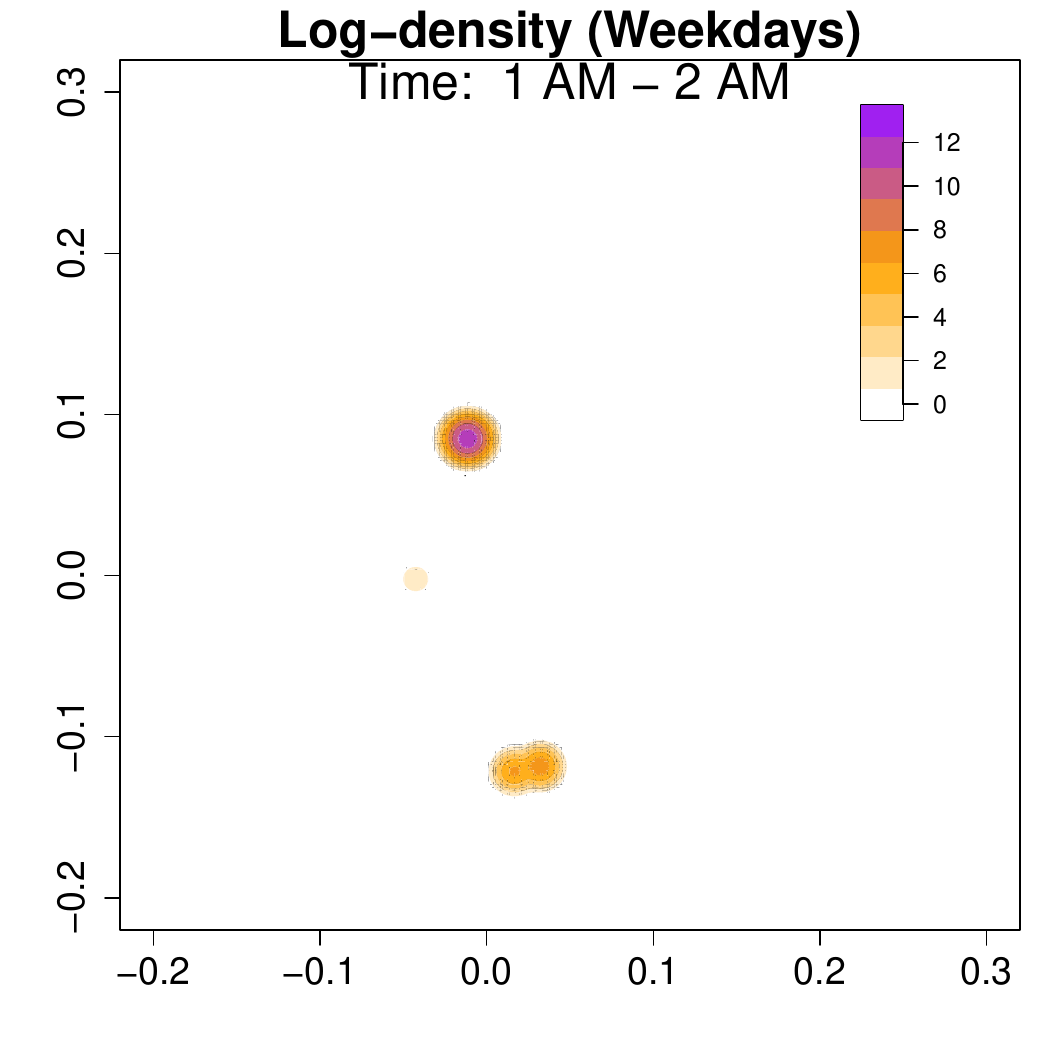}
		\includegraphics[width=0.2\linewidth]{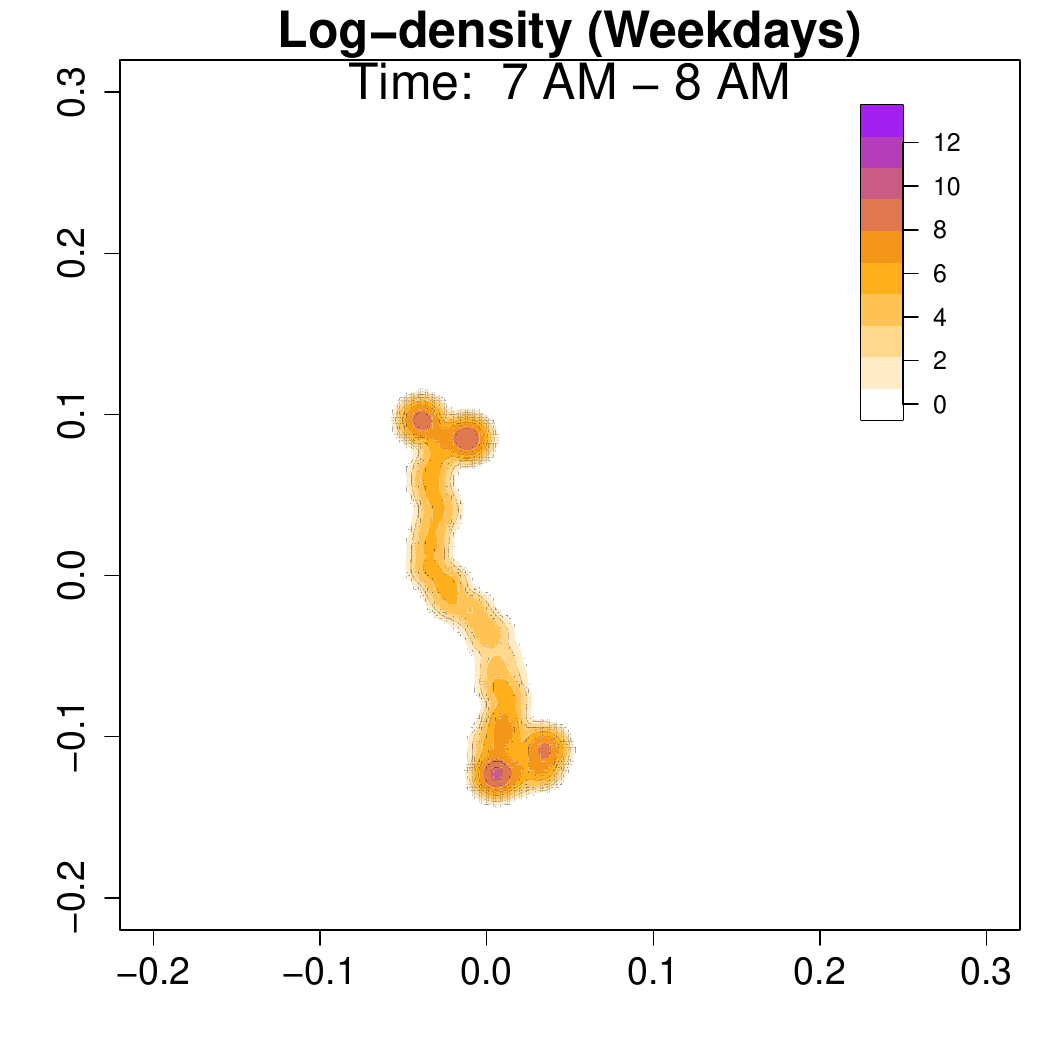}
		\includegraphics[width=0.2\linewidth]{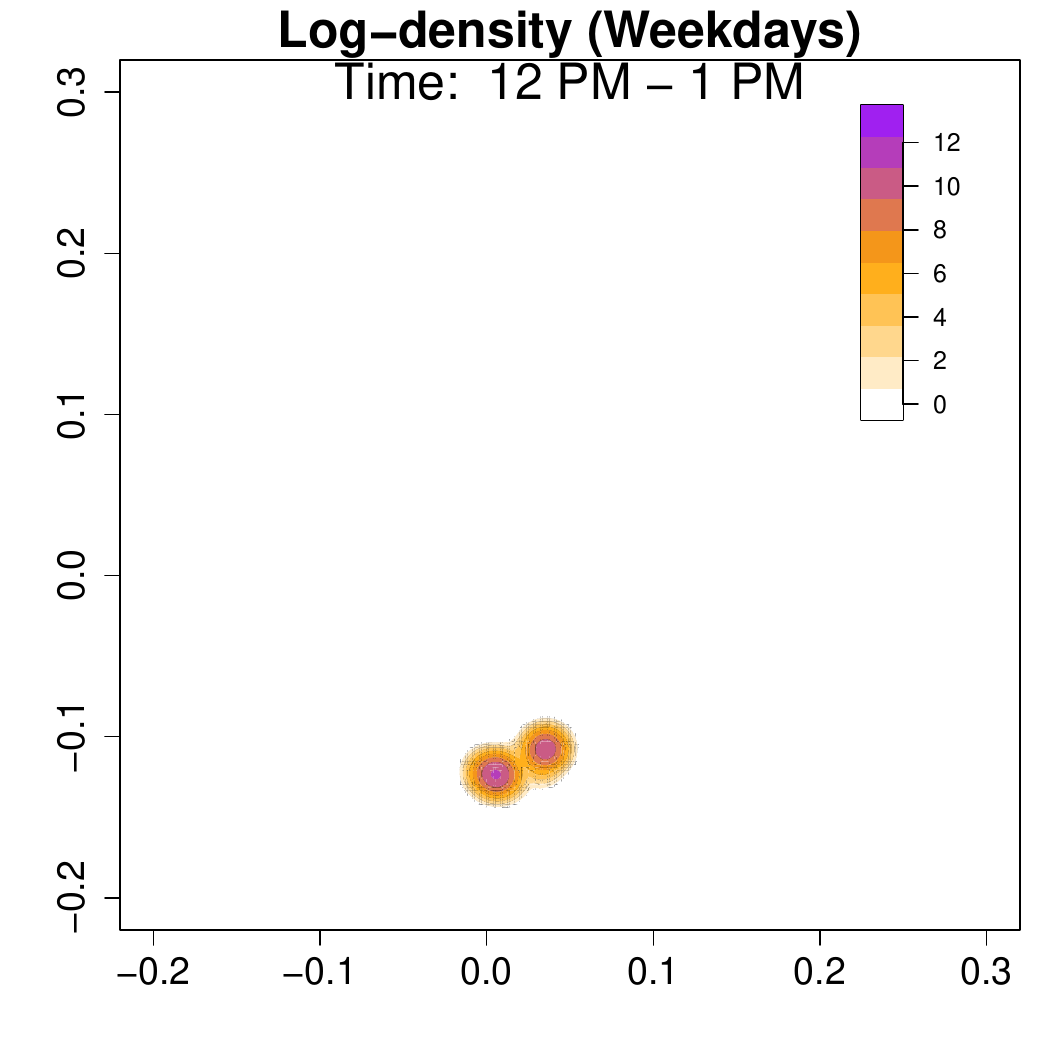}
		\includegraphics[width=0.2\linewidth]{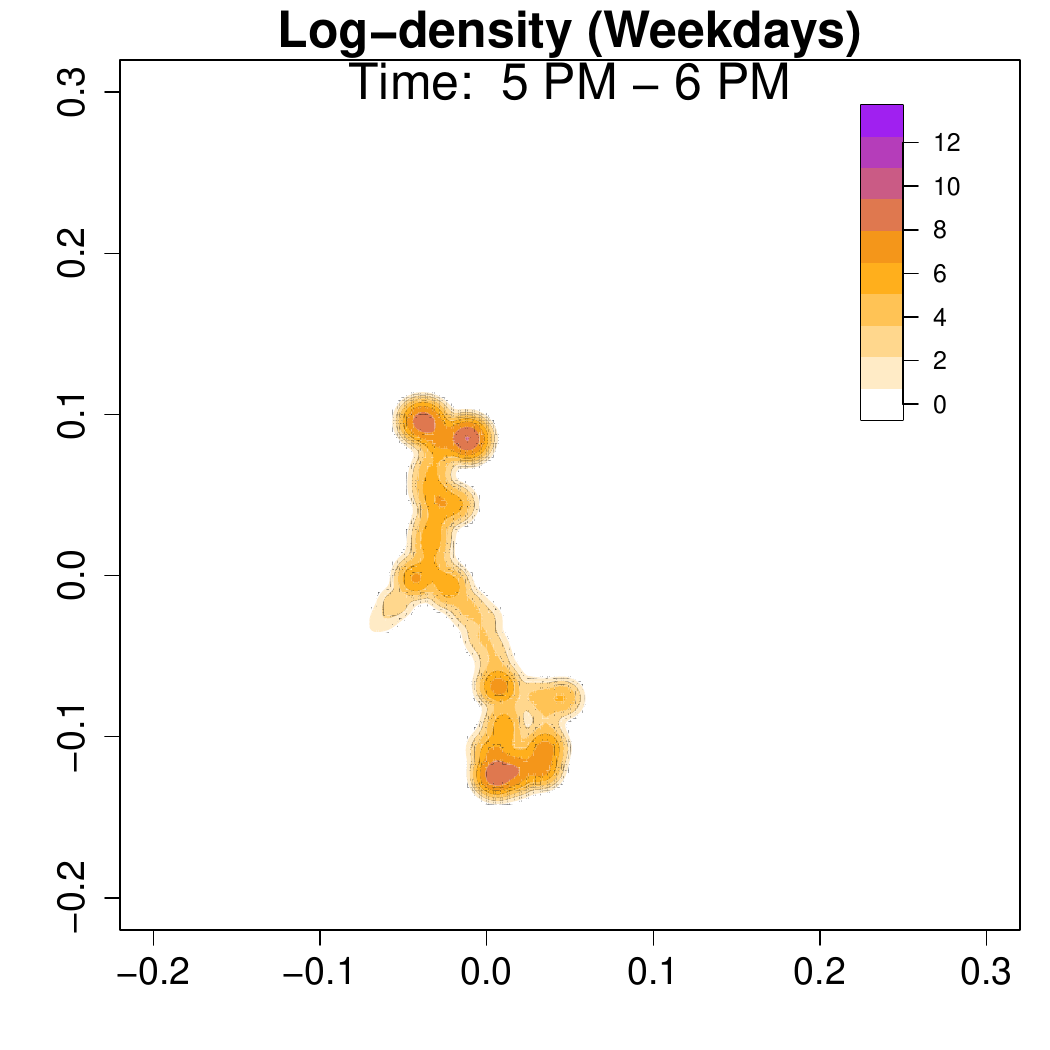}
		\includegraphics[width=0.2\linewidth]{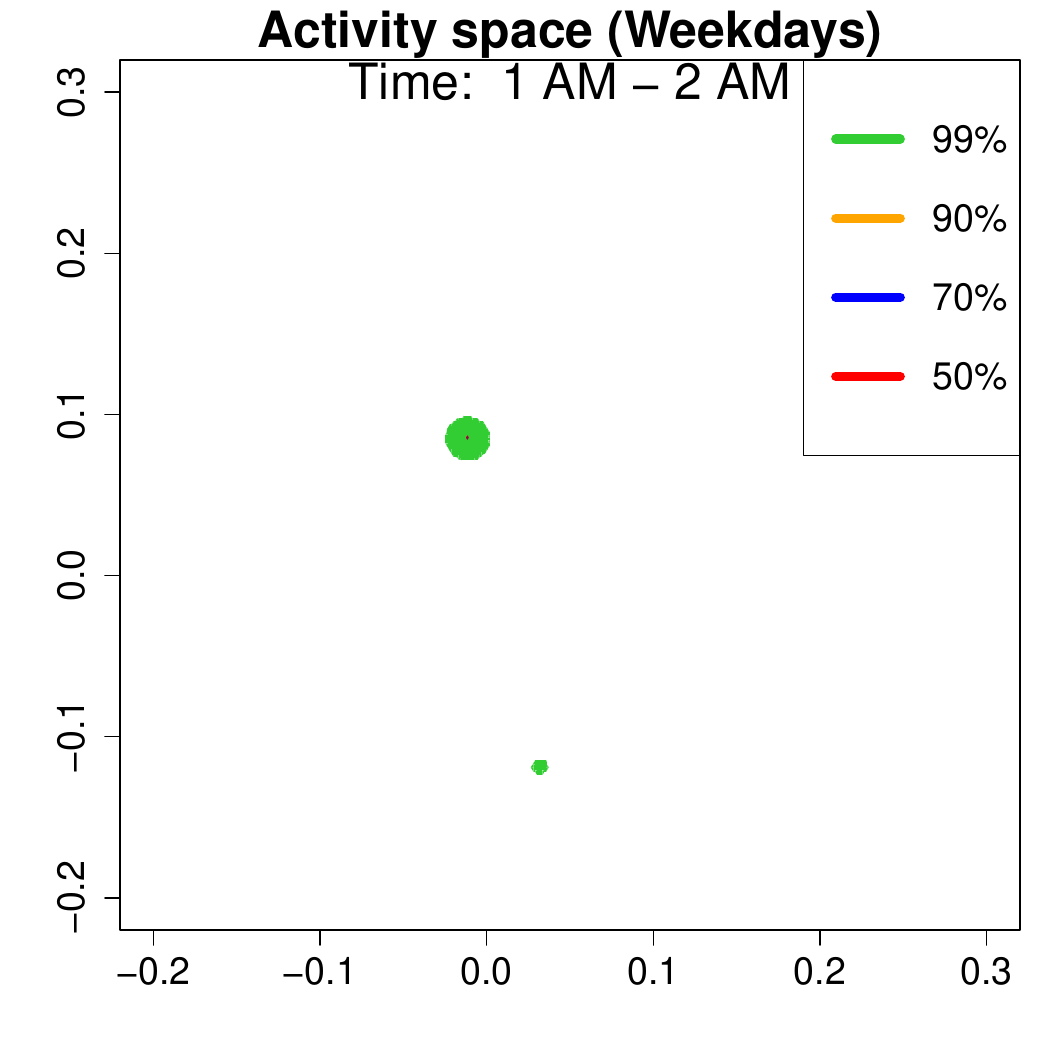}
		\includegraphics[width=0.2\linewidth]{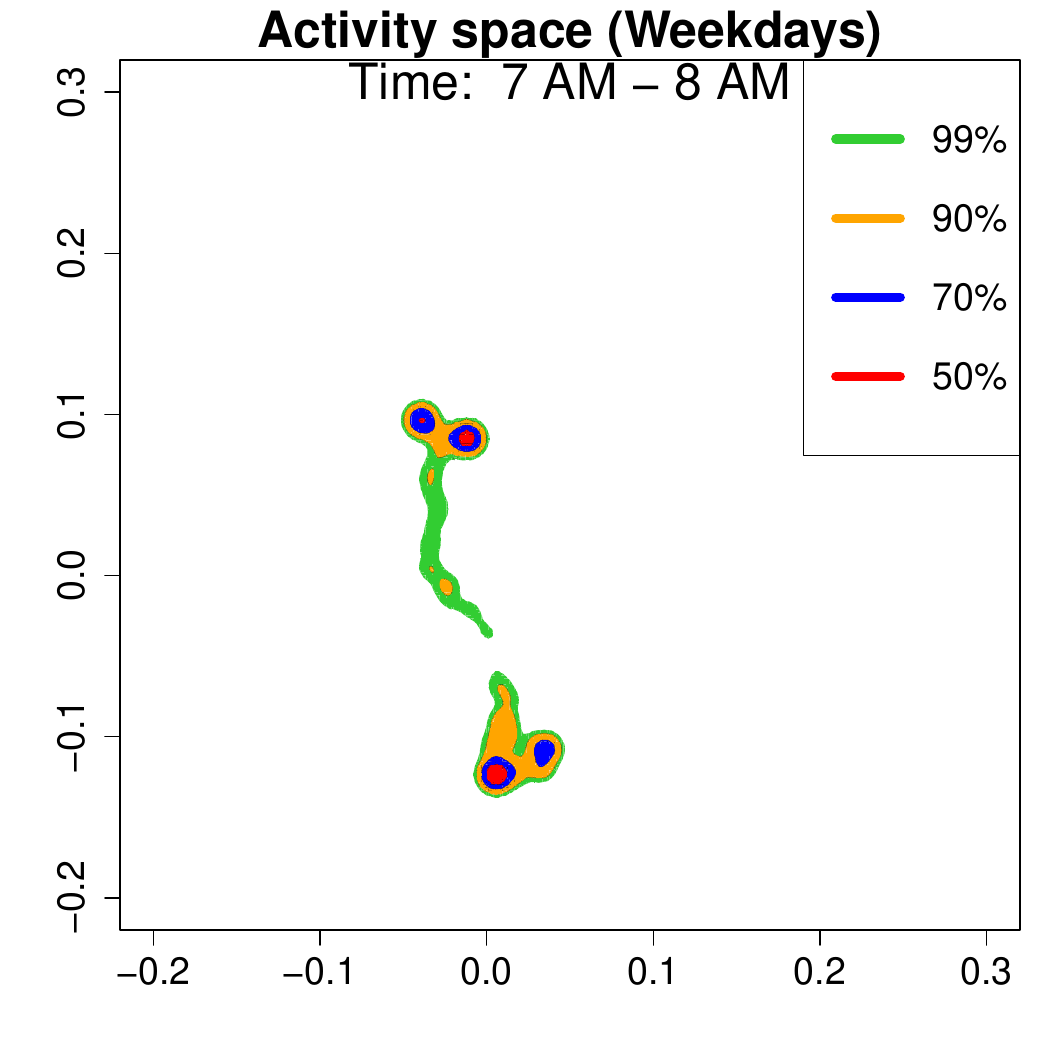}
		\includegraphics[width=0.2\linewidth]{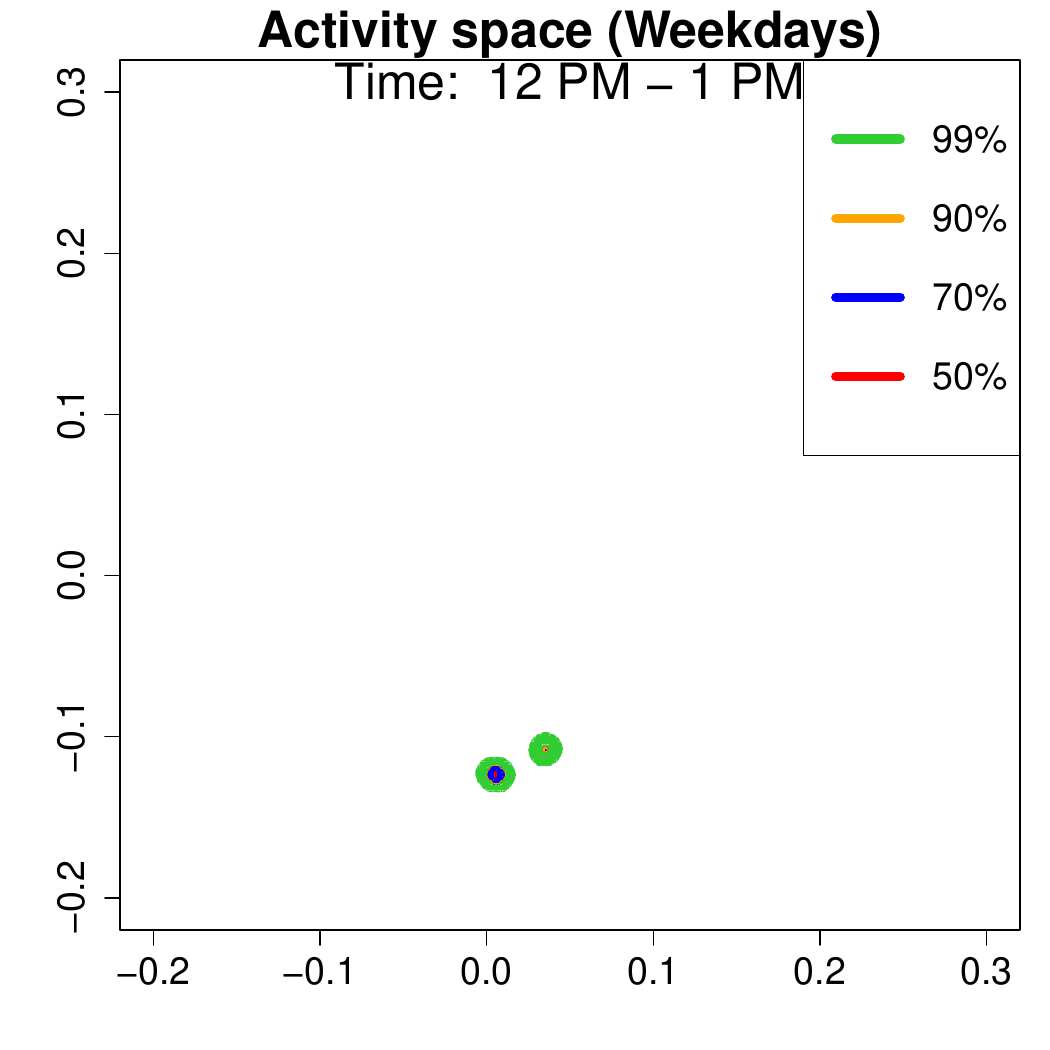}
		\includegraphics[width=0.2\linewidth]{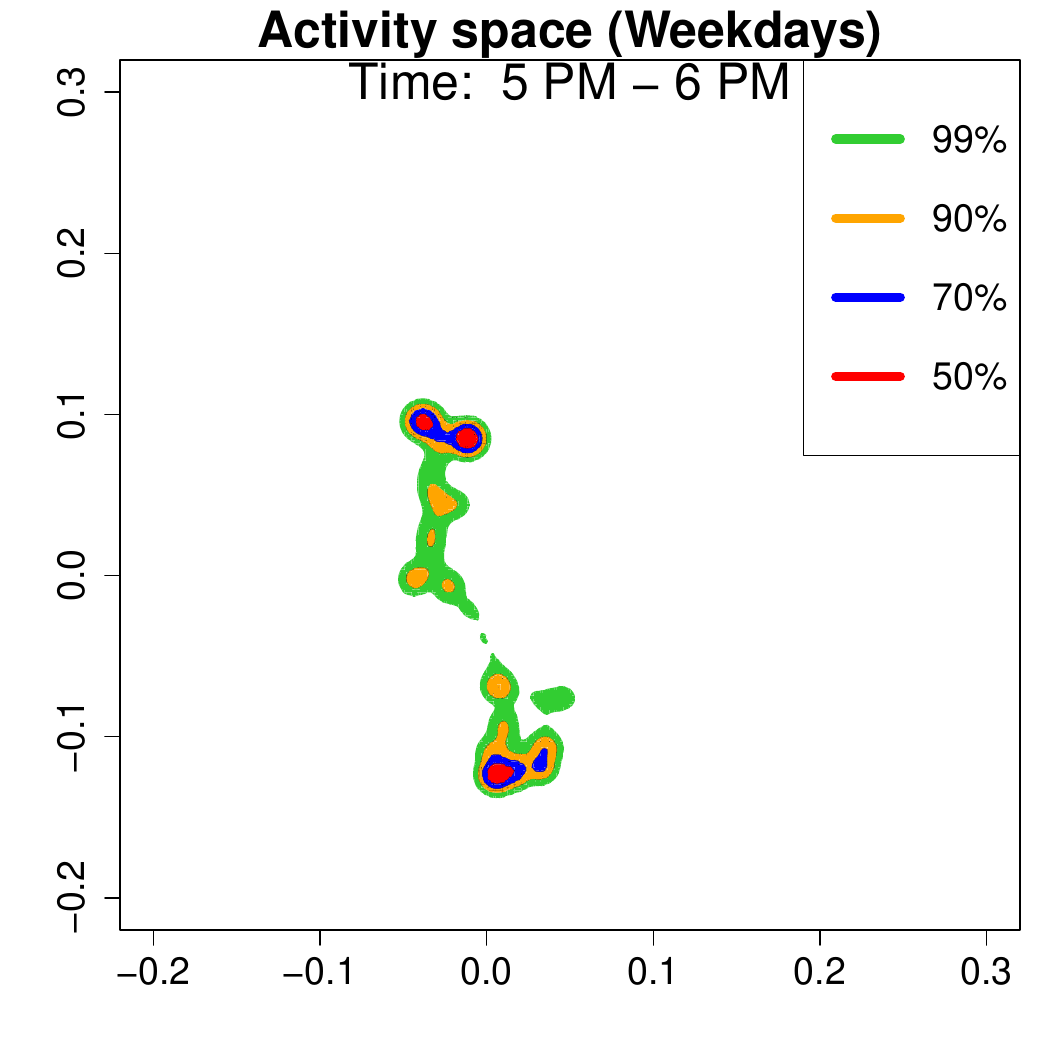}
		\caption{The hourly GPS log-density and activity space of the example individual during the weekdays.
			Top row displays the logarithm of the GPS density and the bottom row shows the corresponding activity space
			of $99\%$ (green), $90\%$ (orange), $70\%$ (blue), and $50\%$ (red) average time. 
			We display the interval-specific GPS density of the individual during 1-2 AM (sleep time), 7-8 AM (morning rush hour), 12-1 PM (workplace/school), 
			and 5-6 PM (evening rush hour).}
		\label{fig:p125}
	\end{figure}
	
	\begin{figure}
		\centering
		\includegraphics[width=0.2\linewidth]{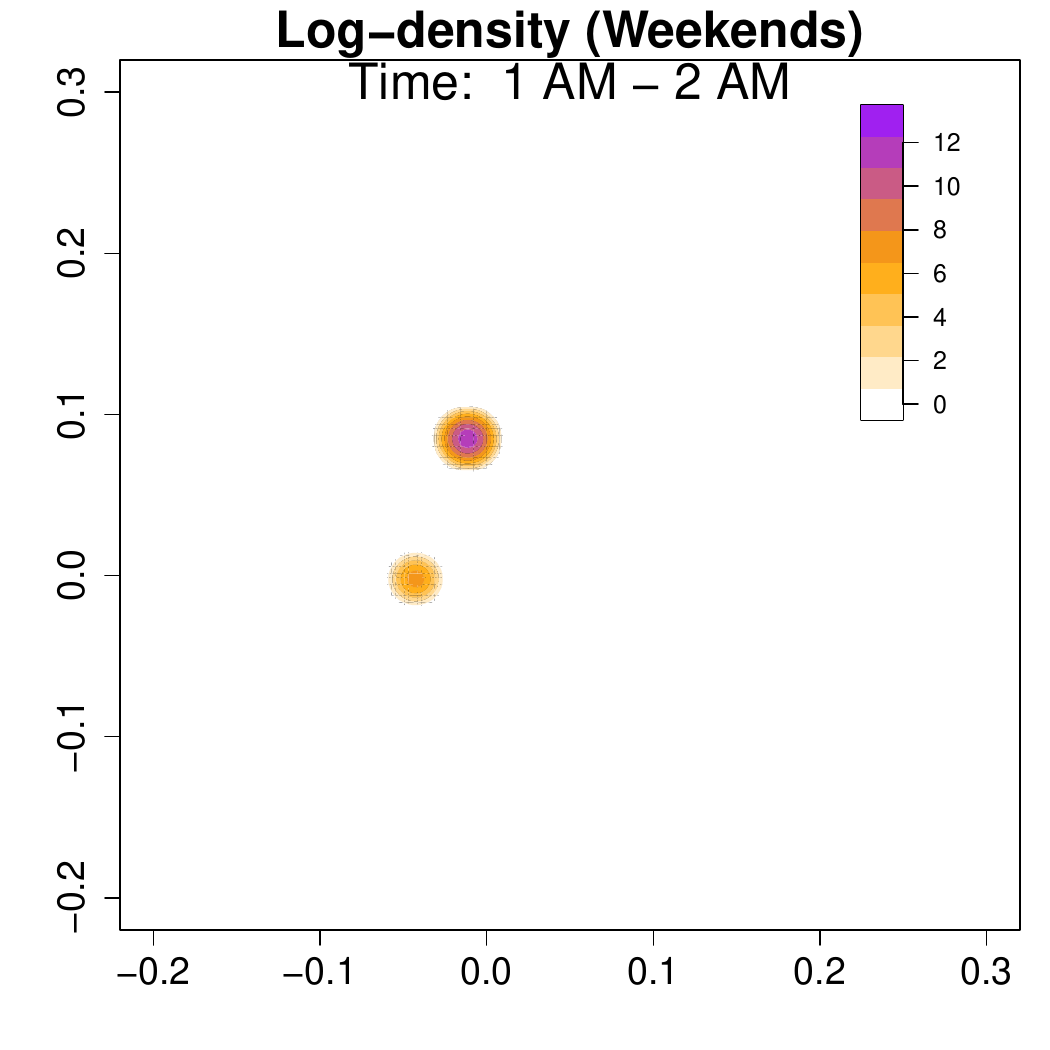}
		\includegraphics[width=0.2\linewidth]{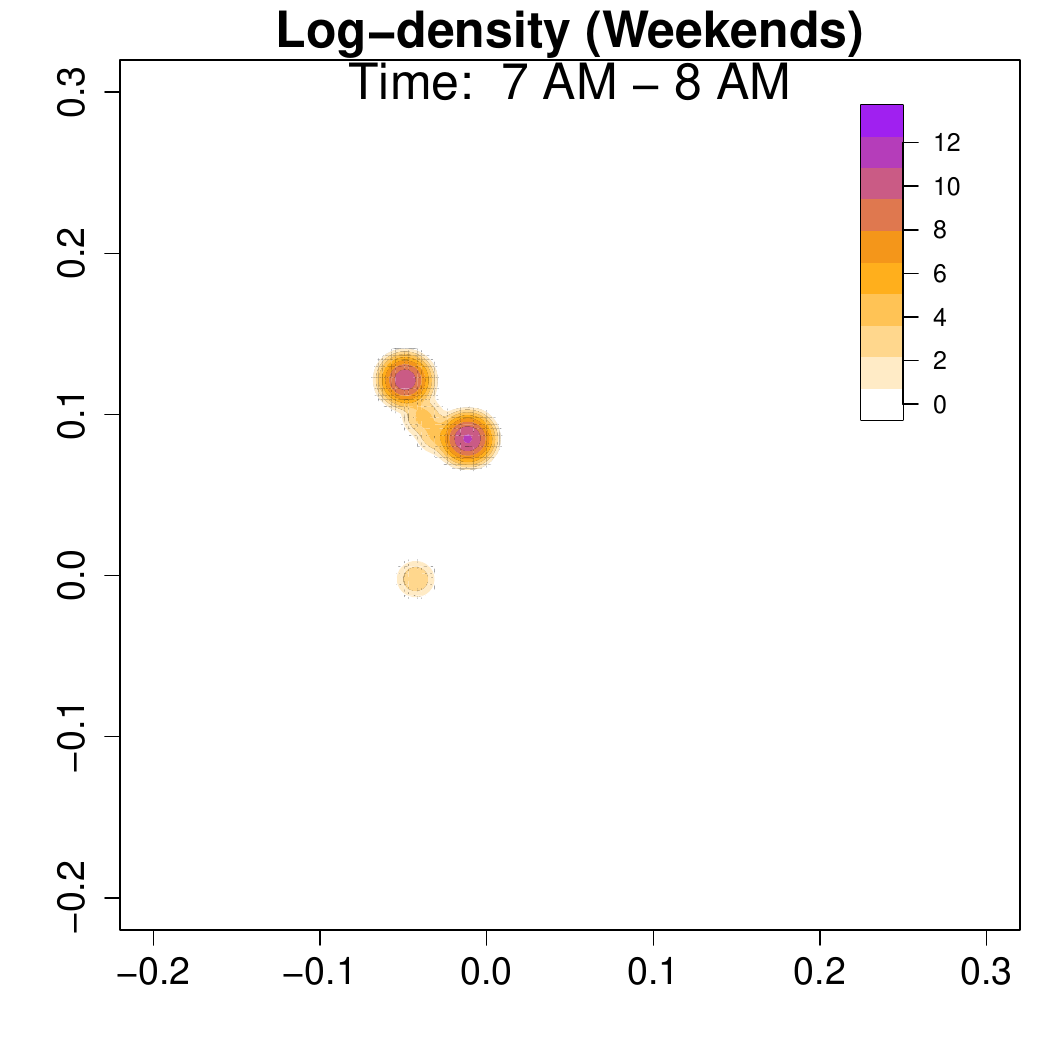}
		\includegraphics[width=0.2\linewidth]{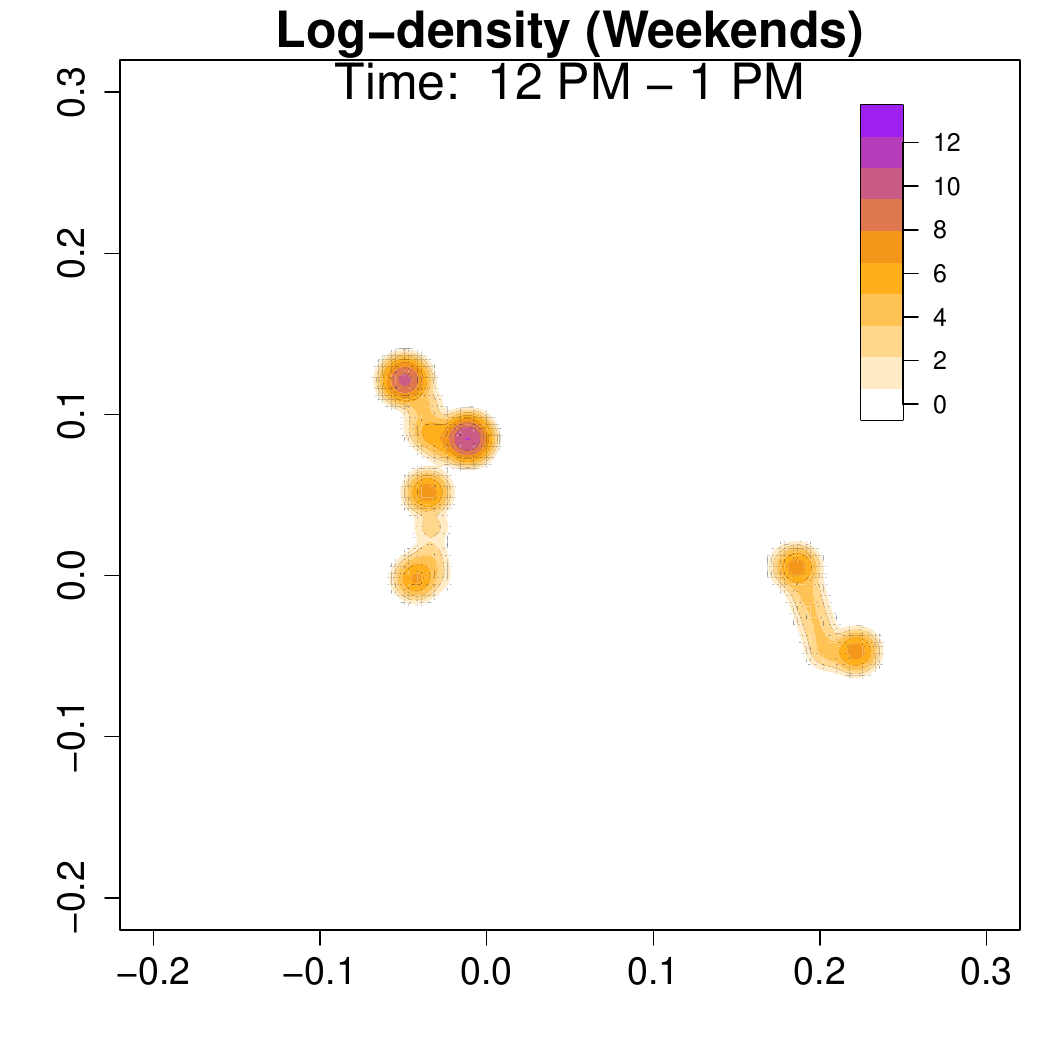}
		\includegraphics[width=0.2\linewidth]{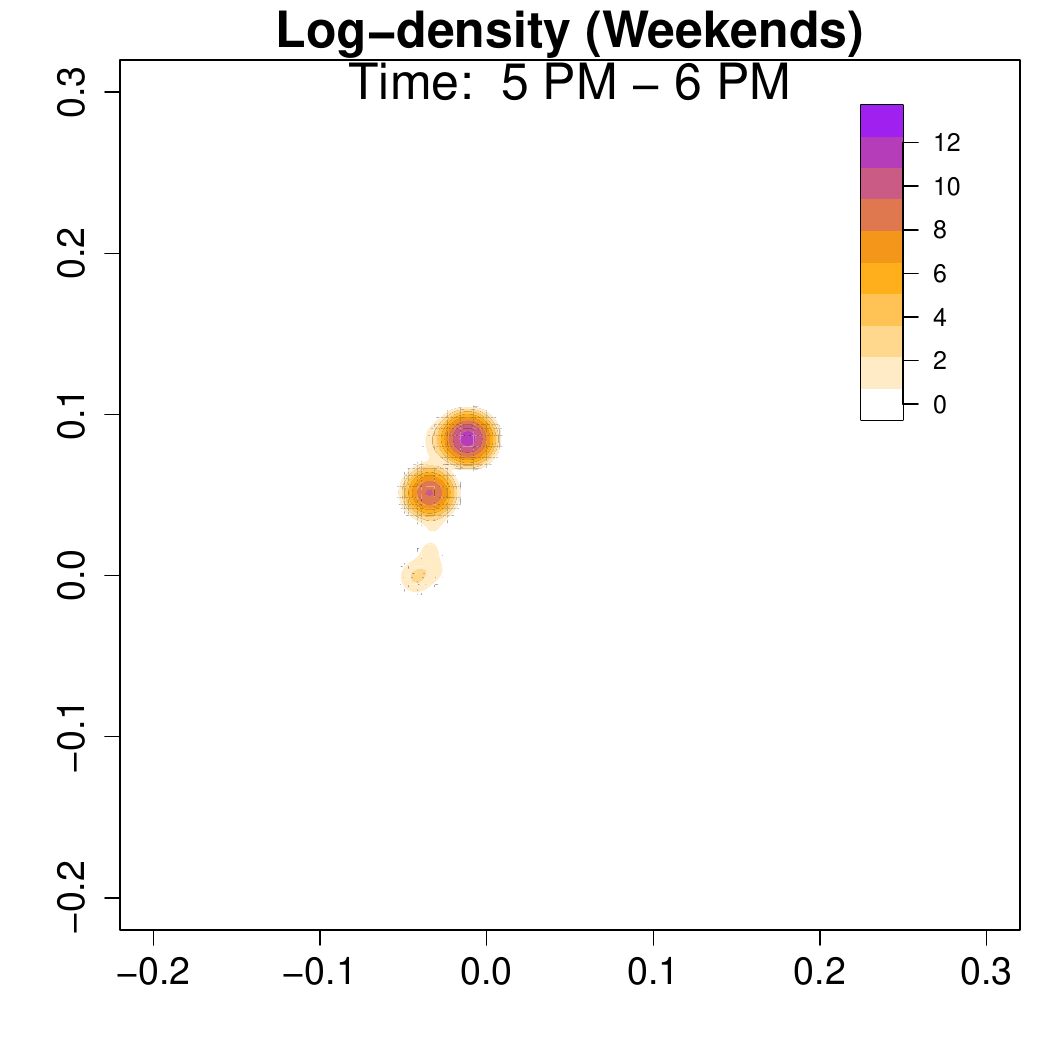}
		\includegraphics[width=0.2\linewidth]{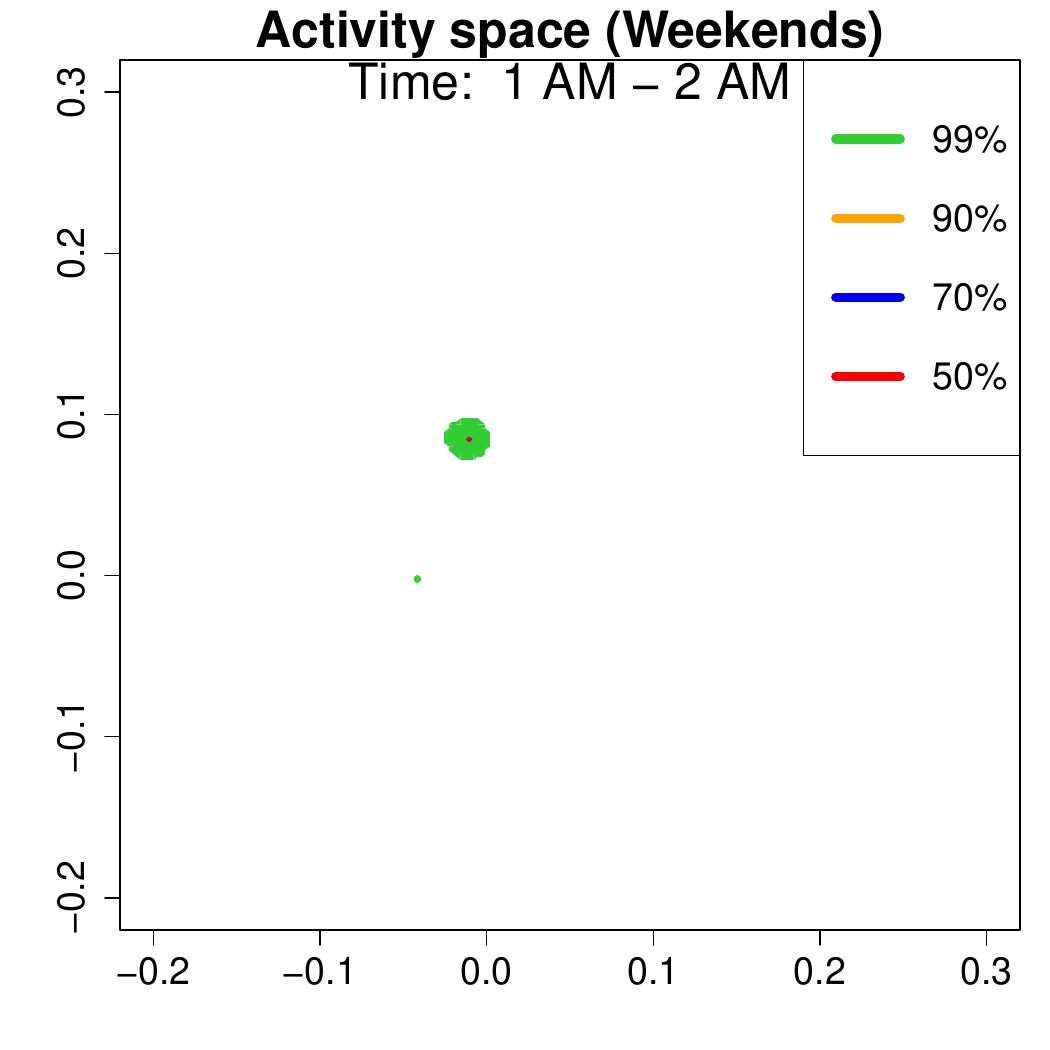}
		\includegraphics[width=0.2\linewidth]{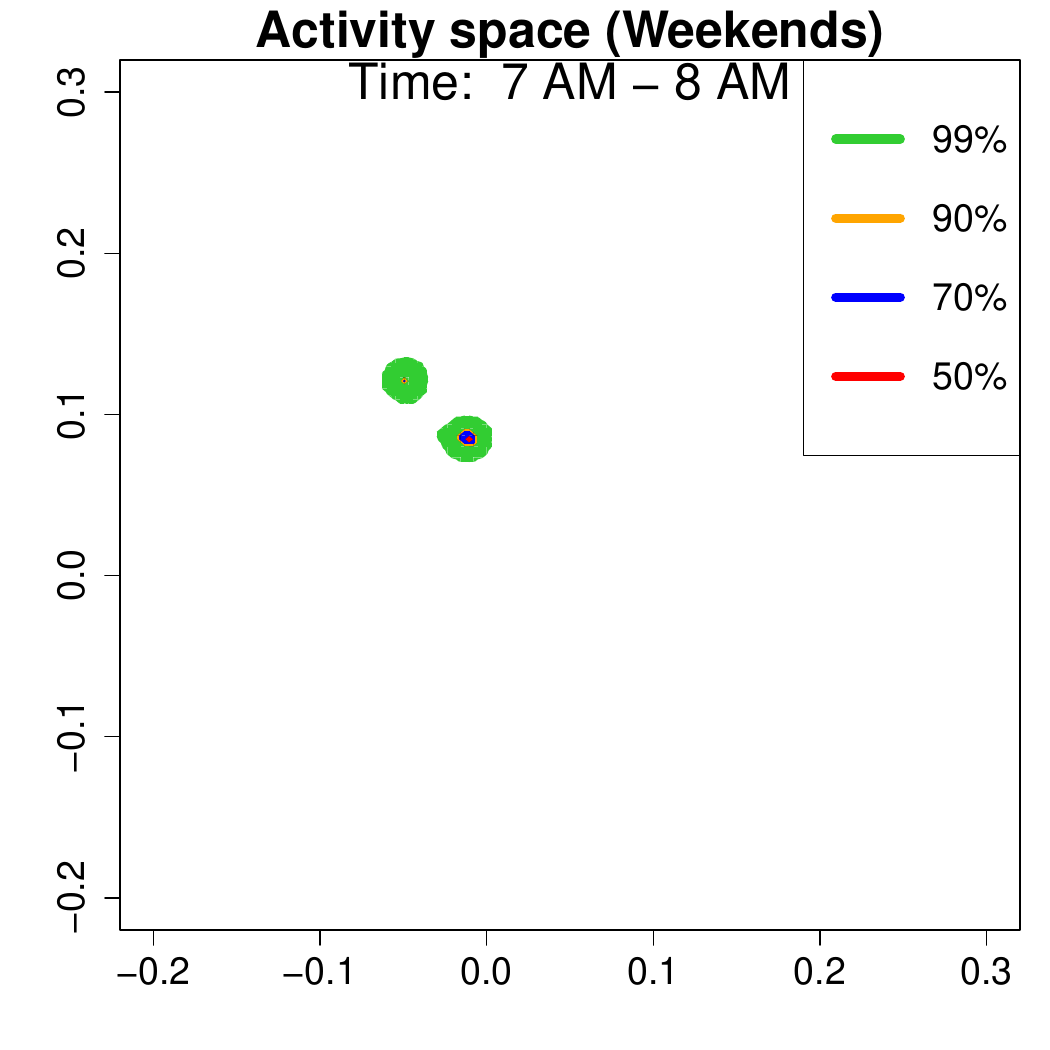}
		\includegraphics[width=0.2\linewidth]{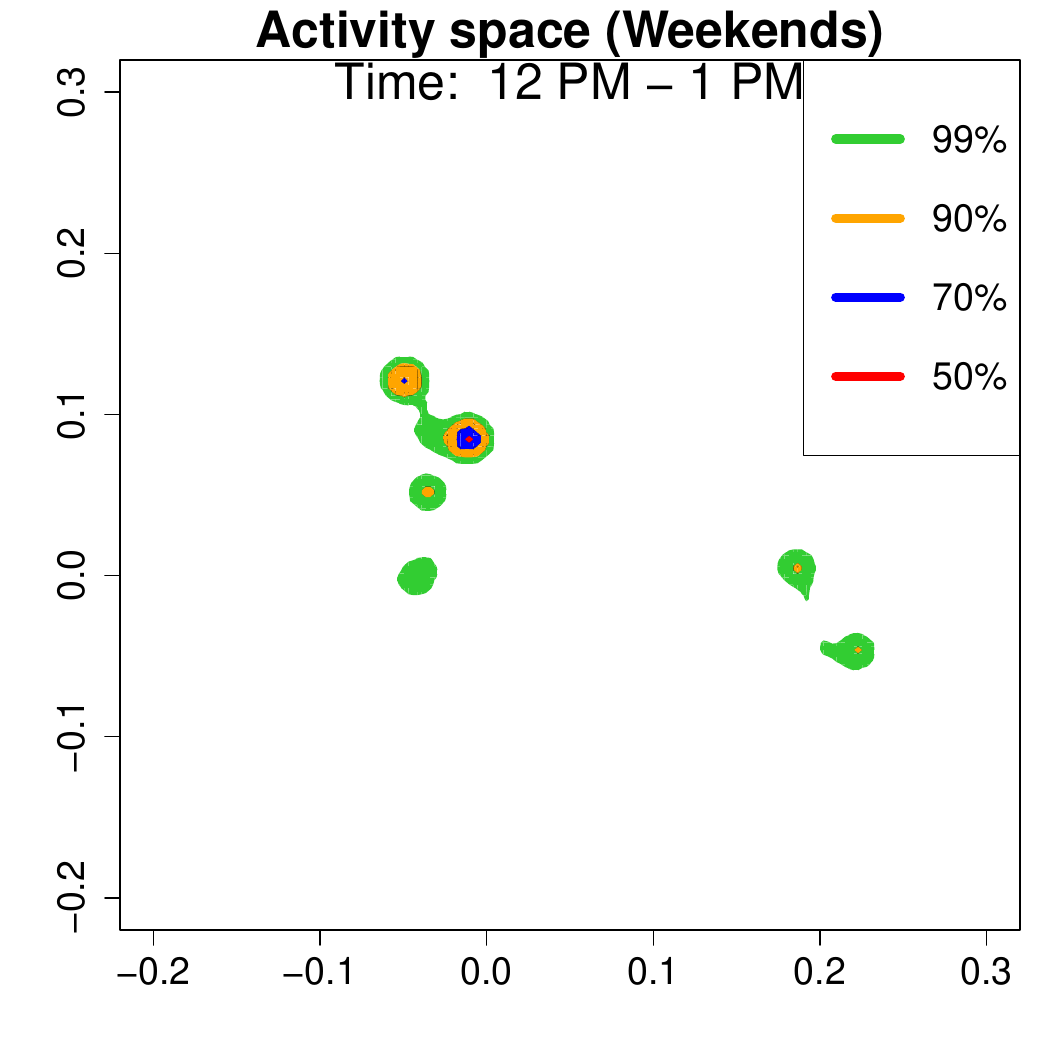}
		\includegraphics[width=0.2\linewidth]{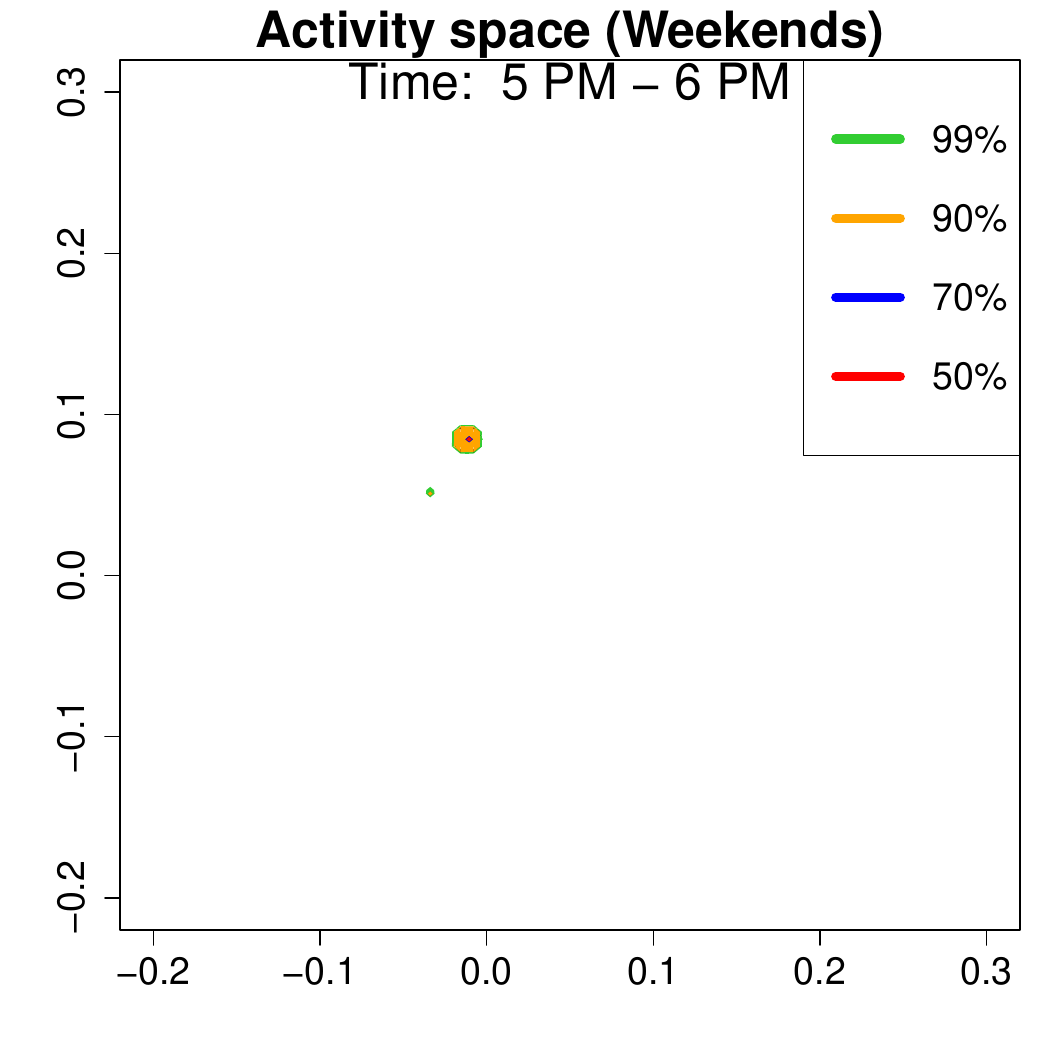}
		\caption{The hourly GPS log-density and activity space of the example individual during the weekends.
			}
		\label{fig:p125we}
	\end{figure}
	
	{\bf Analysis on the weekends patterns.}
	Figure~\ref{fig:p125we} shows the GPS density during weekends, revealing distinct mobility patterns compared to weekdays. Notably,
	the office locations (bottom two spots in the third column of Figure~\ref{fig:p125}) show no density.
	Additional high-density spots appear, likely representing leisure locations (e.g., parks, malls, or beaches).
	The fourth column indicates a frequent weekend location near home, which is a plaza containing some restaurants and shops. 
	Comparing Figures~\ref{fig:p125} and~\ref{fig:p125we}, we observe a more dispersed mobility pattern on weekends, 
	with multiple high-density regions, whereas weekdays exhibit a concentrated movement pattern, consistent with regular commuting behavior.

	\begin{figure}
		\centering
		\includegraphics[width=0.25\linewidth]{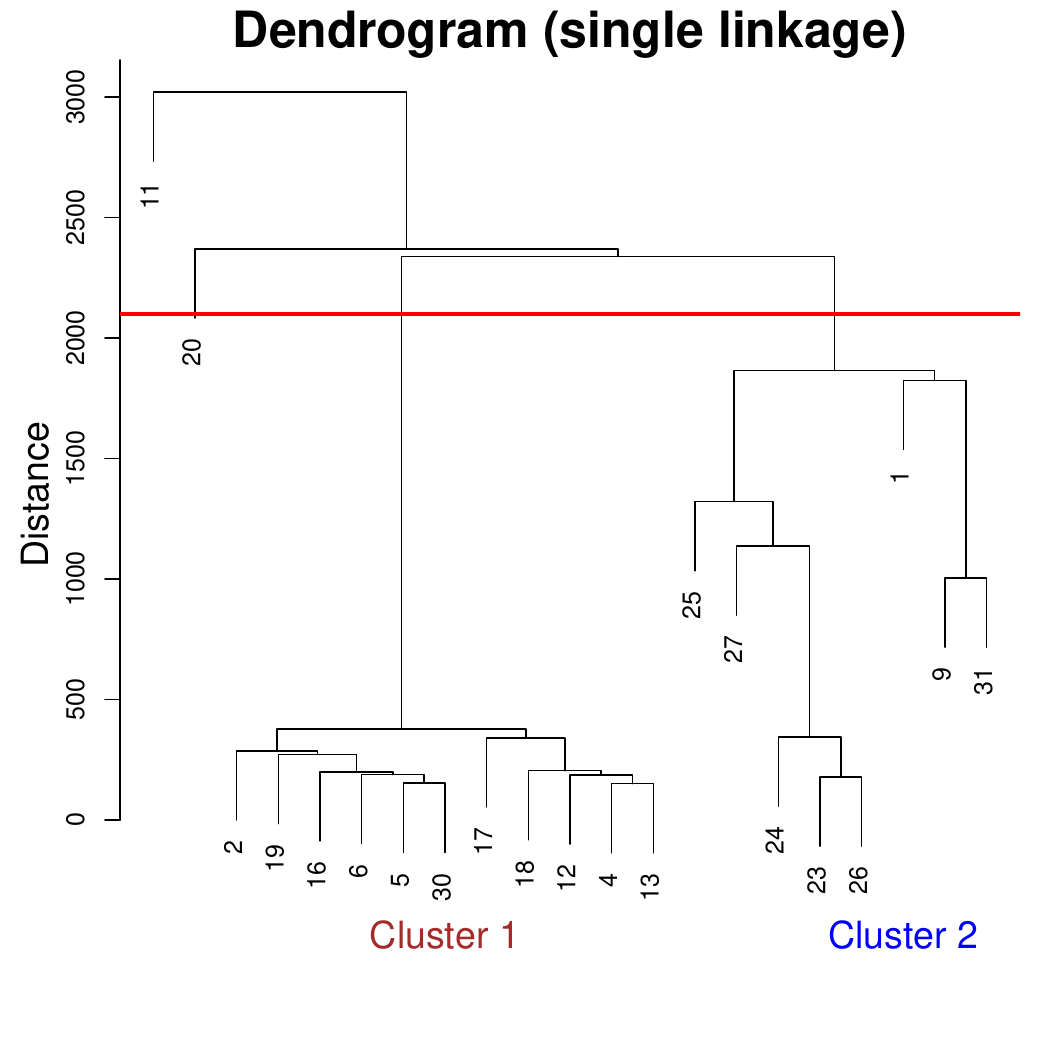}
		\includegraphics[width=0.25\linewidth]{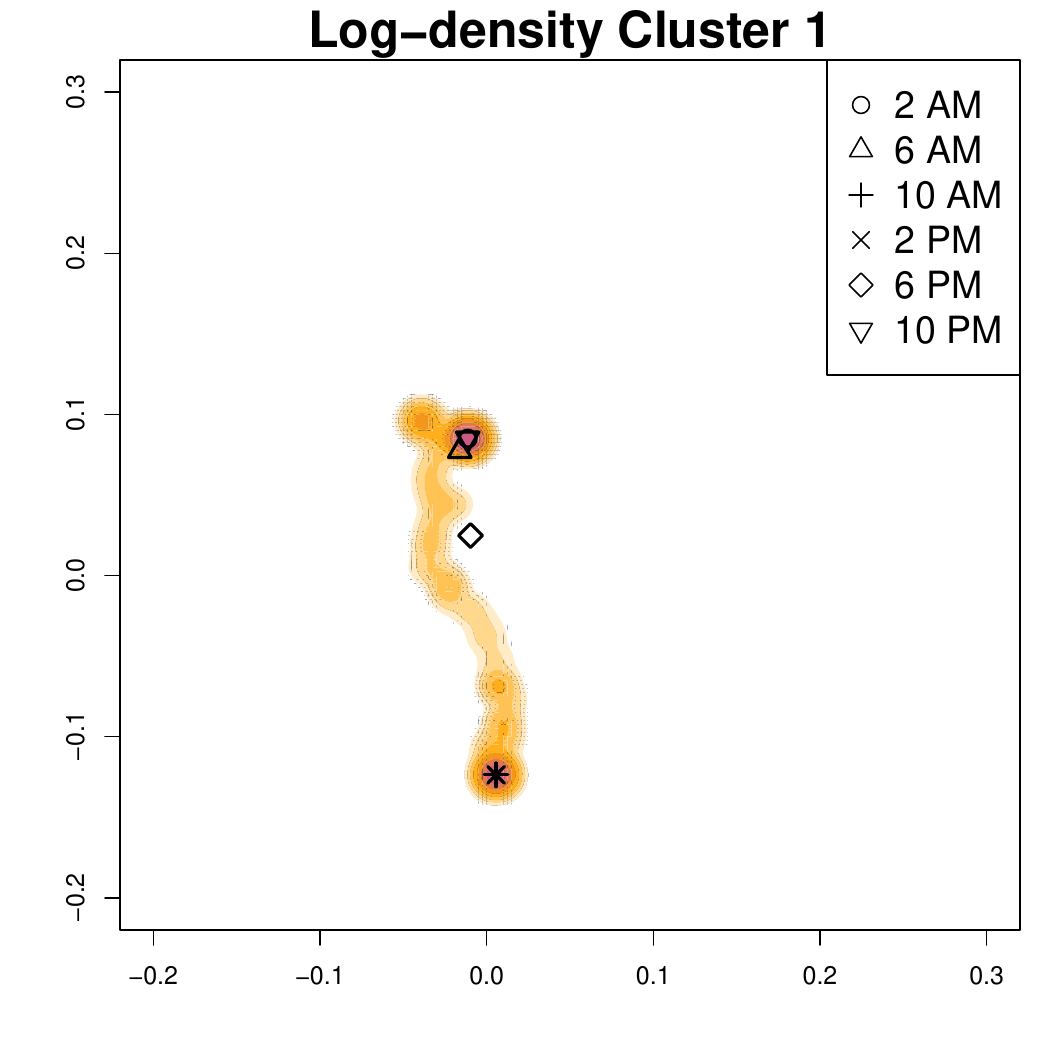}
		\includegraphics[width=0.25\linewidth]{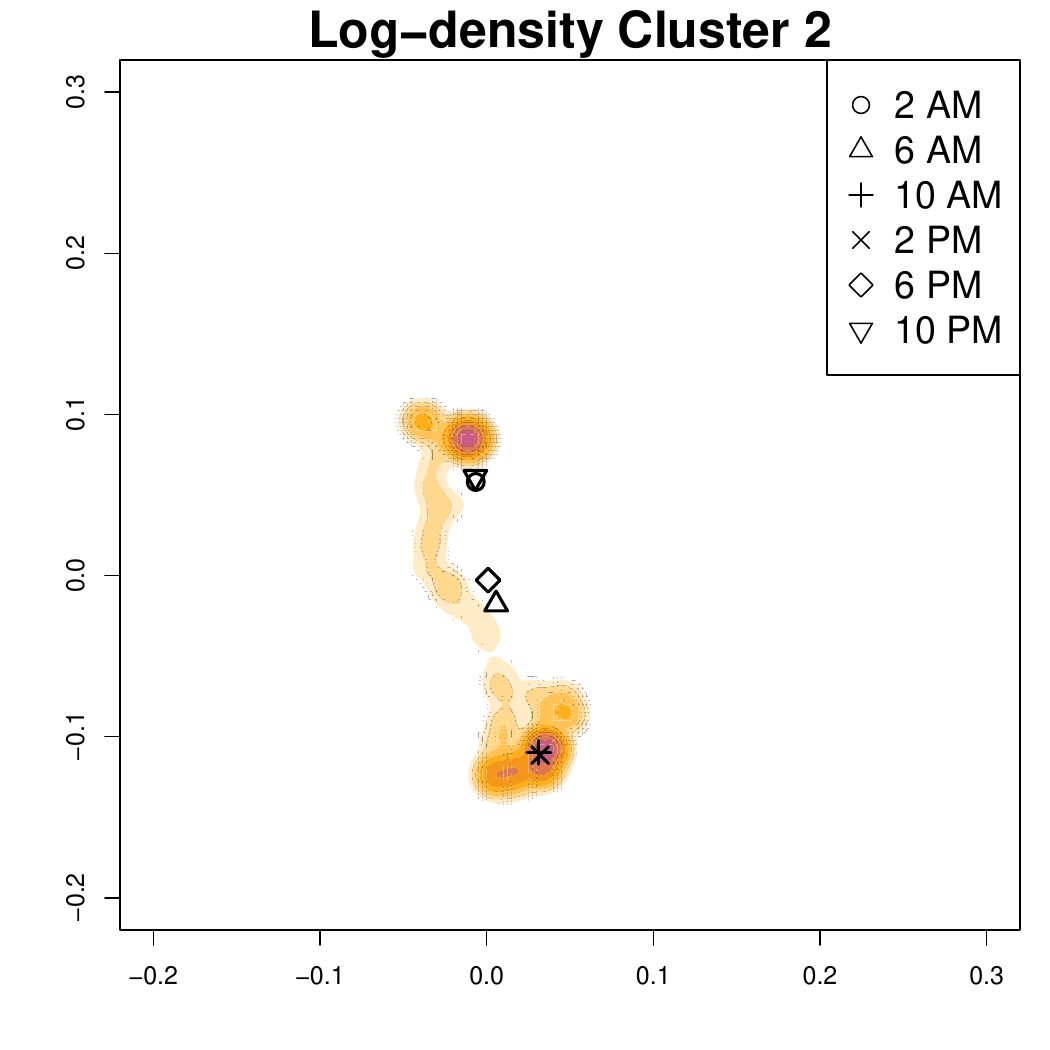}
		\caption{Cluster analysis of the weekdays data. 
			The first panel shows the dendrogram under single linkage clustering.
			We use the threshold as indicated in the red horizontal line to form 4 clusters.
			The third and fourth cluster only contain a single day so we ignore it, focusing on the first two clusters. 
			The middle and right panels display the average GPS density within cluster 1 and cluster 2
			and the corresponding conditional center at different time point.}
		\label{fig:p125clu}
	\end{figure}
	
	{\bf Cluster Analysis on Weekday Mobility.}
	
	Given the structured weekday patterns, we conduct a cluster analysis to further investigate mobility trends. The results are presented in Figure~\ref{fig:p125clu}.
	For each day, we use the same smoothing bandwidths ($h_X$ and $h_T$ ) and compute pairwise distances via the integrated squared difference between log-densities:
	$$
	D_{ij} = \int \left[\log(\hat f_{c,i}(x)+1) - \log(\hat f_{c,j}(x)+1)\right]^2 dx.
	$$
	Applying single-linkage hierarchical clustering, the dendrogram (first panel of Figure~\ref{fig:p125clu}) identifies four clusters, separated by the red horizontal line.	
	Cluster 3 contains Day 11 exclusively. All GPS records on Day 11 are located at home, which corresponds to a wildfire event in the individual's region that forced he or she to remain indoors throughout the day.  So we classify it as an outlier.
	And cluster 4 also consists of one day (Day 20).  On this day, the individual made an additional stop when going back home from office, visiting an entertainment venue that was only recorded once in the entire dataset. 
	 
	We focus on clusters 1 and 2, comparing their average GPS densities and conditional centers at six different time points. 
		Although both clusters represent commuting patterns, they exhibit distinct anchor locations around the workplace and different evening behaviors.
		In the lower part of the second and third panels of Figure~\ref{fig:p125clu}, Cluster 1 shows two high-density regions in the evening, primarily at home and near the office, whereas Cluster 2 contains additional high-density area representing an alternative place to stay overnight (to the right of the main office location) and a park nearby. 
		In particular, on days in Cluster 2, the individual sometimes stays in an alternative home overnight after leaving office rather than returning to the primary home, 
		 or visits the park, in the evening (e.g., for a walk) after staying in the alternative home for about 2 hours (and then goes back to the alternative home after visiting the park)
		By contrast, in Cluster 1, the individual tends to commute directly between home and office without these additional stops.
		
	To highlight these differences, Figure~\ref{fig:p125clu3} compares log-density estimates for four specific time intervals:  
	7-8 AM (morning rush hour),
	5-6 PM (evening rush hour),
	7-8 PM (post-work activities), and
	9-10 PM (late evening).
	The top row (Cluster 1) exhibits only 2 high-density regions in the evening, while the bottom row (Cluster 2) shows 4 high-density regions (home, office, park and alternative apartment).
	As a result, Cluster 2's pattern emphasizes extended or alternative routines after work including returning to the alternative home or a visit to the park near the work place.
	The clustering analysis uncovers two distinct weekday mobility patterns:
	Cluster 1: The individual commutes directly between home and office.
	Cluster 2: The individual has two options after working: the individual goes back home, otherwise, the individual goes back to the alternative home (sometimes also containing the visit to park after working).
	This method effectively reveals hidden mobility structures in the data, demonstrating the utility of GPS density estimation in understanding individual movement behavior.

	\begin{figure}
		\centering
		\includegraphics[width=0.2\linewidth]{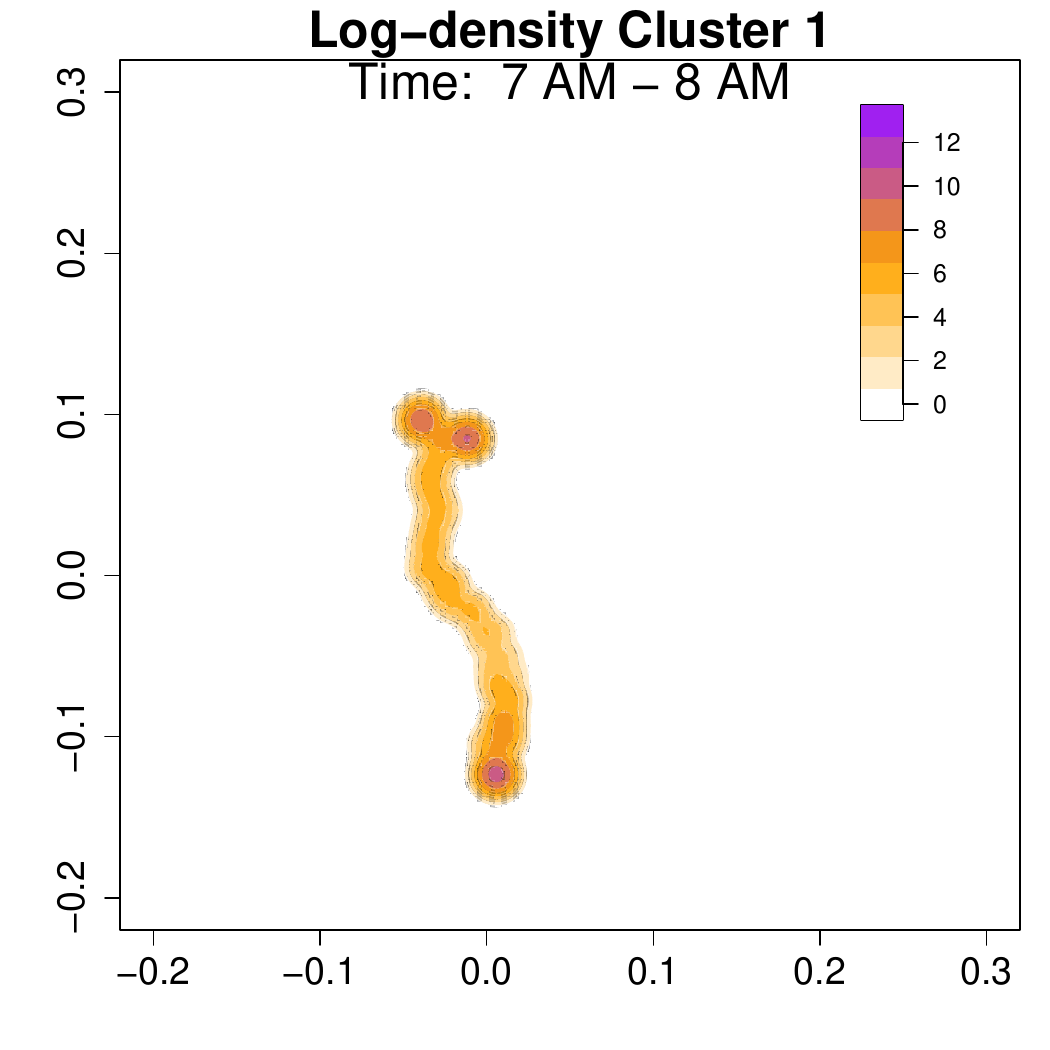}
		\includegraphics[width=0.2\linewidth]{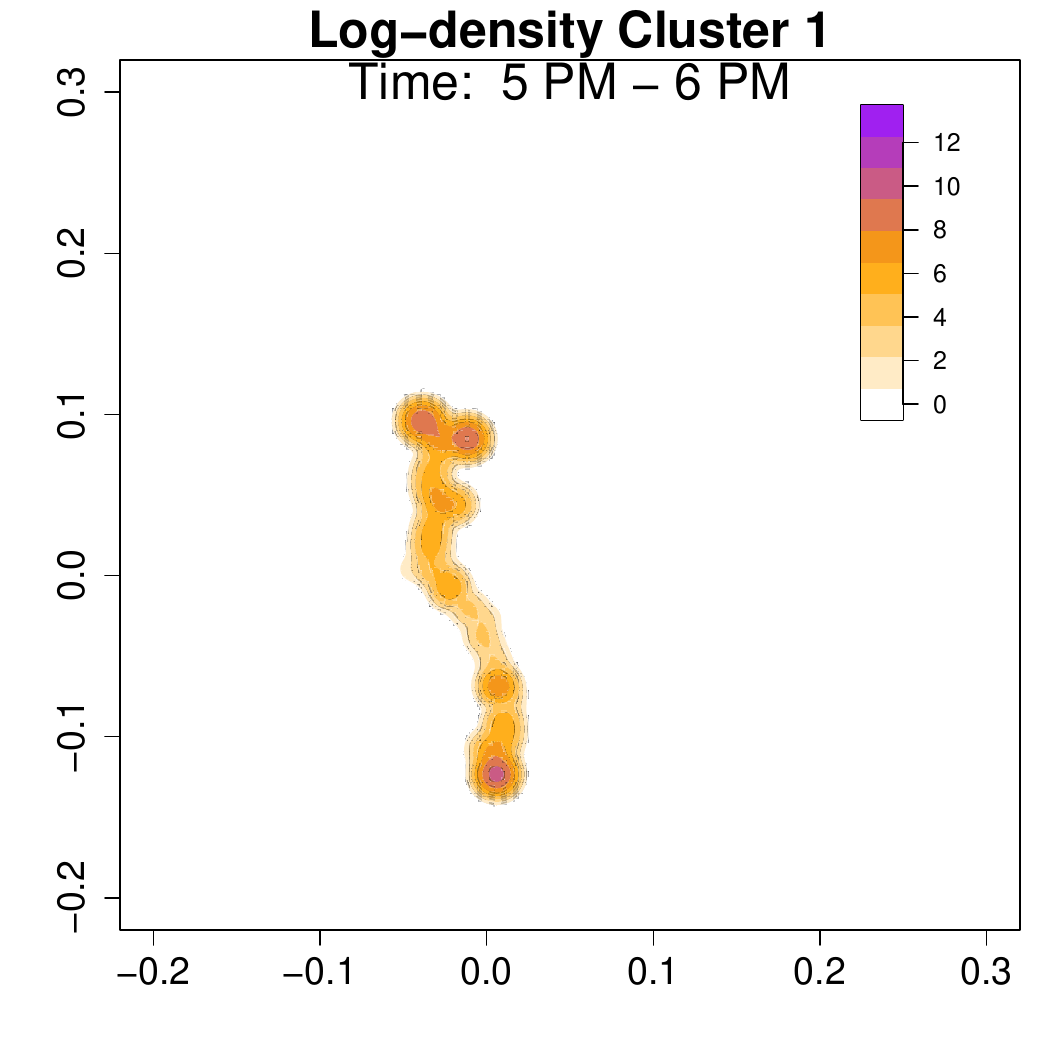}
		\includegraphics[width=0.2\linewidth]{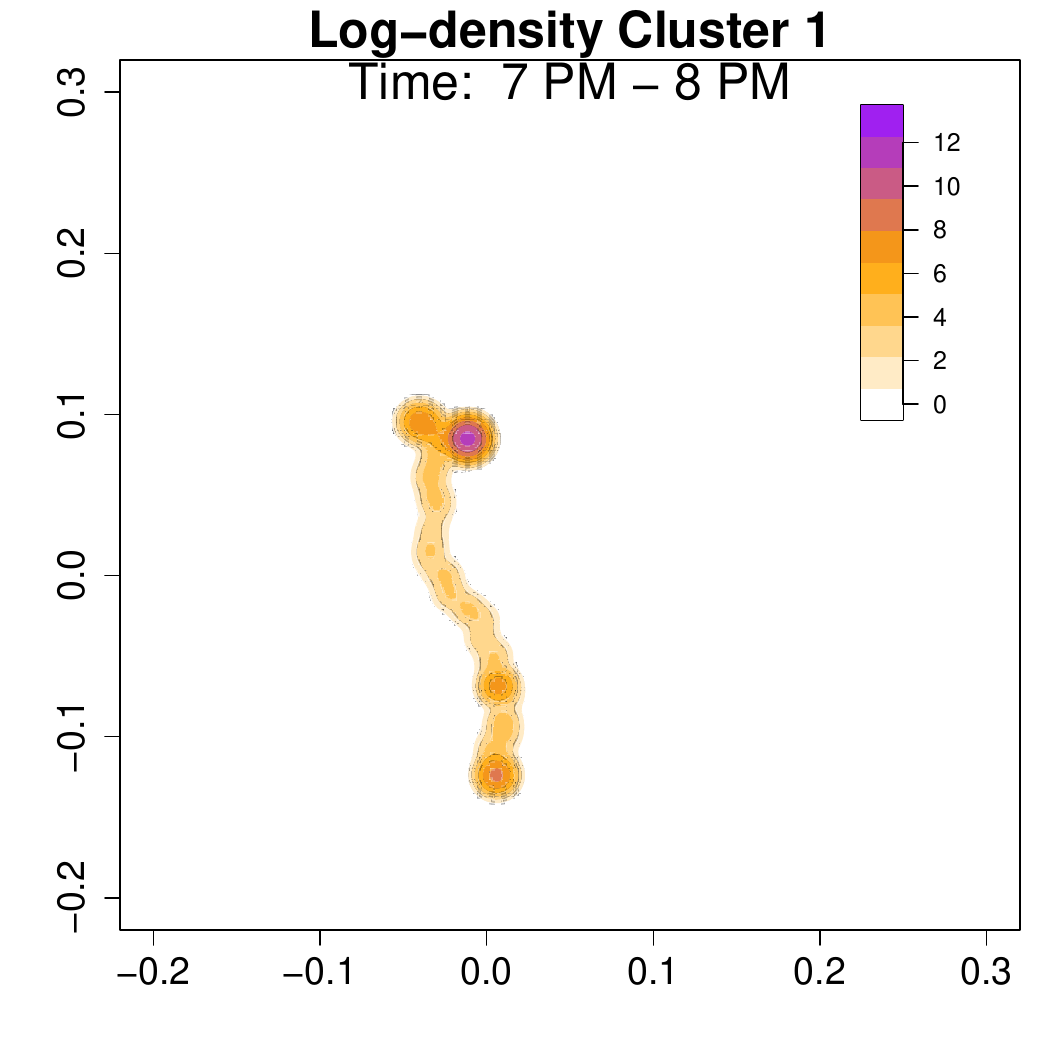}
		\includegraphics[width=0.2\linewidth]{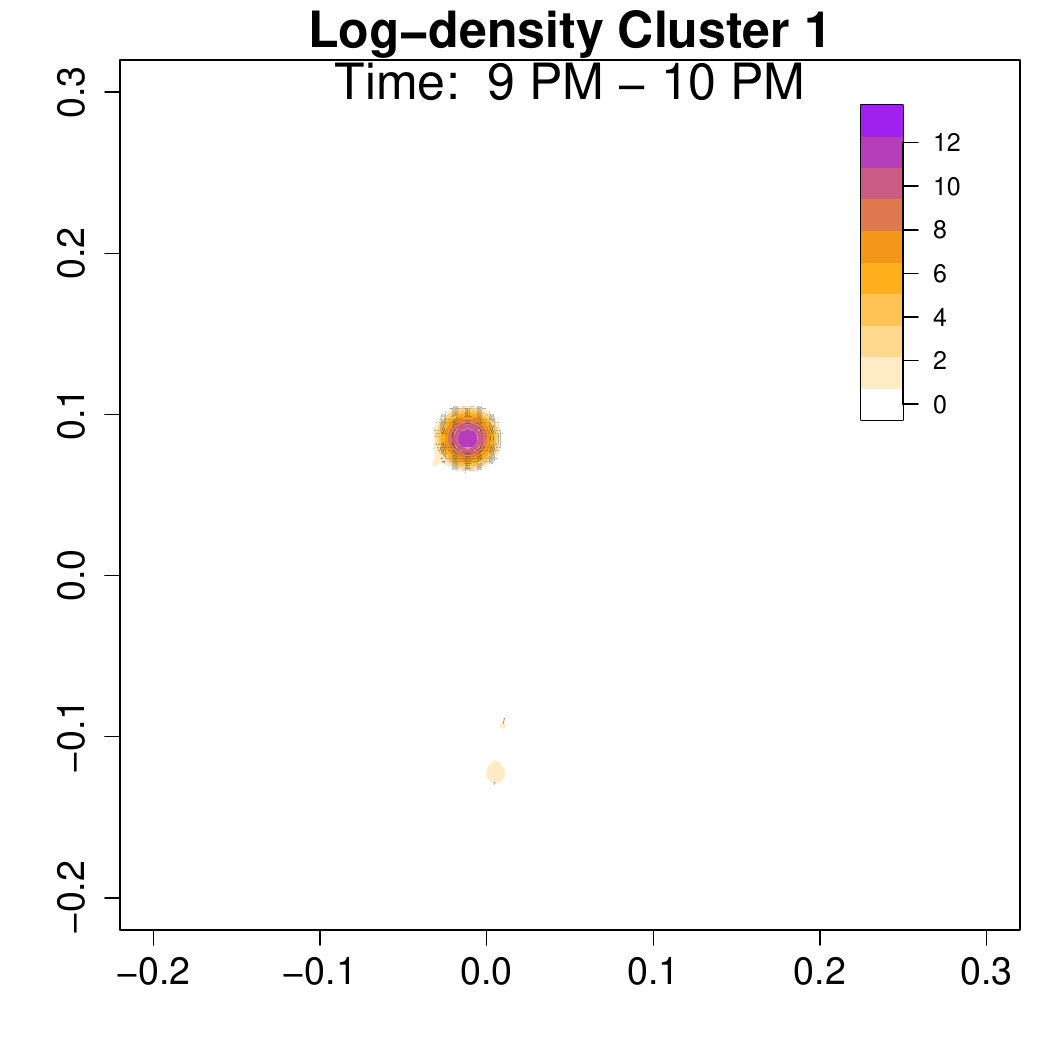}
		\includegraphics[width=0.2\linewidth]{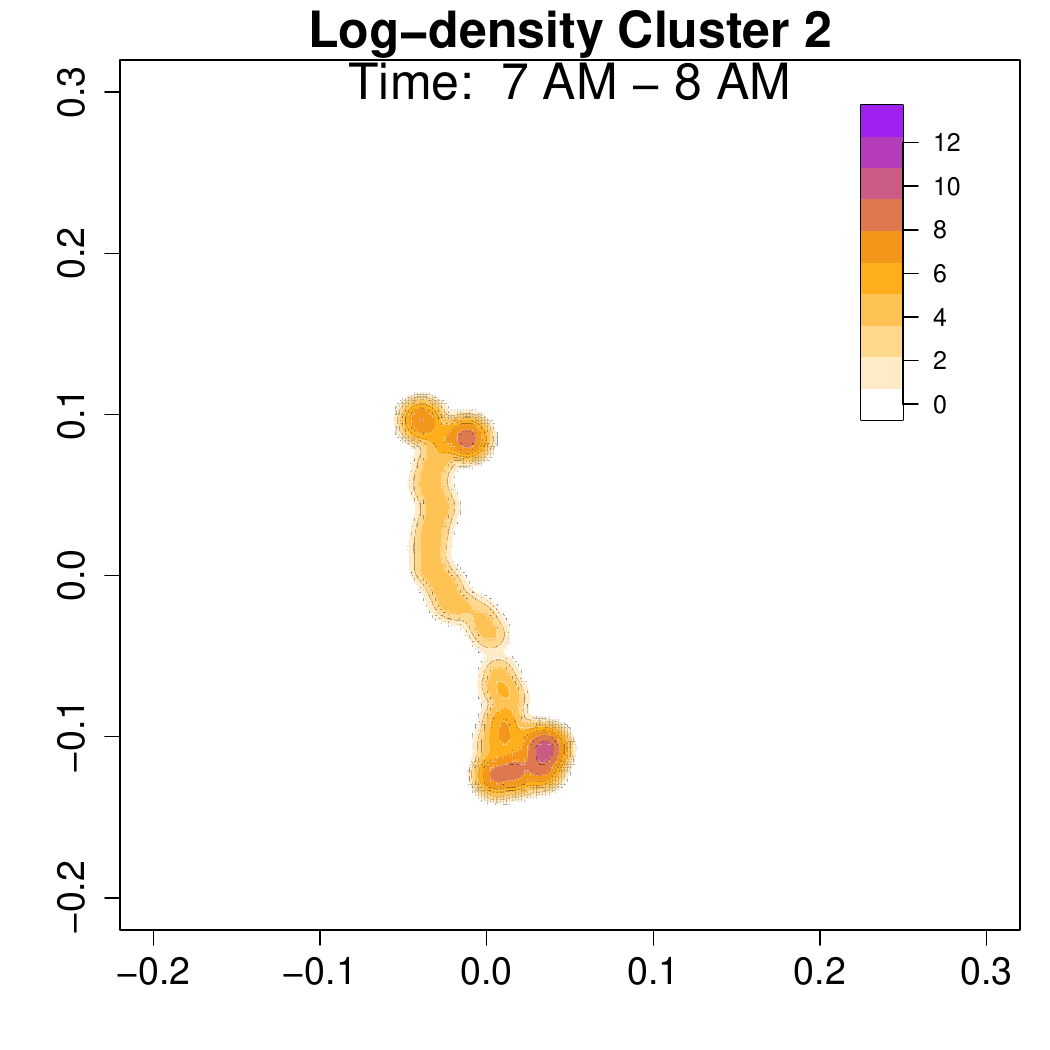}
		\includegraphics[width=0.2\linewidth]{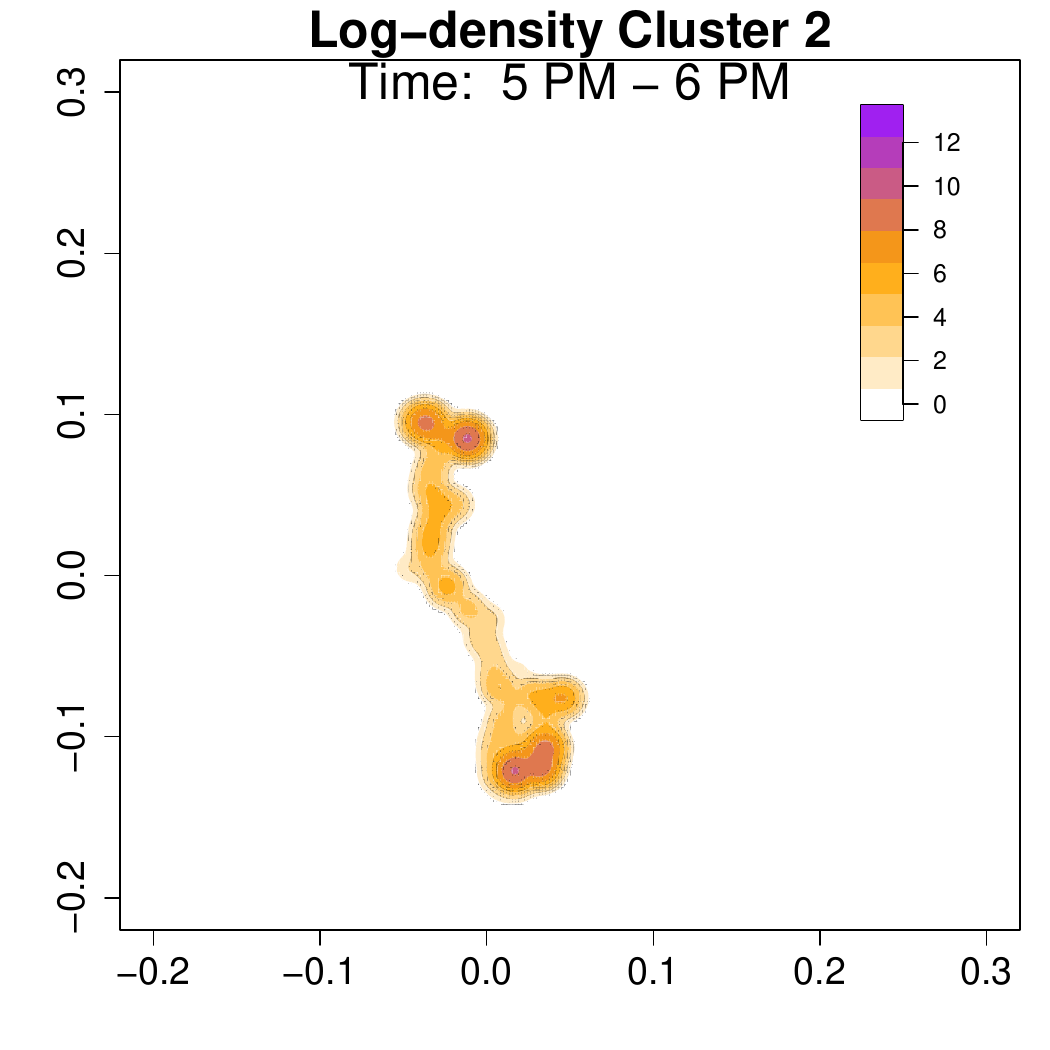}
		\includegraphics[width=0.2\linewidth]{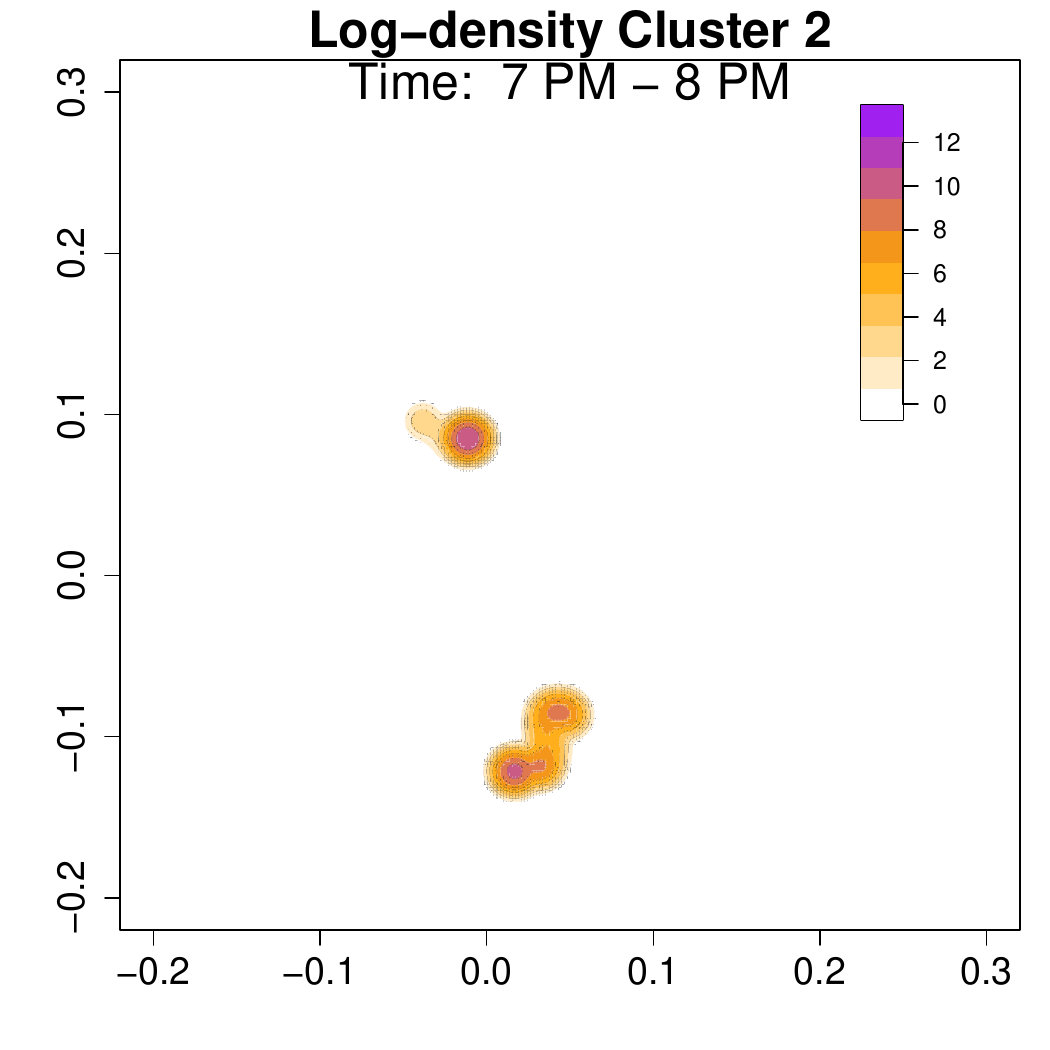}
		\includegraphics[width=0.2\linewidth]{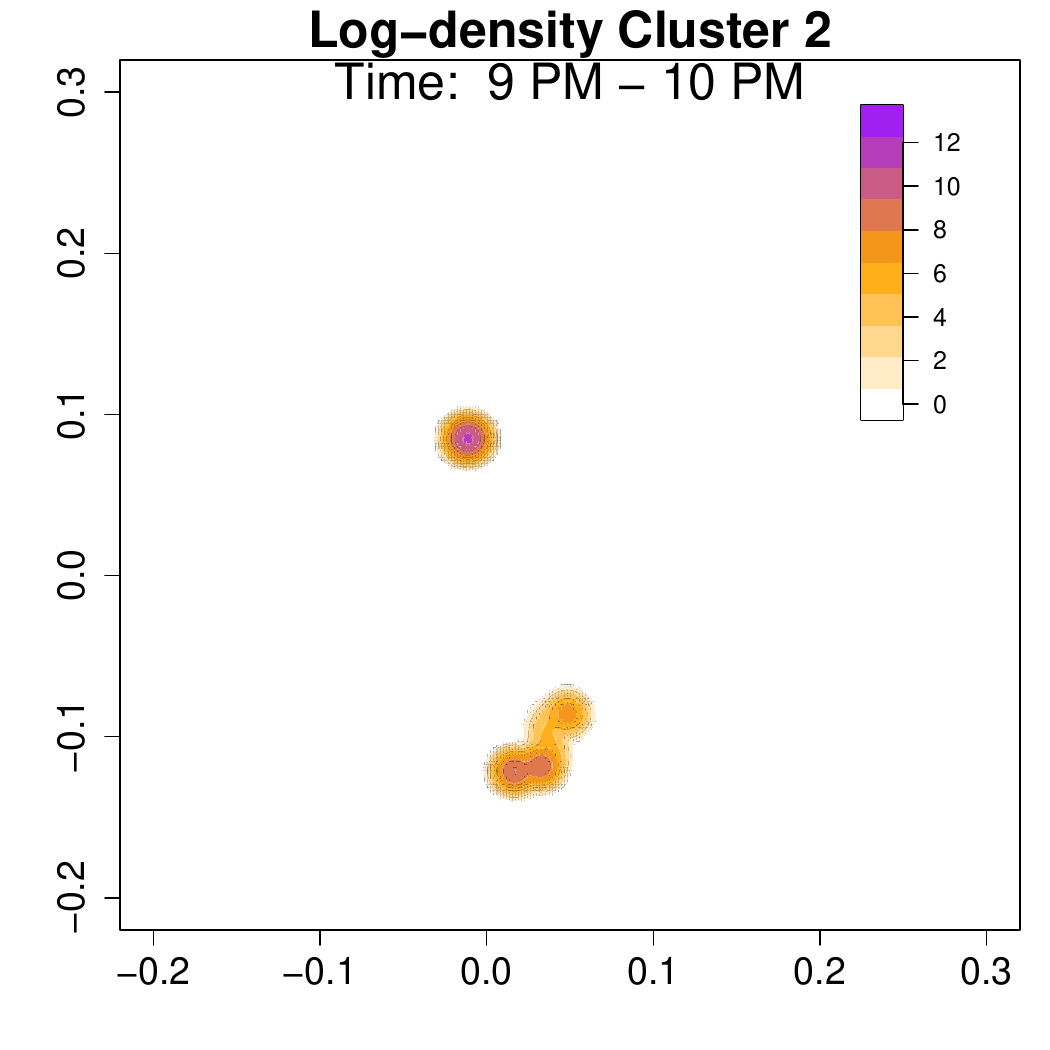}
		\caption{A comparison of the log-density between the two clusters at four specific time intervals: 7-8 AM, 5-6 PM, 7-8 PM, and 9-10 PM.
			In the morning rush hour, the two clusters show a similar pattern. 
			In the evening time (second to the last columns), the two cluster show significantly different patterns.}
		\label{fig:p125clu3}
	\end{figure}

	\section{Discussion and future directions}
	In this paper, we introduced a statistical framework for analyzing GPS data. The general framework presented in Section~\ref{sec::gps} applies to GPS data collected from various entities, extending beyond human mobility. For studying human activity patterns, we developed the SMM in Section~\ref{sec::smm}. Despite its simplicity, the SMM effectively captures major movement patterns within a single day and enables meaningful analyses. Our framework lays the foundation for numerous future research directions.

	{\bf Population study. }
	While this paper primarily focuses on analyzing the activity patterns of a single individual, extending our framework to a population-level analysis remains an open challenge. One key difficulty lies in the variation of anchor locations across individuals, as each person has distinct home and workplace locations. To accommodate population-level studies, modifications to the SMM are necessary to generalize anchor locations and account for heterogeneous mobility patterns. Also, a population-level study will lead to an additional asymptotic regime where the number of individual increases. Together with the two asymptotic regimes of the current framework (number of days and temporal resolution), we encounter a scenario with three asymptotics. The convergence rate and choice of tuning parameters will require adjustments. 

	{\bf Connections to functional data. }
	Our framework transforms GPS data into a density map, establishing interesting connections with functional data analysis \citep{ramsay2002applied, wang2016functional}. In particular, the daily average GPS density 
	$\hat f_{GPS, i}$
	for an individual can be viewed as independent and identically distributed (IID) 2D functional data under the assumption that timestamps are generated in an IID setting. Existing functional data analysis techniques may therefore be applicable, offering new perspectives on GPS-based modeling. Notably, methods introduced in Section~\ref{sec::smm} suggest potential links between structural assumptions in the SMM and those in functional data analysis. 
	In addition, when the sampling frequency is high, the two coordinates in the GPS data can be separated viewed as two functional data with a common but irregular spacing \citep{gervini2009detecting}. Under this point of view, we may directly apply methods from functional data to handle GPS data. 
	This alternative perspective may lead to a different framework for analyzing GPS data. Further exploration of these connections is an important research direction.
	
	{\bf Limiting regime of measurement errors. }
	Our current analysis considers asymptotic limits where the number of days $n\rightarrow\infty$ and
	the temporal resolution $|t_{i,j}-t_{i,j+1}|\rightarrow0$. 
	Another intriguing limit involves the measurement noise level, where $\E(\|\epsilon\|^2)\rightarrow 0$,
	reflecting advancements in GPS technology. However, as measurement errors vanish, the density function
	$f_{GPS}$ 
	becomes ill-defined. In this noiseless scenario, the distribution of a random trajectory point
	$X^* = X(U)$ 
	consists of a mixture of point masses (at anchor locations) and 1D distributions representing movement along trajectories. 
	Addressing the challenges posed by singular measures \citep{10.1214/19-AOAS1311, chen2019generalized} in this regime is a crucial area for future investigation.
	
	{\bf Alternative random trajectory models. }
	While the SMM proposed in Section~\ref{sec::smm} provides practical utility and meaningful insights, it does not precisely model an individual's movement. Specifically, the SMM assumes a constant velocity when an individual is in motion, leading to a step-function velocity profile with abrupt jumps at transitions between movement and stationary states. This assumption deviates from realistic human mobility patterns. Developing more sophisticated models that accurately capture velocity dynamics and movement variability remains an open problem for future research.

	{\bf Topological data analysis. }
	Topological data analysis (TDA; \citealt{wasserman2018topological, chazal2021introduction}) is a data analysis technique to study the shape of features inside data. 
	A common application of TDA is to study the morphology of the probability density function. 
	The TDA can be applied to analyze GPS data in two possible ways.
	The first application is to use TDA to analyze the morphology of
	GPS density functions. This allows us to further understand the activity space.
	The second applications is to use TDA to investigate the latent trajectories generated the GPS data. 
	The common/persistent features of trajectories informs us important characteristics
	of the individual's activity patterns.


	{\small
		\bibliographystyle{apalike}
		\bibliography{AMS}
	}

	\pagebreak
	\bigskip
	\begin{center}
		{\large\bf SUPPLEMENTARY MATERIAL}
	\end{center}
	
	\appendix
	
	\begin{description}
		
		\item[Numerical techniques (Section \ref{sec::numerics}):]
		We provide some useful numerical techniques for computing our estimators.
		
		\item[Simulations (Section \ref{sec::simulation}):] We provide a simulation study to investigate the performance of our method.
		
		\item[Assumptions (Section \ref{sec::assum}):]  All technical assumptions for the asymptotic theory are included in this section.
		
		\item[Proofs (Section \ref{sec::proof}):] The proofs of the theoretical results are given in this section.
		
		\item[Data and R scripts:] See the online supplementary file.
		
	\end{description}

	\section{Numerical techniques}	\label{sec::numerics}
	
	\subsection{Conditional estimator as a weighted estimator} \label{sec::ckde::comp}
	Recall that the estimator $\hat f_c(x)$
	in equation \eqref{eq::int_ckde} is
	\begin{equation*}
		\begin{aligned}
			\hat f_c(x) = \int_0^1 \hat f(x|t) dt &= \frac{1}{nh^2}\sum_{i=1}^n \sum_{j=1}^{m_i}  \tilde{W}_{i,j} \cdot K\left(\frac{X_{i,j}-x}{h}\right),\\
			\tilde{W}_{i,j} &=  n\cdot \int_0^1\frac{K_T\left(\frac{d_T(t_{i,j},t)}{h_T}\right)}{\sum_{i'=1}^n\frac{1}{m_{i'}}\sum_{j'=1}^{m_{i'}}K_T\left(\frac{d_T(t_{i',j'},t)}{h_T}\right)}dt
		\end{aligned}
	\end{equation*}
	and the conditional PDF in equation \eqref{eq::ckde},
	$$
	\hat f(x|t) =  \frac{\sum_{i=1}^n\frac{1}{m_i} \sum_{j=1}^{m_i} K\left(\frac{X_{i,j}-x}{h_X}\right)K_T\left(\frac{d_T(t_{i,j},t)}{h_T}\right)}{
		h_X^2 \cdot \sum_{i'=1}^n \frac{1}{m_{i'}}\sum_{j'=1}^{m_{i'}} K_T\left(\frac{d_T(t_{i',j'},t)}{h_T}\right)} = 
	\frac{1}{nh^2}\sum_{i=1}^n \sum_{j=1}^{m_i}  \omega_{i,j}(t) \cdot K\left(\frac{X_{i,j}-x}{h}\right),
	$$
	where 
	$$
	\omega_{i,j}(t) =n\cdot  \frac{K\left(\frac{X_{i,j}-x}{h_X}\right)K_T\left(\frac{d_T(t_{i,j},t)}{h_T}\right)}{\sum_{i'=1}^n \frac{1}{m_{i'}}\sum_{j'=1}^{m_{i'}} K_T\left(\frac{d_T(t_{i',j'},t)}{h_T}\right)}.
	$$
	is a local weight at time $t $.
	So both estimators are essentially the weighted 2D KDE with different weights. 
	This allows us to quickly compute the estimator using existing kernel smoothing library such as the \texttt{ks} library in \texttt{R}.
	
	Moreover, the weights are connected by the integral
	$$
	\tilde{W}_{i,j} = \int_0^1  \omega_{i,j}(t). 
	$$
	If we are interested in only the time interval $A\subset [0,1]$,
	we just need to adjust the weight to be
	$$
	\tilde{W}_{A, i,j} = \int_{A}  \omega_{i,j}(t)
	$$
	and rerun the weighted 2D KDE.
	In practice, we often perform a numerical approximation to the integral based on a dense grid of $[0,1]$.
	Computing the weights $\tilde{W}_{A, i,j}$ can be done easily by
	using only the grid points inside the interval $A$.

	\subsection{A fast algorithm for the activity space}	\label{sec::num::as}
	
	Recall that the $\rho$-activity space is
	\begin{equation*}
		\begin{aligned}
			\hat Q_\rho &= \min_{\lambda}\left\{ \hat{ \mathcal{L}}^*_\lambda:  \frac{1}{n}\sum_{i=1}^n \sum_{j=1}^{m_i} W^\dagger_{i,j} I(X_{i,j}\in  \hat{ \mathcal{L}}^*_\lambda ) \geq \rho \right\},\\
			\hat{ \mathcal{L}}^*_\lambda& = \left\{x: \hat {f}(x)\geq\lambda \right\}.
		\end{aligned}
	\end{equation*}

	Numerically, the minimization in $\hat Q_\rho$ is not needed. 
	Since the time-weight EDF $\hat F_w(x)$
	is a distribution function with a probability mass of $\frac{1}{n} W^\dagger_{i,j}$ on $X_{i,j}$.
	Therefore, the summation $\frac{1}{n}\sum_{i=1}^n\sum_{j=1}^{m_i} W^\dagger_{i,j} I(X_{i,j}\in  \hat{ \mathcal{L}}^*_\lambda )$
	remains the same when we vary $\lambda$ without including any observation $X_{i,j}$.
	Using this insight, 
	let $\hat p_{i,j} = \hat f(X_{i,j})$ be the estimated density at $X_{i,j}$.
	We then ordered all observations (across every day) according to the estimated density evaluated at each point
	$$
	\hat p_{(1)}<\hat p_{(2)}<\ldots\hat p_{(N_n)}
	$$
	and let $(X_{(1)}, W^\dagger_{(1)}),\ldots, (X_{(N_n)}, W^\dagger_{(N_n)})$
	be the corresponding location and weight from $(X_{i,j}, W^\dagger_{i,j})$.
	Note that $N_n =  \sum_{i=1}^n m_i$
	is total number of observations.
	
	In this way, if we choose $\lambda = p_{(k)}$,
	the level set $\hat{ \mathcal{L}}^*_{p_{(k)}} $
	will cover all observations of indices $(k), (k+2),\ldots, (N_n)$,
	leading to the proportion of $\frac{1}{n}\left(W^\dagger_{(k)}+W^\dagger_{(k+1)}+\ldots W^\dagger_{(N_n)}\right)$
	being covered inside $\hat{ \mathcal{L}}^*_{p_{(k)}} $.
	Note that the division of $n$ is to account for the summation over $n$ days.
	As a result,
	for a given proportion $\rho$,
	we only need to find the largest index $k^*_\rho$
	such that the upper cumulative sum of $t_{(k)}$ is above $\rho$:
	\begin{equation}
		k^*_\rho = \max\left\{k: \frac{1}{n}\sum_{\ell = k}^{N_n} W^\dagger_{(\ell)}\geq \rho\right\}.
		\label{eq::Q3}
	\end{equation}
	With this, the choice $\lambda = p_{(k^*_\rho)}$
	leads to the set 
	\begin{equation}
		\hat{ \mathcal{L}}^*_{p_{(k^*_\rho)}} = \hat Q_\rho.
		\label{eq::Q4}
	\end{equation}

	\section{Simulation}	\label{sec::simulation}
	\subsection{Design of simulations}
	\begin{table}
		\centering
		\begin{tabular}{|c|c|c|c|c|c|}
			\hline
			Pattern                    & Probability            & \begin{tabular}[c]{@{}c@{}}Expected duration\\ (Hours)\end{tabular} & $\eta$ & $q$  & Actions              \\ \hline
			\multirow{7}{*}{Pattern 1} & \multirow{7}{*}{15/28} & 9                                                                  & 0.15   & 0.5  & At home              \\ \cline{3-6} 
			&                        & 0.5                                                                & 0.08   & 0.25 & Home to office       \\ \cline{3-6} 
			&                        & 8                                                                  & 0.15   & 0.5  & At office            \\ \cline{3-6} 
			&                        & 0.6                                                                & 0.08   & 0.25 & Office to home       \\ \cline{3-6} 
			&                        & 2                                                                  & 0.2    & 0.6  & At home              \\ \cline{3-6} 
			&                        & 1                                                                  & 0.06   & 0.15 & Home to park to home \\ \cline{3-6} 
			&                        & 2.9                                                                &        &      & Home                 \\ \hline
			\multirow{9}{*}{Pattern 2} & \multirow{9}{*}{5/28}  & 8.5                                                                & 0.15   & 0.5  & Home                 \\ \cline{3-6} 
			&                        & 0.5                                                                & 0.08   & 0.25 & Home to office       \\ \cline{3-6} 
			&                        & 8                                                                  & 0.15   & 0.5  & At office            \\ \cline{3-6} 
			&                        & 0.75                                                               & 0.08   & 0.25 & Office to restaurant \\ \cline{3-6} 
			&                        & 1                                                                  & 0.08   & 0.25 & At restaurant        \\ \cline{3-6} 
			&                        & 0.4                                                                & 0.08   & 0.25 & Restaurant to home   \\ \cline{3-6} 
			&                        & 0.95                                                               & 0.1    & 0.3  & At home              \\ \cline{3-6} 
			&                        & 1                                                                  & 0.06   & 0.15 & Home to park to home \\ \cline{3-6} 
			&                        & 2.9                                                                &        &      & At home              \\ \hline
			\multirow{5}{*}{Pattern 3} & \multirow{5}{*}{4/28}  & 11                                                                 & 0.3    & 1    & At home              \\ \cline{3-6} 
			&                        & 0.75                                                               & 0.15   & 0.45 & Home to supermarket  \\ \cline{3-6} 
			&                        & 2.5                                                                & 0.3    & 1    & At supermarket       \\ \cline{3-6} 
			&                        & 0.75                                                               & 0.15   & 0.45 & Supermarket to home  \\ \cline{3-6} 
			&                        & 9                                                                  &        &      & At home              \\ \hline
			\multirow{5}{*}{Pattern 4} & \multirow{5}{*}{1/28}  & 10                                                                 & 0.3    & 1    & At home              \\ \cline{3-6} 
			&                        & 0.8                                                                & 0.3    & 0.7  & Home to beach        \\ \cline{3-6} 
			&                        & 5.7                                                                & 0.35   & 1    & Beach                \\ \cline{3-6} 
			&                        & 0.8                                                                & 0.3    & 0.7  & Beach to home        \\ \cline{3-6} 
			&                        & 6.7                                                                &        &      & At home              \\ \hline
			Pattern 5                  & 3/28                   & 24                                                                 &        &      & At home              \\ \hline
		\end{tabular}
		\caption{An SMM model that we use for simulation. Specifically, there are totally five daily activity patterns used for simulation. 
			Each location’s duration follows a truncated normal distribution with mean 
			(Expected Duration), standard deviation \(\eta\), and half-width \(q\). 
			Times shown are in hours. The final row for each pattern (where \(\eta\) and \(q\) are blank) absorbs any remaining time up to 24 hours.}
		\label{Table:pattern_detail}
	\end{table}
	
	To evaluate the performance of the proposed activity density estimators, we construct a simplified scenario involving a single fictitious individual residing in a small-scale world. This world contains five anchor points: home, office, restaurant, beach, and supermarket connected by several routes (Figure \ref{fig::map_pattern}, the first panel).
	
	{\bf Activity patterns.} 
	
	We design five distinct daily activity patterns for this individual, illustrated in the 2nd to 6th panel in Figure \ref{fig::map_pattern}. Specific details including the order of visits, the expected time spent at each location, and the probability of selecting each pattern are summarized in Table \ref{Table:pattern_detail}. Each day, the individual randomly chooses one of these five patterns.

	For each location in a chosen pattern, the duration is drawn from a truncated normal distribution:
	\[
	z_i \mid b \sim \mathcal{N}\bigl(\mu_{b,i},\,\eta_{b,i}^2\bigr) \quad\text{truncated to}\quad \bigl(\mu_{b,i}-q_{b,i},\;\mu_{b,i}+q_{b,i}\bigr),
	\]
	where \(\mu_{b,i}\) and \(\eta_{b,i}\) are the mean and standard deviation, and \(q_{b,i}\in [0,1]\) sets a reasonable range for the duration. The duration of final location in each pattern is computed as whatever time remains in the 24-hour day once all previous locations have been visited.
	
	There are two reasons we use a truncated Gaussian distribution for the duration of each position:
	\begin{itemize}
	\item First, controlling variance: By adjusting \(\eta_{b,i}\), we can reflect realistic scenarios; for example, weekday departures from home tend to be less variable than weekend departures for a beach visit. Hence, we set the variance of duration of the first position higher on weekends.
	\item Second, ensuring a reasonable time distribution: Truncation avoids unrealistic time of visiting for some places (for example, a restaurant stay should not exceed typical operating hours).
	\end{itemize}
	
	Our ultimate goal is to estimate \(f_{GPS}\) and \(f_{GPS,A}\) over the time interval \(A = \bigl[\tfrac{8}{24},\,\tfrac{10}{24}\bigr]\), which corresponds to the morning rush hour between 8~AM and 10~AM.
	

	{\bf Timestamp generation.} For simplicity, we assume that every day has equal number of observations, i.e., $m_i = m$,
	and we choose $(n,m)$ as
	\[
	n \in \{7, 30, 90\} \quad
	m \in \{159, 479, 1439\},
	\]
	such choice of $n$ reflects one week, one month, and three months of GPS observation data. 
	The choices of \(m\) correspond to one observation every 9~minutes, 3~minutes, or 1~minute on average.

	To investigate estimator performance under both regular and irregular frequencies of observations, we generate timestamps in two different ways:
	\begin{enumerate}
		\item \textbf{Even-spacing:} In each simulated day, the \(m\) timestamps are evenly spaced at intervals of \(\tfrac{1440}{m+1}\) minutes. Then, the GPS observations are recorded every \(\frac{1440}{m+1}\) minutes.
		
		
		\item \textbf{Realistic:} To generate non-uniform timestamps which align with the realistic recording, We use real GPS data described in Section~\ref{sec:real-data} in the following procedure:
		\begin{enumerate}
			\item Randomly select one individual from the real dataset.
			\item For each  day, randomly choose one day from that individual.
			\item 
			If the chosen day has fewer than \(m\) real timestamps, generate the additional timestamps by sampling from a Gaussian kernel density estimate of that day's timestamp distribution (using Silverman's rule of thumb for bandwidth).  
			If the chosen day has more than \(m\) real timestamps, randomly remove timestamps until \(m\) timestamps remain, where the probability of removal for each observation is the same.
			
		\end{enumerate}
		This process yields \(m\) timestamps per simulated day, preserving the original non-uniformity found in the selected individual's data. 
		The real data's timestamp distribution is possibly very skew. 
		This procedure preserves the original day-to-day skew and clustering patterns of timestamps found in real data. Figure~\ref{fig:p125_summary} illustrates an example of a highly skewed timestamp distribution for one individual.
		
	\end{enumerate}

	{\bf Measurement errors.}	We introduce measurement errors by drawing from a Gaussian noise distribution with \(\sigma = 0.1\) or \(\sigma = 0.2\).
	The corresponding $\log (f_{GPS}+\xi)$ is shown in Figure~\ref{fig:True density}, where $\xi$ is set as 0.0001.
	Note that we use a log-scale because the anchor locations (home and office) have a very high $f_{GPS}$, which would otherwise dominate the scale so that we cannot see the overall pattern well.
	In the top panels, we display $\log (f_{GPS}+\xi)$ while in the bottom panel,
	we focus on the density during the morning rush hour (8 AM to 10 AM). 
	
	Compared with the real-data analysis in Section~\ref{sec:real-data}, here we choose a smaller \(\xi\). This adjustment is caused by the scale of coordinates. The density value will be lower if the data is in a larger spatial region. This smaller \(\xi\) ensures low-frequency or low-duration regions, such as the route between home and the beach, remain distinguishable from truly unvisited areas in both visualization and the further analysis such as clustering. As a result, in subsequent clustering analyses, the reduced \(\xi\) value helps more clearly separate visited versus unvisited locations, thereby improving the distinctiveness among different mobility patterns.
	
	
	\begin{figure}
		\centering
		\includegraphics[width=0.35\linewidth]{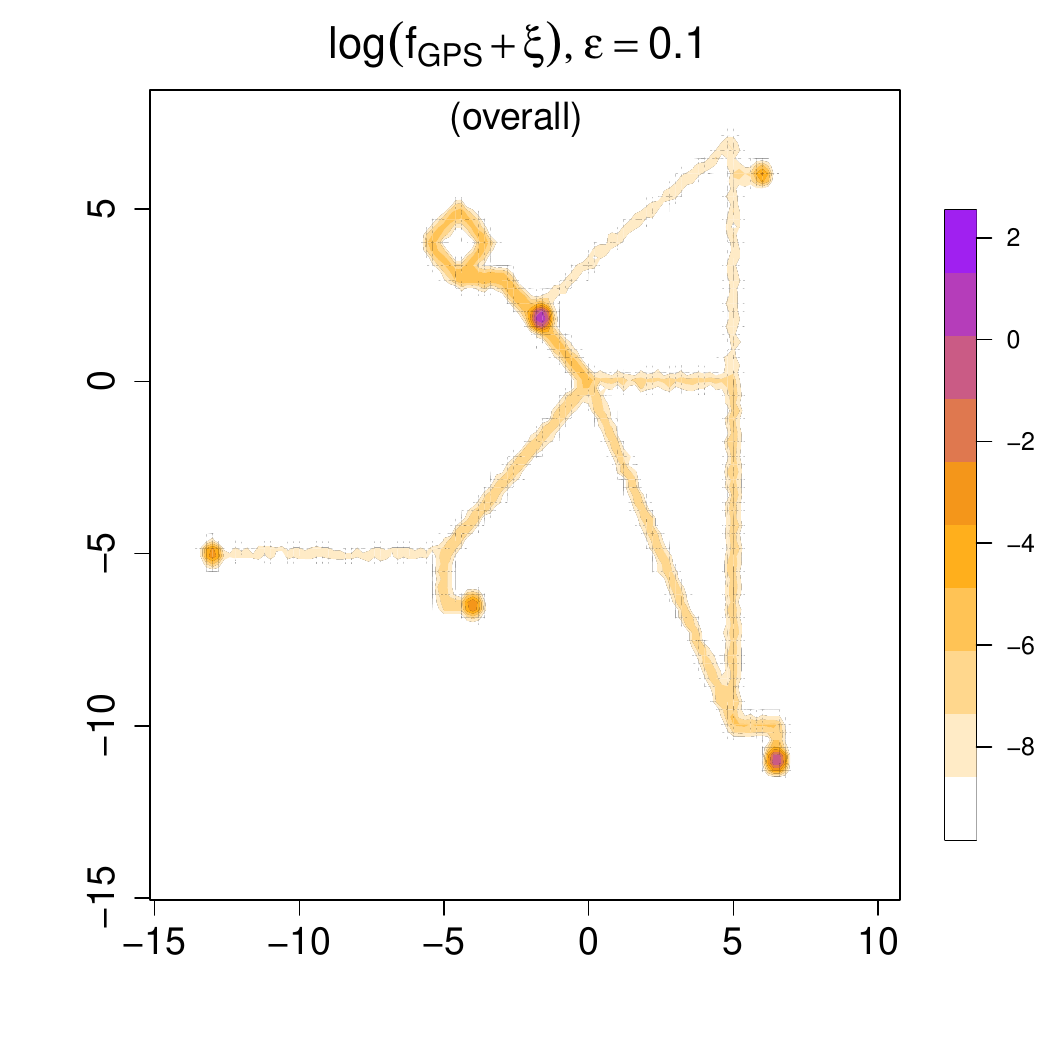}
		\includegraphics[width=0.35\linewidth]{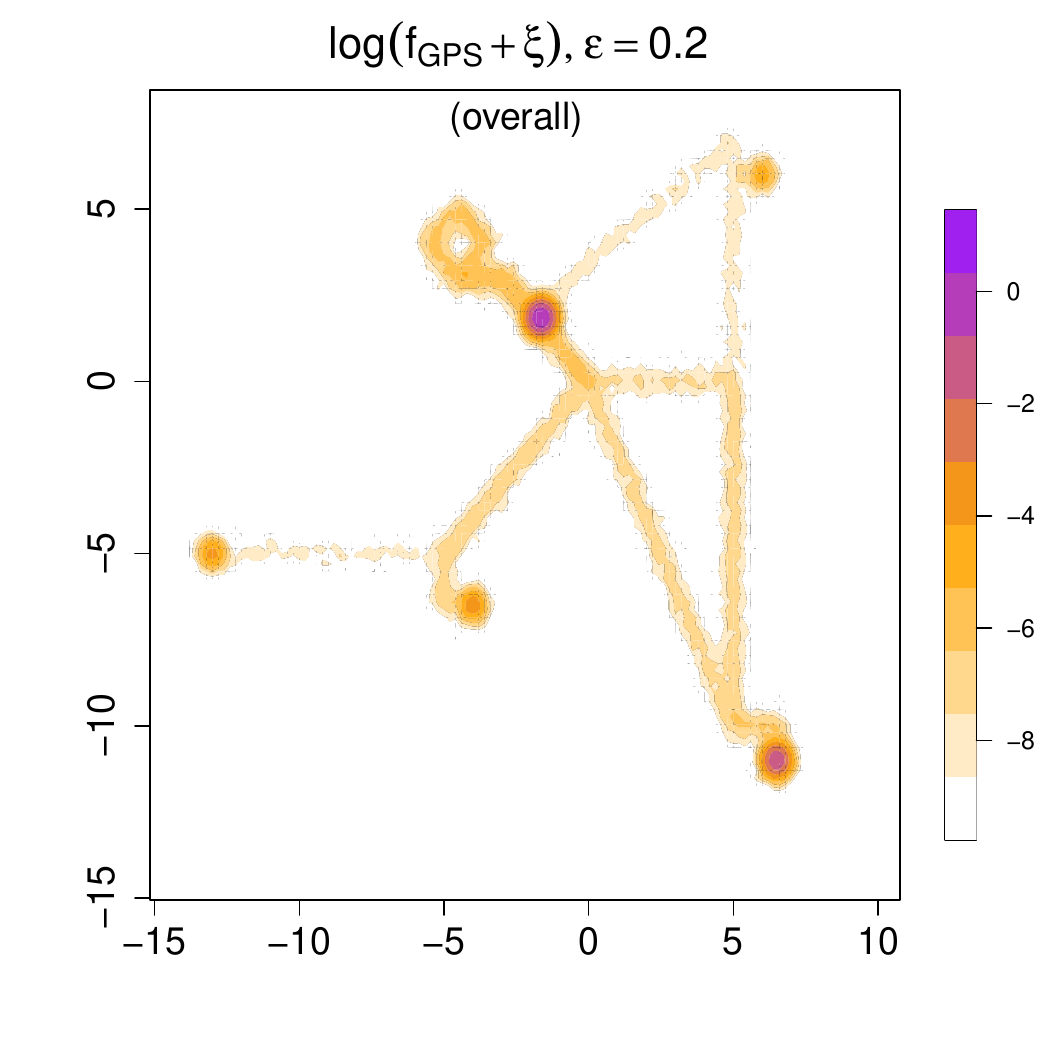}
		\includegraphics[width=0.35\linewidth]{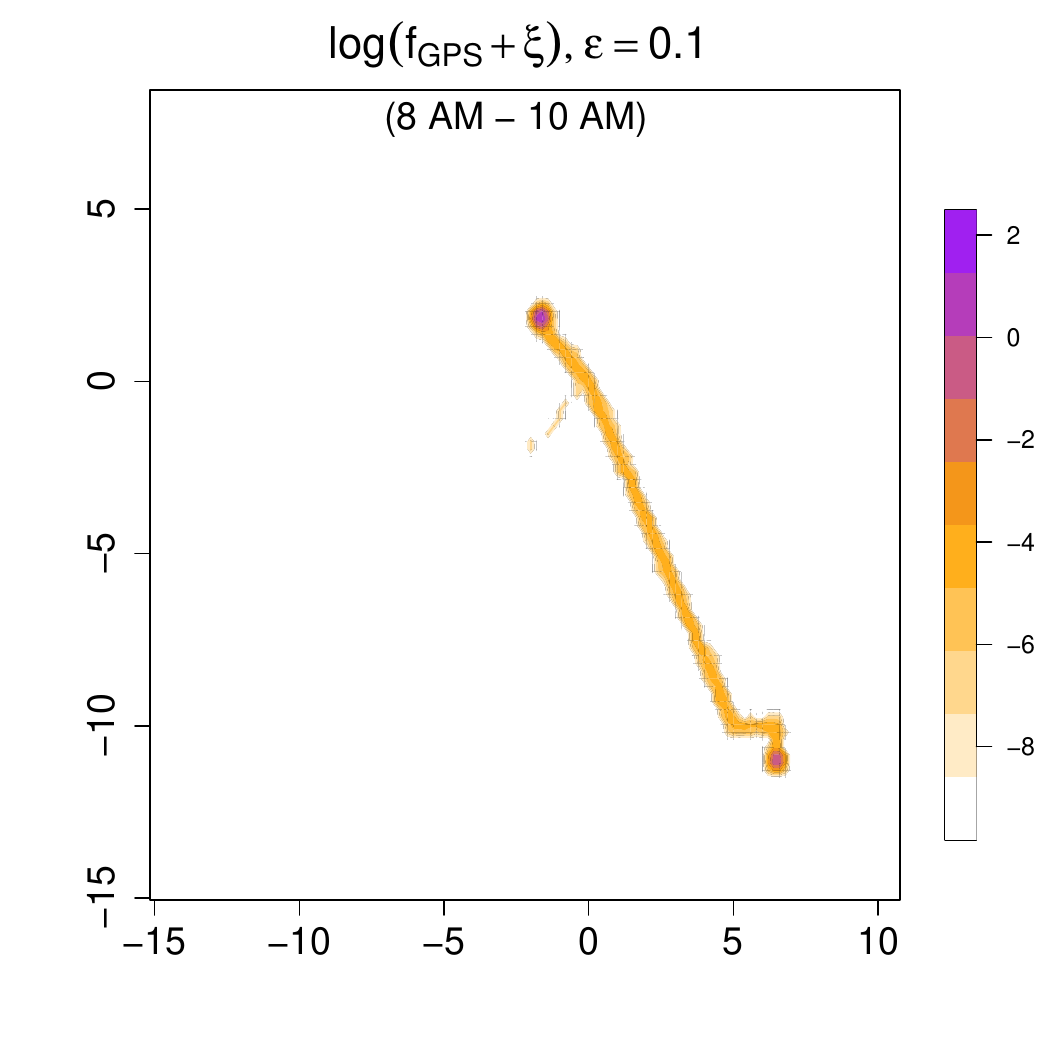}
		\includegraphics[width=0.35\linewidth]{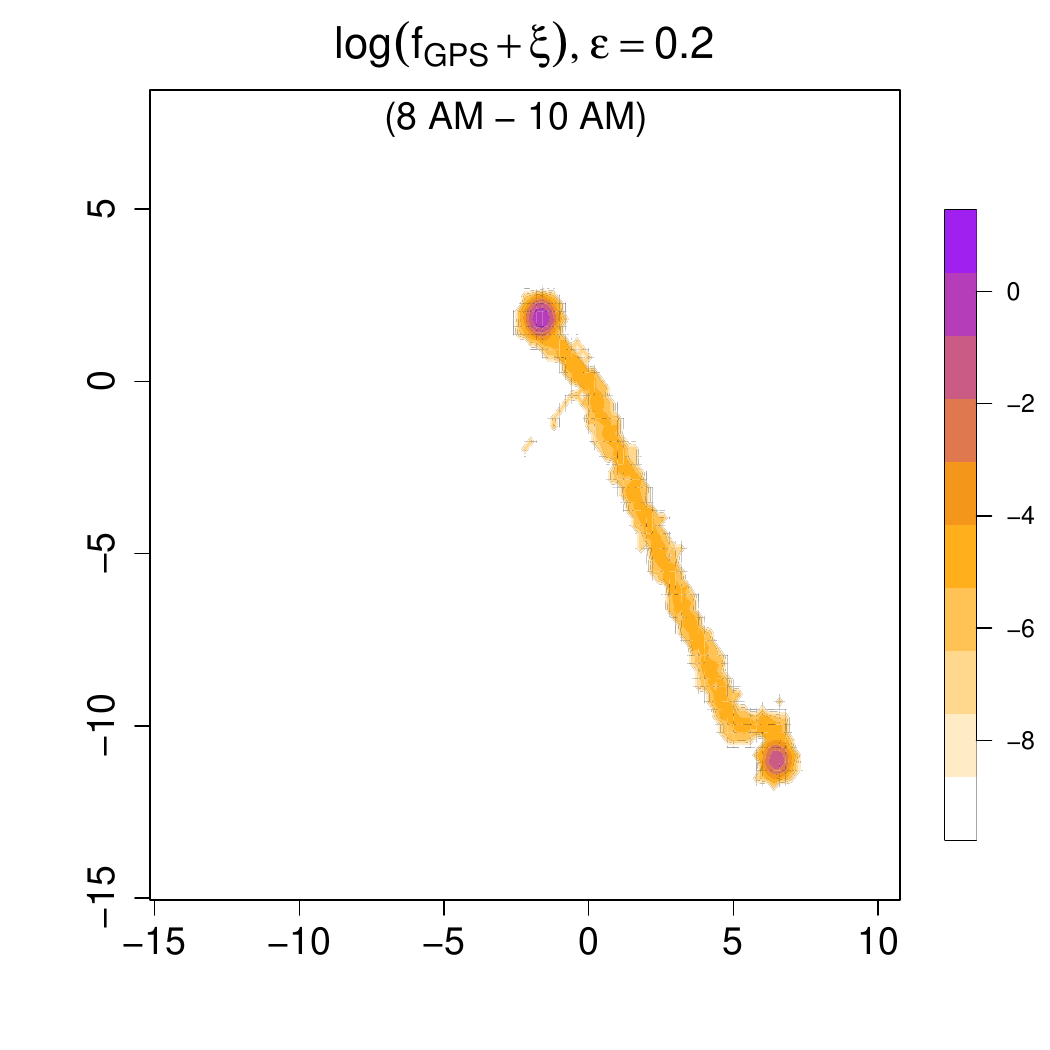}
		\caption{The log-density map for measurement errors under the level of $\sigma=0.1$ (left column) and $\sigma=0.2$ (right column)  for $\log (f_{GPS}+\xi)$ in the top
			and $\log (f_{GPS, A}+\xi)$ in the bottom with $A$ represents the interval of 8 AM to 10 AM. In each panel, $\xi$ is set as 0.0001.}
		\label{fig:True density}
	\end{figure}
	
	
	{\bf Bandwidth selection.}
	After preliminary experiments using mean integrated squared error (MISE) as a performance measure, we found 
	the following reference rule for the spatial and temporal smoothing bandwidth:
	\[
	h_X \;=\; 0.065\,\biggl(\frac{\sqrt{s_{x,1}^2 + s_{x,2}^2}}{\sum_{i=1}^n m_i}\,\biggr)^{\tfrac{1}{6}}, \quad h_T = 0.05\,\biggl(\frac{n}{\sum_{i=1}^n m_i}\,\biggr)^{\tfrac{1}{3}}
	\]
	where 
	\[
	s_{x,l} \;=\; \sqrt{\sum_{k=1}^{n_d}\sum_{i=1}^{m_k} w_{i,j}\,\bigl(x_{i,j,l} - \mu_l\bigr)^2}, l=1,2
	\quad 
	w_{i,j} \;=\; \frac{t_{i,j+1} - t_{i,j-1}}{2\,n},
	\quad
	\mu_1 \;=\; \sum_{i=1}^{n}\sum_{j=1}^{m_i} w_{i,j}\,x_{i,j,l}.
	\]
	Here, \((x_{i,j,1}, x_{i,j,2})\) denotes the spatial coordinates of the \(j\)-th observation on day \(i\), and \(t_{i,j}\) is its (normalized) timestamp within day \(i\). The weights \(w_{k,i}\) ensure proper normalization over the total number of days \(n\) and allow for variable sampling rates within a day.
	
	Compared with standard bandwidth heuristics such as Silverman's rule of thumb for multivariate kernels \cite{Chacon2011}, our scale constants are noticeably smaller. This discrepancy arises because Silverman's rule is often derived under assumptions of (approximately) unimodal or smooth data distributions. In contrast, our simulated environment exhibits strongly localized activity peaks due to prolonged stays at a small number of anchor points (e.g., home and office). As a result, the density near these anchor points is extremely high compared with other regions. A smaller bandwidth helps avoid oversmoothing these sharp peaks, thereby reducing bias and improving local accuracy of the estimated density. From a MISE perspective, capturing these density spikes precisely is critical, since errors around densely occupied points can contribute disproportionately to the integrated error. Consequently, a smaller scale constant not only accommodates the multimodal nature of the data but also minimizes MISE by enhancing resolution in areas with concentrated activity.
	
	

	\subsection{Results on the MISE}

	{\bf Estimating the average GPS density.}
	For each combination of \((n,m)\), we report the averaged mean integrated squared error (MISE) over a \(121\times 99\) grid averaged over 100 repetitions, where each cell has area \(0.2\times 0.2\). The results for estimating \(f_{GPS}\) appear in Tables \ref{Table:daily-1}--\ref{Table:daily-2}.
	
	\begin{table}
		\center
		\begin{tabular}{cccccccc}
			\hline
			\multicolumn{2}{c}{Even-spacing}                    & \multicolumn{3}{c}{$\epsilon=0.2$}                                                                        & \multicolumn{3}{c}{$\epsilon=0.1$}                                                                         \\ \hline
			\multicolumn{1}{c}{$n$} & \multicolumn{1}{c}{$m$}   & \multicolumn{1}{c}{$\hat{f}_{w}$} & \multicolumn{1}{c}{$\hat f_c$} & \multicolumn{1}{c}{$\hat f_{naive}$} & \multicolumn{1}{c}{$\hat {f}_{w}$} & \multicolumn{1}{c}{$\hat f_c$} & \multicolumn{1}{c}{$\hat f_{naive}$} \\ \hline
			\multicolumn{1}{l|}{7}  & \multicolumn{1}{l|}{159}  & 0.0839                            & {\bf 0.0812}                         & \multicolumn{1}{l|}{0.0839}          & 0.8571                             & {\bf 0.8482}                         & 0.8571                               \\
			\multicolumn{1}{l|}{7}  & \multicolumn{1}{l|}{479}  & 0.0619                            & {\bf 0.0599}                         & \multicolumn{1}{l|}{0.0619}          & 0.6426                             & {\bf 0.6353}                         & 0.6426                               \\
			\multicolumn{1}{l|}{7}  & \multicolumn{1}{l|}{1439} & 0.0456                            & {\bf 0.0438}                         & \multicolumn{1}{l|}{0.0456}          & 0.4442                             & {\bf 0.4375}                         & 0.4442                               \\
			\multicolumn{1}{l|}{30} & \multicolumn{1}{l|}{159}  & 0.0473                            & {\bf 0.045}                          & \multicolumn{1}{l|}{0.0473}          & 0.5967                             & {\bf 0.5874}                         & 0.5967                               \\
			\multicolumn{1}{l|}{30} & \multicolumn{1}{l|}{479}  & 0.0359                            & {\bf 0.0341}                         & \multicolumn{1}{l|}{0.0359}          & 0.4321                             & {\bf 0.4249}                         & 0.4321                               \\
			\multicolumn{1}{l|}{30} & \multicolumn{1}{l|}{1439} & 0.0209                            & {\bf 0.0197}                         & \multicolumn{1}{l|}{0.0209}          & 0.268                              &{\bf  0.2623}                         & 0.268                                \\
			\multicolumn{1}{l|}{90} & \multicolumn{1}{l|}{159}  & 0.0368                            & {\bf 0.0344}                         & \multicolumn{1}{l|}{0.0368}          & 0.4592                             & {\bf 0.4495}                         & 0.4592                               \\
			\multicolumn{1}{l|}{90} & \multicolumn{1}{l|}{479}  & 0.0225                            & {\bf 0.021}                          & \multicolumn{1}{l|}{0.0225}          & 0.3109                             & {\bf 0.3041}                         & 0.3109                               \\
			\multicolumn{1}{l|}{90} & \multicolumn{1}{l|}{1439} & 0.0113                            & {\bf 0.0105}                         & \multicolumn{1}{l|}{0.0113}          & 0.1729                             & {\bf 0.1675}                         & 0.1729                               \\ \hline
		\end{tabular}
		\caption{The MISE of the three methods for estimating the daily density during under evenly spacing timestamps.}
		\label{Table:daily-1}
	\end{table}
	
	\begin{table}
		\center
		\begin{tabular}{cccccccc}
			\hline
			\multicolumn{2}{c}{Realistic}                       & \multicolumn{3}{c}{$\epsilon=0.2$}                       & \multicolumn{3}{c}{$\epsilon=0.1$}             \\ \hline
			$n$                     & $m$                       & $\hat{f}_{w}$ & $\hat f_c$ & $\hat f_{naive}$            & $\hat {f}_{w}$ & $\hat f_c$ & $\hat f_{naive}$ \\ \hline
			\multicolumn{1}{c|}{7}  & \multicolumn{1}{c|}{159}  & 0.113         & {\bf0.0844}     & \multicolumn{1}{c|}{0.1457} & 0.9321         & {\bf0.8747}     & 1.0424           \\
    			\multicolumn{1}{c|}{7}  & \multicolumn{1}{c|}{479}  & 0.097         & {\bf0.0622}     & \multicolumn{1}{c|}{0.1192} & 0.7135         & {\bf0.6473 }    & 0.819            \\
			\multicolumn{1}{c|}{7}  & \multicolumn{1}{c|}{1439} & 0.0925        & {\bf0.0489}     & \multicolumn{1}{c|}{0.1017} & 0.5616         & {\bf0.4736}     & 0.642            \\
			\multicolumn{1}{c|}{30} & \multicolumn{1}{c|}{159}  & 0.0708        & {\bf0.0472}     & \multicolumn{1}{c|}{0.0934} & 0.6512         & {\bf0.5873 }    & 0.7415           \\
			\multicolumn{1}{c|}{30} & \multicolumn{1}{c|}{479}  & 0.0575        & {\bf0.0325 }    & \multicolumn{1}{c|}{0.0787} & 0.4859         & {\bf0.4143  }   & 0.5773           \\
			\multicolumn{1}{c|}{30} & \multicolumn{1}{c|}{1439} & 0.0462        &{\bf 0.0219 }    & \multicolumn{1}{c|}{0.0642} & 0.349          & {\bf0.2801  }   & 0.4359           \\
			\multicolumn{1}{c|}{90} & \multicolumn{1}{c|}{159}  & 0.0562        & {\bf0.0337 }    & \multicolumn{1}{c|}{0.0778} & 0.5228         & {\bf0.4529   }  & 0.6108           \\
			\multicolumn{1}{c|}{90} & \multicolumn{1}{c|}{479}  & 0.0415        & {\bf0.0224  }   & \multicolumn{1}{c|}{0.0641} & 0.3734         & {\bf0.3119    } & 0.4677           \\
			\multicolumn{1}{c|}{90} & \multicolumn{1}{c|}{1439} & 0.0286        & {\bf0.0106  }   & \multicolumn{1}{c|}{0.0432} & 0.2358         & {\bf0.1751    } & 0.3172           \\ \hline
		\end{tabular}
		\caption{The MISE of the three methods for estimating the daily density during under realistic timestamp distributions.}
		\label{Table:daily-2}
	\end{table}

	In Tables \ref{Table:daily-1} (even-spacing),
	we see that all methods work well with the conditional method $\hat f_c$
	that performs slightly better than the others. And the marginal method $\hat f_w$ has almost same MISE as naive method $\hat f_{naive}$. 
	Because timestamps are recorded at uniform intervals, each observation naturally represents a comparable portion of the day, thus reducing the penalty from ignoring time differences. A similar conclusion holds if timestamps are generated uniformly in \([0,1]\).
	
	In contrast, when timestamps are drawn from the real dataset, the results change significantly. Real-world GPS records are often non-uniform in time, so the naive estimator which assumes each data point is equally representative suffers from a pronounced bias. This effect is especially evident in the last row of Table~\ref{Table:daily-2}, where \(\hat{f}_{\mathrm{naive}}\) exhibits a substantially higher MISE than either of the proposed estimators. Notably, both \(\hat{f}_w\) and \(\hat{f}_c\) handle irregular sampling much more effectively, resulting in significantly lower MISE.

	{\bf Estimating the interval-specific GPS density.}
	In addition to the average GPS densities,
	we also consider the case of estimating the GPS density during the morning rush hour $A = [\frac{8}{24}, \frac{10}{24}]$ (8 AM - 10 AM).
	Tables \ref{Table:daily-4} and \ref{Table:daily-6} present the results for this interval-specific density.
	
	The performance trends mirror those observed for daily averages (Table~\ref{Table:daily-4}): all three estimators remain consistent, and the conditional method \(\hat{f}_c\) has a slight edge over the other two. The naive and weighted methods \(\hat{f}_w\) show comparable MISE values. 
	
	Table \ref{Table:daily-6} shows the MISE of three estimators when we generate the timestamps from a realistic distribution. As with the daily density estimation, the real (irregular) distribution of timestamps reveals a clear advantage for \(\hat{f}_c\). The naive approach again shows higher MISE, though in this specific morning interval, the difference between \(\hat{f}_w\) and \(\hat{f}_{\mathrm{naive}}\) is smaller than in the full-day scenario. This arises because the distribution of the real data timestamps between 8~AM and 10~AM are not as severely skewed as the distribution of timestamps in the whole day, so the advantage of $\hat f_w$ is somewhat diminished.
	
	Overall, the results confirm that both \(\hat{f}_w\) and \(\hat{f}_c\) more robustly handle uneven timestamp distributions compared to the naive approach, and that \(\hat{f}_c\) typically offers the best performance in terms of MISE across a variety of sampling settings.
	\begin{table}
		\center
		\begin{tabular}{cccccccc}
			\hline
			\multicolumn{2}{c}{Realistic}                       & \multicolumn{3}{c}{$\epsilon=0.2$}                       & \multicolumn{3}{c}{$\epsilon=0.1$}             \\ \hline
			$n$                     & $m$                       & $\hat{f}_{w}$ & $\hat f_c$ & $\hat f_{naive}$            & $\hat {f}_{w}$ & $\hat f_c$ & $\hat f_{naive}$ \\ \hline
			\multicolumn{1}{c|}{7}  & \multicolumn{1}{c|}{159}  & 0.1802        &{\bf 0.1212 }    & \multicolumn{1}{c|}{0.1794} & 1.2959         & {\bf0.8483}     & 1.2935           \\
			\multicolumn{1}{c|}{7}  & \multicolumn{1}{c|}{479}  & 0.1287        & {\bf0.0917  }   & \multicolumn{1}{c|}{0.1281} & 0.9944         & {\bf0.6345 }    & 0.9922           \\
			\multicolumn{1}{c|}{7}  & \multicolumn{1}{c|}{1439} & 0.1184        & {\bf0.0789  }   & \multicolumn{1}{c|}{0.1184} & 0.8261         & {\bf0.4762  }   & 0.8253           \\
			\multicolumn{1}{c|}{30} & \multicolumn{1}{c|}{159}  & 0.0971        & {\bf0.0676   }  & \multicolumn{1}{c|}{0.0964} & 0.864          & {\bf0.5967    } & 0.8615           \\
			\multicolumn{1}{c|}{30} & \multicolumn{1}{c|}{479}  & 0.0787        & {\bf0.0548   }  & \multicolumn{1}{c|}{0.0783} & 0.6974         & {\bf0.436 }     & 0.6955           \\
			\multicolumn{1}{c|}{30} & \multicolumn{1}{c|}{1439} & 0.0498        & {\bf0.0333   }  & \multicolumn{1}{c|}{0.0497} & 0.4488         & {\bf0.2693}     & 0.4482           \\
			\multicolumn{1}{c|}{90} & \multicolumn{1}{c|}{159}  & 0.0777        & {\bf0.0545    } & \multicolumn{1}{c|}{0.0772} & 0.7083         & {\bf0.4752  }   & 0.7058           \\
			\multicolumn{1}{c|}{90} & \multicolumn{1}{c|}{479}  & 0.049         & {\bf0.0341    } & \multicolumn{1}{c|}{0.0489} & 0.4616         & {\bf0.3144    } & 0.46             \\
			\multicolumn{1}{c|}{90} & \multicolumn{1}{c|}{1439} & 0.038         & {\bf0.0204    } & \multicolumn{1}{c|}{0.0379} & 0.3093         & {\bf0.1817    } & 0.3089           \\ \hline
		\end{tabular}
		\caption{The MISE of the three methods for estimating the density during the rush hour (8-10 AM) under evenly spacing timestamps.}
		\label{Table:daily-4}
	\end{table}
	
	\begin{table}
		\center
		\begin{tabular}{cccccccc}
			\hline
			\multicolumn{2}{c}{Realistic}                       & \multicolumn{3}{c}{$\epsilon=0.2$}                       & \multicolumn{3}{c}{$\epsilon=0.1$}             \\ \hline
			$n$                     & $m$                       & $\hat{f}_{w}$ & $\hat f_c$ & $\hat f_{naive}$            & $\hat {f}_{w}$ & $\hat f_c$ & $\hat f_{naive}$ \\ \hline
			\multicolumn{1}{c|}{7}  & \multicolumn{1}{c|}{159}  & 0.1489        & {\bf0.1099}     & \multicolumn{1}{c|}{0.1522} & 1.0886         & {\bf0.816 }     & 1.1085           \\
			\multicolumn{1}{c|}{7}  & \multicolumn{1}{c|}{479}  & 0.1064        & {\bf0.0812 }    & \multicolumn{1}{c|}{0.1123} & 0.8289         & {\bf0.582  }    & 0.8495           \\
			\multicolumn{1}{c|}{7}  & \multicolumn{1}{c|}{1439} & 0.0909        & {\bf0.0749 }    & \multicolumn{1}{c|}{0.1006} & 0.6751         & {\bf0.4664}     & 0.7054           \\
			\multicolumn{1}{c|}{30} & \multicolumn{1}{c|}{159}  & 0.0894        & {\bf0.065   }   & \multicolumn{1}{c|}{0.0902} & 0.7808         & {\bf0.5702  }   & 0.7857           \\
			\multicolumn{1}{c|}{30} & \multicolumn{1}{c|}{479}  & 0.0637        & {\bf0.0463  }   & \multicolumn{1}{c|}{0.0651} & 0.5714         & {\bf0.3954  }   & 0.5776           \\
			\multicolumn{1}{c|}{30} & \multicolumn{1}{c|}{1439} & 0.0542        & {\bf0.0424  }   & \multicolumn{1}{c|}{0.0573} & 0.4103         & {\bf0.3121  }   & 0.4225           \\
			\multicolumn{1}{c|}{90} & \multicolumn{1}{c|}{159}  & 0.0616        & {\bf0.0482   }  & \multicolumn{1}{c|}{0.0627} & 0.5631         & {\bf0.4508    } & 0.568            \\
			\multicolumn{1}{c|}{90} & \multicolumn{1}{c|}{479}  & 0.046         & {\bf0.0357    } & \multicolumn{1}{c|}{0.0486} & 0.3906         & {\bf0.3136     }& 0.4013           \\
			\multicolumn{1}{c|}{90} & \multicolumn{1}{c|}{1439} & 0.0286        & {\bf0.0162   }  & \multicolumn{1}{c|}{0.0304} & 0.2301         & {\bf0.1587  }   & 0.2373           \\ \hline
		\end{tabular}
		
		\caption{The MISE of the three methods for estimating the density during the rush hour (8-10 AM) under the realistic timestamp distribution.}
		\label{Table:daily-6}
	\end{table}

	\subsection{Applications based on the simple movement model}
	
	In our simulation setup, Pattern~1 and Pattern~2 correspond to typical weekday schedules (containing the office as an anchor point), whereas Pattern~3, Pattern~4, and Pattern~5 represent potential weekend routines (not containing office as anchor point). 
	We separately analyze the data in weekdays and weekends using the framework of Section~\ref{sec::smm}.
	Throughout this subsection, we fix the measurement error at \(\sigma=0.2\) and assume evenly spaced timestamps with 479 GPS observations per day. We generate a total of 90 simulated days under these conditions.

	\subsubsection{Anchor location recovering}
	From Table~\ref{Table:pattern_detail}, we identify three anchor locations for weekdays (home, office, restaurant) and another set of three for weekends (home, supermarket, beach). Across the entire 90 days, these yield five un anchor points overall. 
	
	Using the method described in Section~\ref{sec::AL} with \(\lambda = 0.0055\), we recover all anchor locations accurately, whether we consider:
	
	1. All 90 days combined,  
	
	2. Weekdays only (Patterns~1 and~2), or 
	
	3. Weekends only (Patterns~3,~4, and~5).
	
	Figure~\ref{fig:sim-anchor points} compares the true anchor points with the detected ones in each of these scenarios. Under a suitably chosen threshold, and assuming sufficient time spent at each anchor point, the technique in Section~\ref{sec::AL} locates the simulated individual's anchor sites with high accuracy.
	
	\begin{figure}
		\centering
		\includegraphics[width=0.25\linewidth]{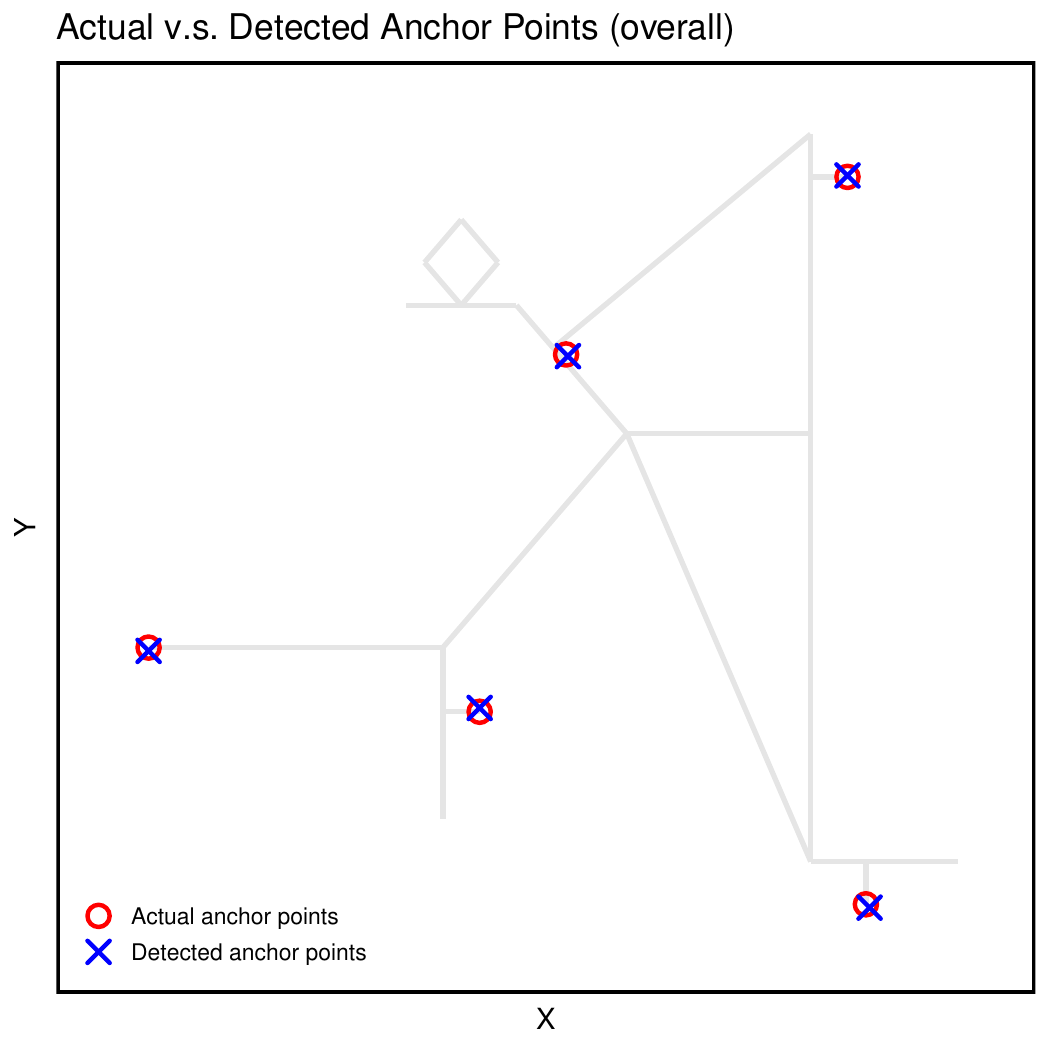}
		\includegraphics[width=0.25\linewidth]{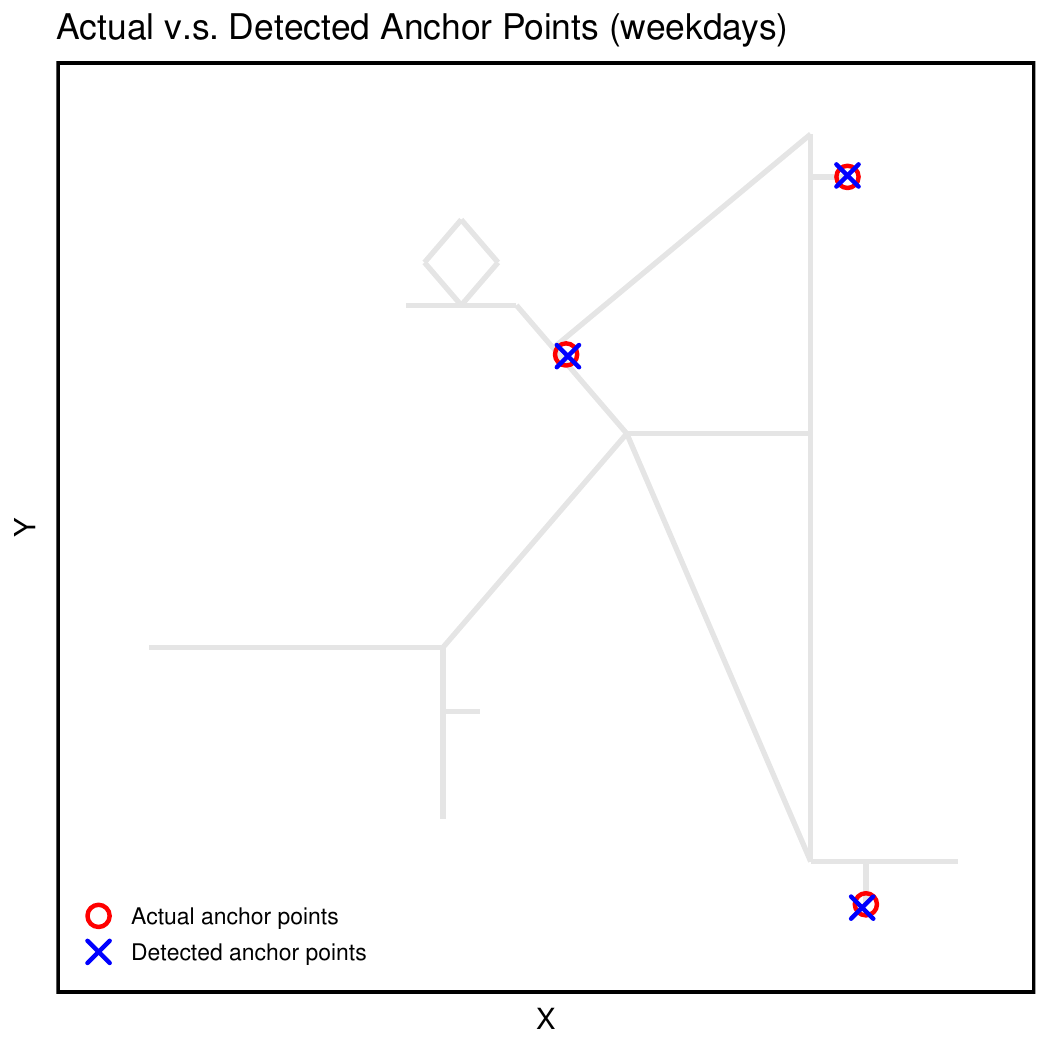}
		\includegraphics[width=0.25\linewidth]{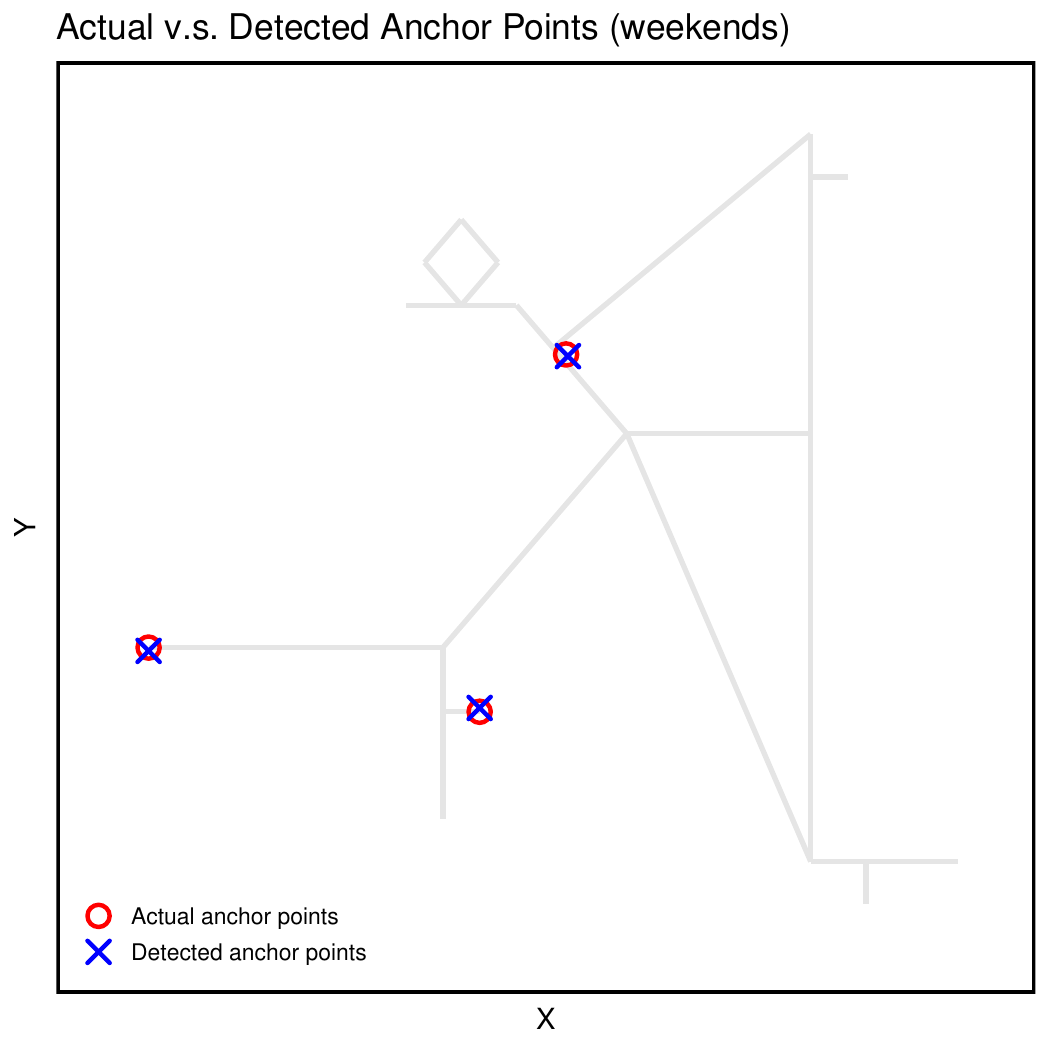}
		\caption{Detected anchor points versus true anchor points for the simulated individual. 
			\textbf{Left}: All days combined. 
			\textbf{Center}: Weekdays only. 
			\textbf{Right}: Weekends only.}
		\label{fig:sim-anchor points}
	\end{figure}
	
	\subsubsection{Analysis on activities}
	Figure~\ref{fig:sim-activity space} depicts the contour region \(Q_{\rho}\), which encloses the spatial area where the individual spends a proportion \(\rho\) of their time. We show contours for \(\rho \in \{0.999, 0.99, 0.90, 0.70, 0.50\}\) in three scenarios:
	\begin{enumerate}
		\item All Days (Left Panel): Over 90~days, 90\% of the individual's time (orange contour) is concentrated at home and office, consistent with Patterns~1 and~2 being selected more frequently.
		\item Weekdays (Middle Panel): On weekdays alone, the same 90\% region again focuses heavily on home and office, indicating most weekday hours are spent in these two anchor locations.
		\item Weekends (Right Panel): On weekends (Patterns~3 to 5), more than 90\% of time is spent at home, reflecting fewer trips away from the house (with occasional visits to the supermarket or beach).
	\end{enumerate}	
	
	\begin{figure}
		\centering
		\includegraphics[width=0.25\linewidth]{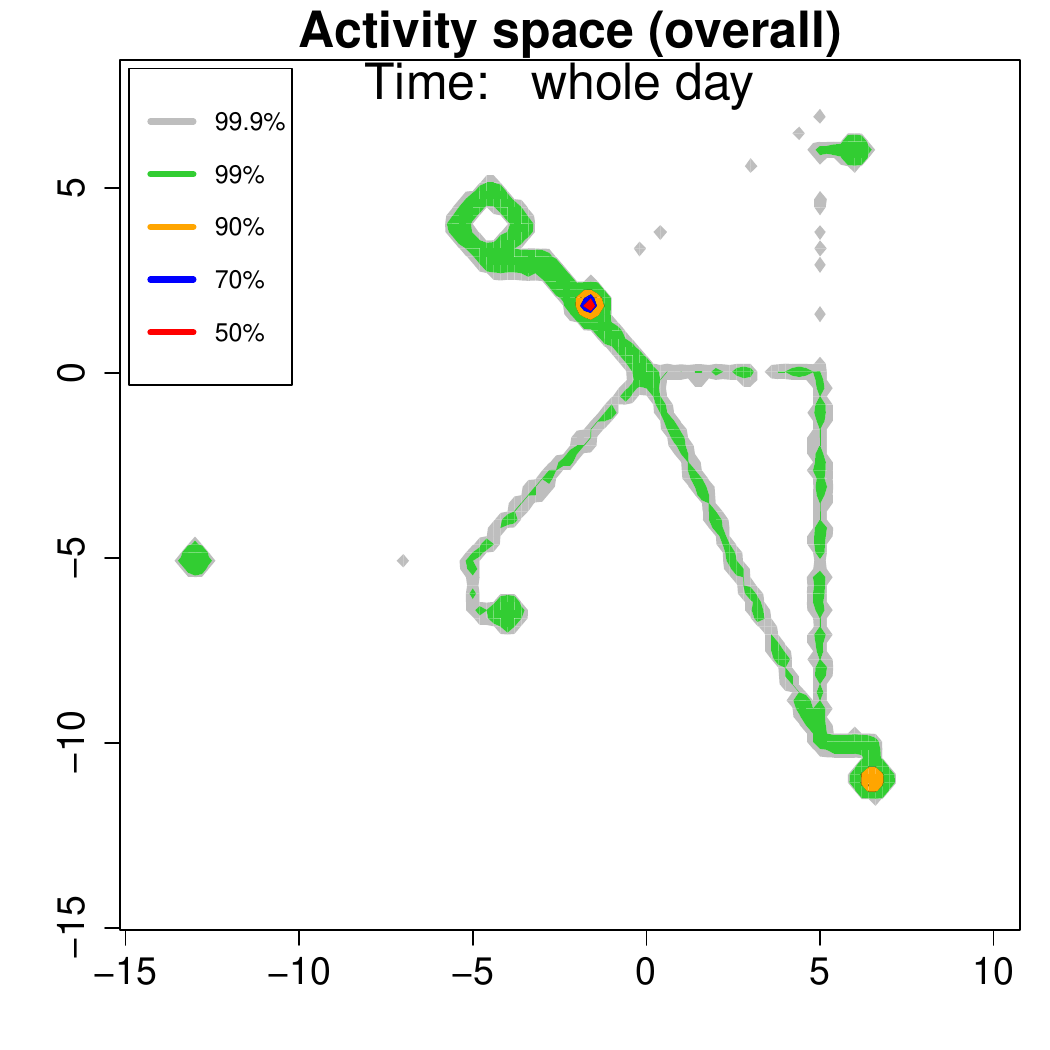}
		\includegraphics[width=0.25\linewidth]{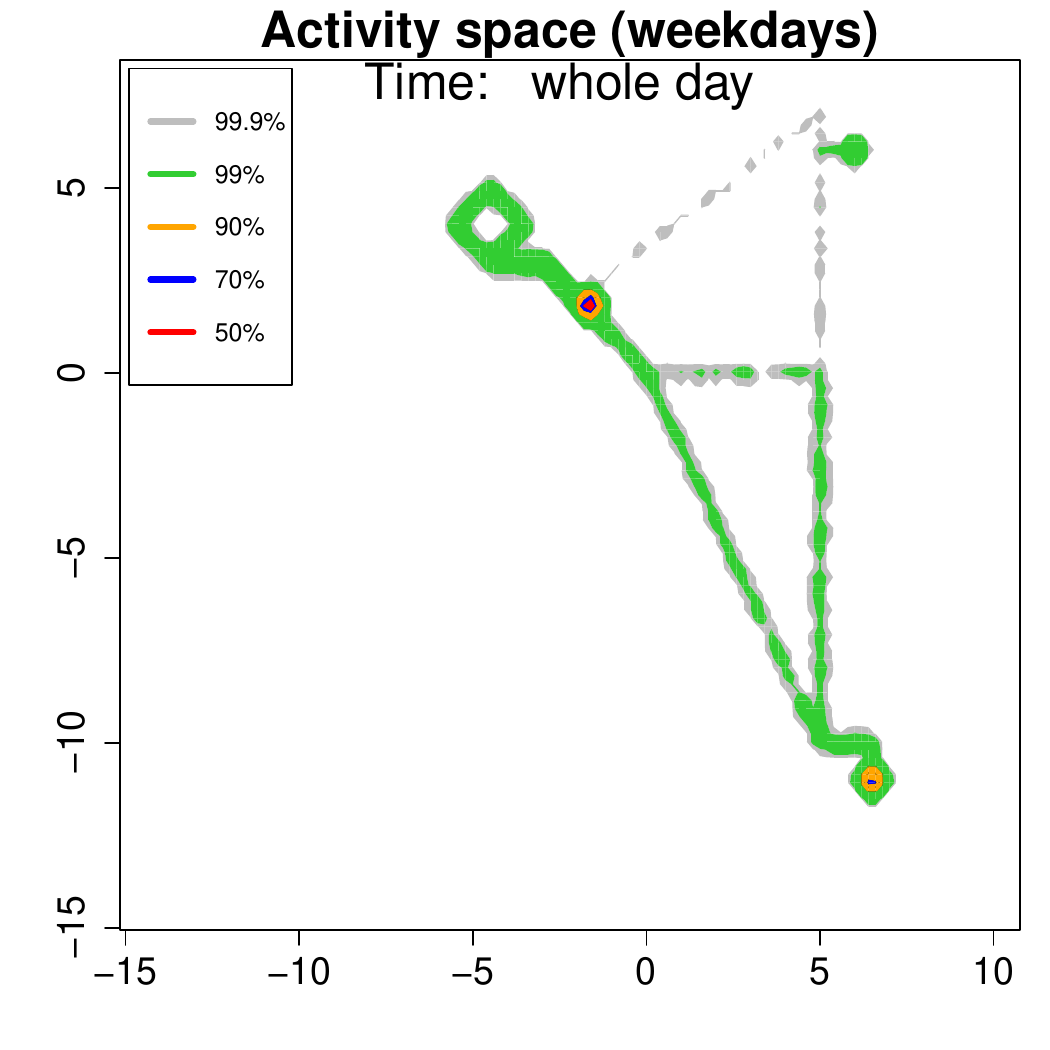}
		\includegraphics[width=0.25\linewidth]{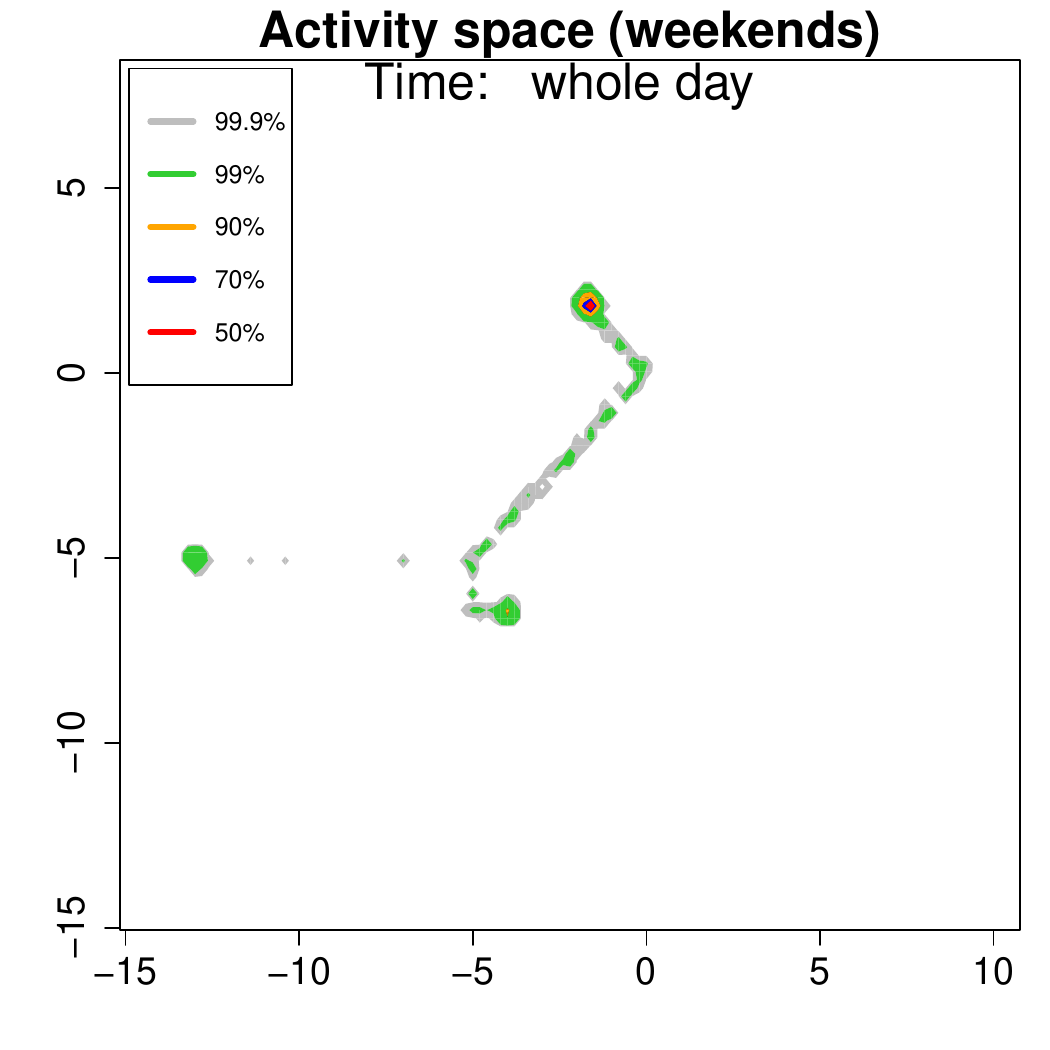}
		\caption{Activity space of the simulated individual for \textbf{all days} (left), \textbf{weekdays} (middle), and \textbf{weekends} (right). Each panel shows contour regions for 
			99.9\% (gray), 99\% (green), 90\% (orange), 70\% (blue), and 50\% (red) of total time spent.
		}
		\label{fig:sim-activity space}
	\end{figure}

	\subsubsection{Clustering of trajectories}	
	\textbf{Clustering of days.} A key aspect of our study is to cluster the simulated individual's weekday days (Patterns 1 and 2) and weekend days (Patterns 3, 4, and 5) based on their GPS trajectories. To do so, we separately analyze the weekday observations and the weekend observations.
	
	For each day, we compute the smoothed conditional density \(\hat{f}_{c}\) (using the bandwidths \(h_X\) and \(h_T\) determined earlier) and measure pairwise distances via the integrated squared difference between log-densities:
	\[
	D_{ij} 
	\;=\;
	\int \Bigl[\log\bigl(\hat f_{c,i}(x)+\xi\bigr) 
	\;-\; 
	\log\bigl(\hat f_{c,j}(x)+\xi\bigr)\Bigr]^{2}\,dx,
	\]
	where \(\hat f_{c,i}\) is the estimated conditional density for day \(i\), and \(\hat f_{c,j}\) for day \(j\). Same as the log-density visualization in Figure~\ref{fig:True density}, here, $\xi$ is set as 0.0001.
	

	 Figure~\ref{fig:sim-affinity} further illustrates the structure of the clusters by showing the re-ordered distance matrices for weekdays (left panel) and weekends (right panel). In both cases, we see a pronounced block structure, indicating clear separations between the different activity patterns measured based on our proposed density estimation method.
	\begin{figure}
		\centering
		\includegraphics[width=0.25\linewidth]{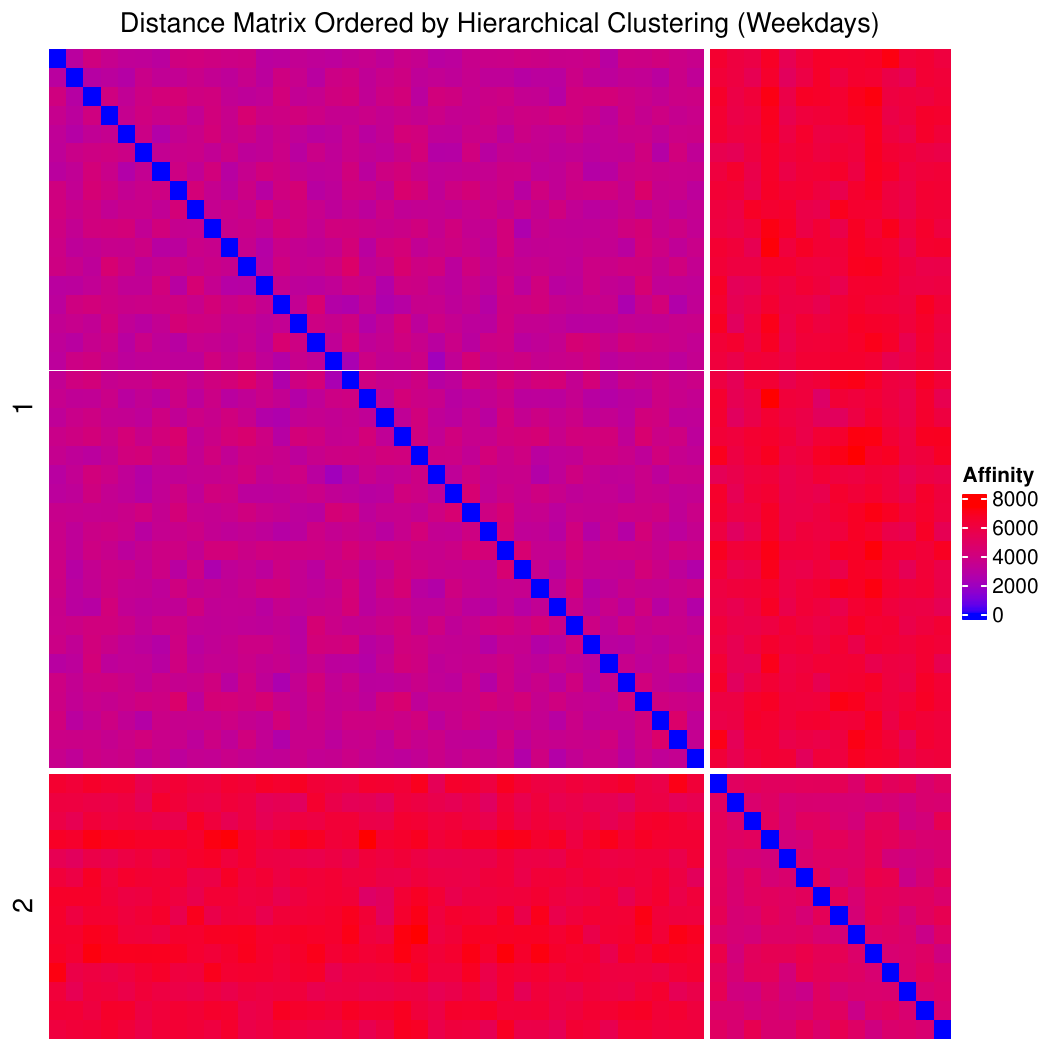}
		\includegraphics[width=0.25\linewidth]{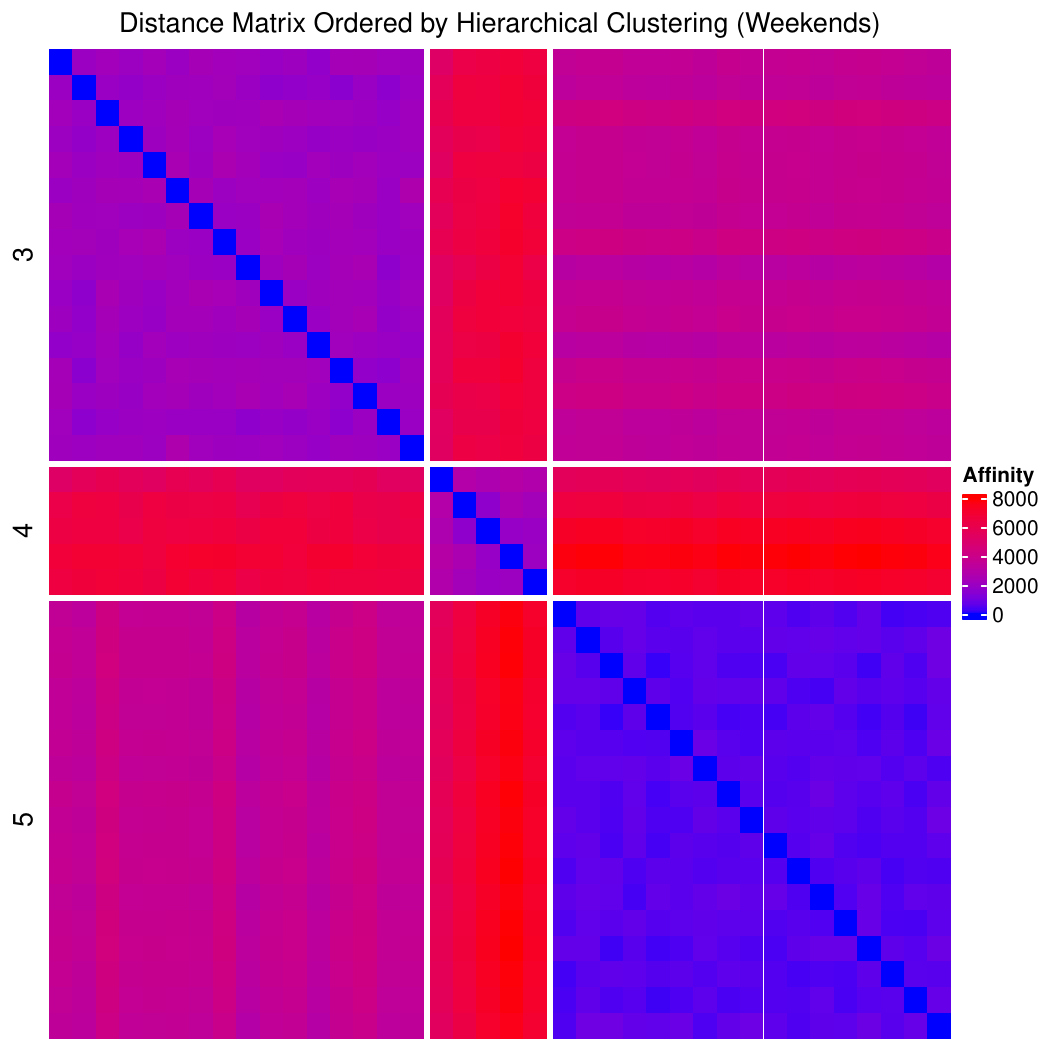}

		\caption{
			Re-ordered distance matrices by hierarchical clustering for \textbf{weekday} observations (left) and \textbf{weekend} observations (right). Each block corresponds to one of the simulated patterns, reflecting clear clustering boundaries.}
		\label{fig:sim-affinity}
	\end{figure}

		We perform single-linkage hierarchical clustering, producing dendrograms of weekday and weekend data (Figure~\ref{fig:sim-dendrogram}). After comparing with the known activity pattern for each day, we find that hierarchical clustering achieves 100\% accurate day-level classifications, perfectly matching the original activity patterns used to generate the data.
	\begin{figure}
		\centering
		\includegraphics[width=0.25\linewidth]{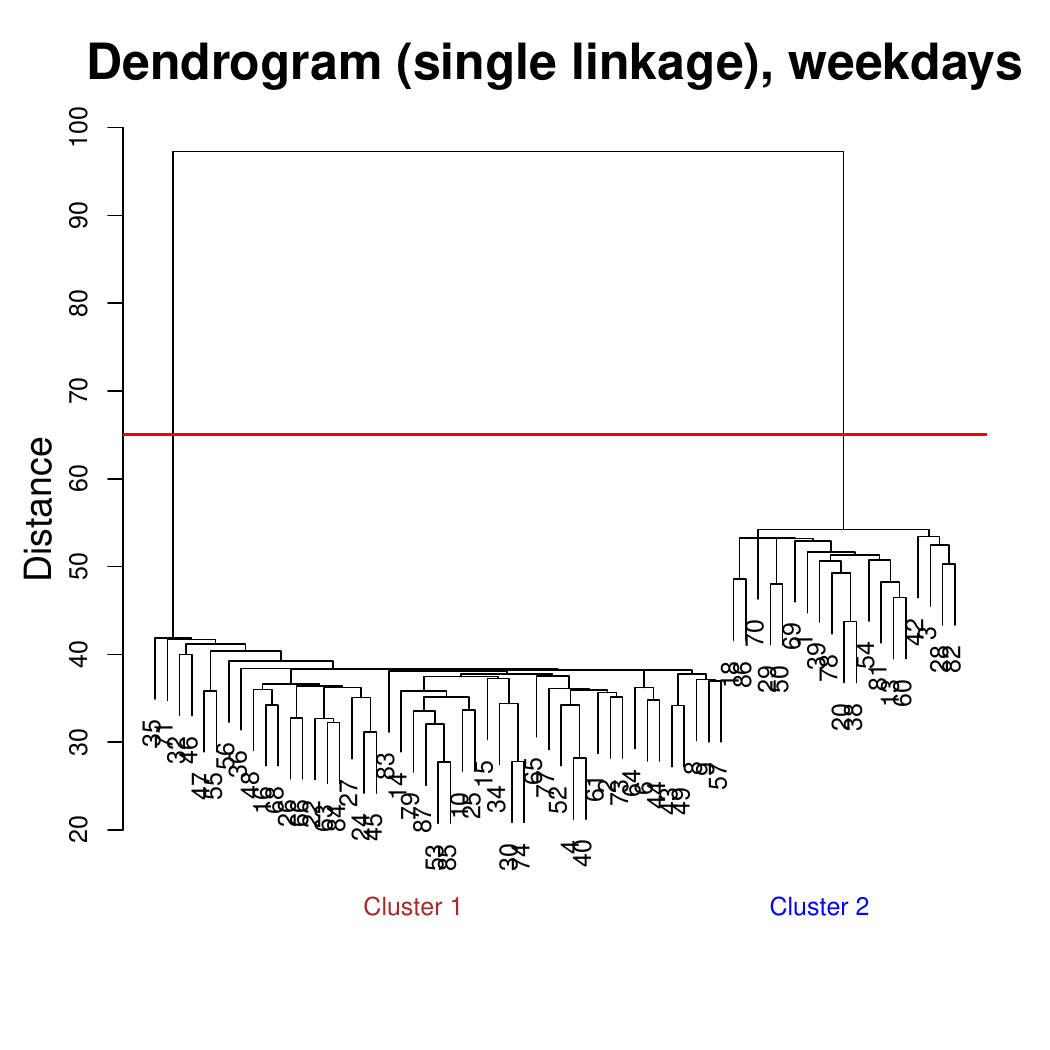}
		\includegraphics[width=0.25\linewidth]{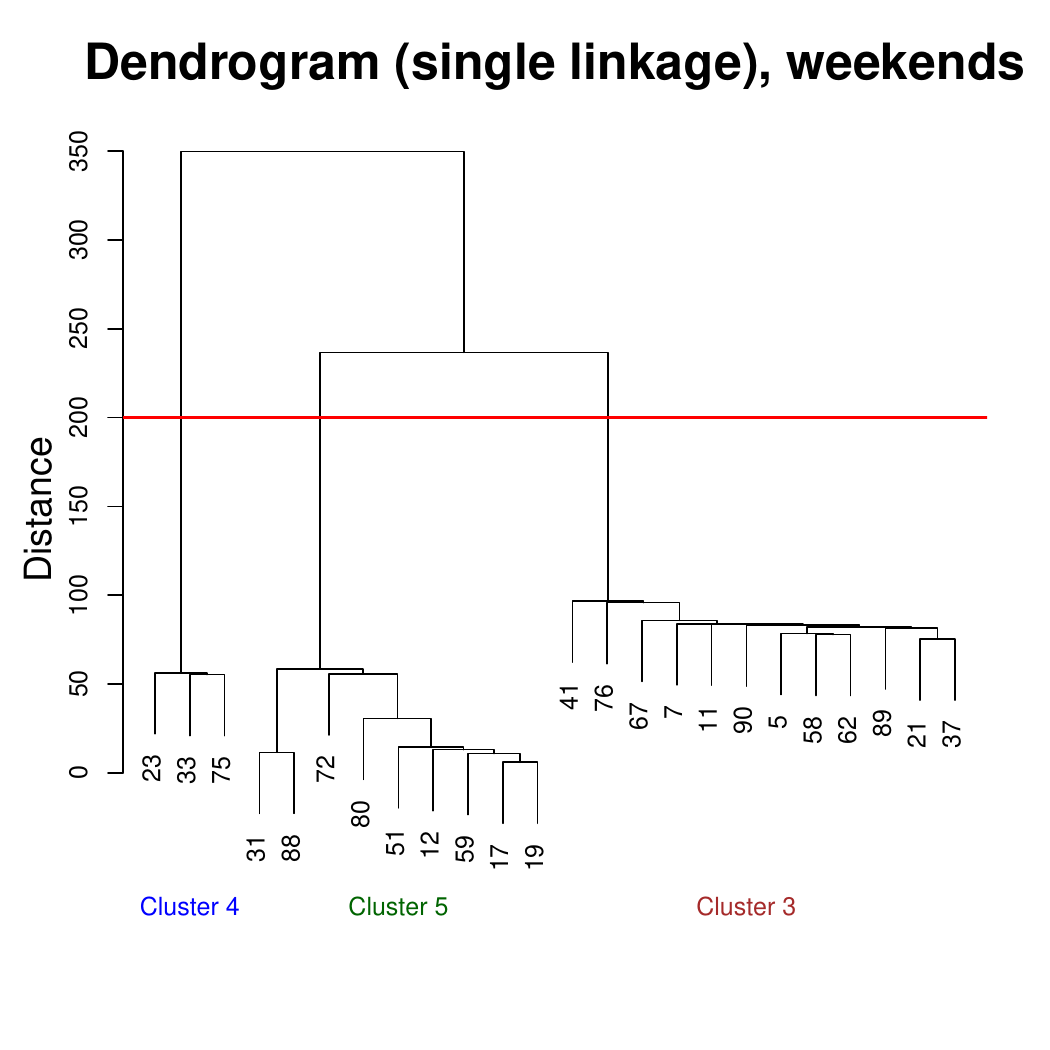}
		\caption{Dendrograms for \textbf{weekday} observations (left) and \textbf{weekend} observations (right). Each branch corresponds to a distinct activity pattern.}
		\label{fig:sim-dendrogram}
	\end{figure}

	\textbf{Conditional density on clustered days.} 
	To further illustrate how the discovered clusters differ in their spatiotemporal patterns, we compare conditional densities at specific times of day.
	
	Figure~\ref{fig:sim-con density-weekdays} shows the estimated conditional density at 6:00~PM for each weekday cluster (i.e., Patterns~1 and~2). The densities differ markedly: in Pattern 1, the individual leaves the office and goes directly home, whereas in Pattern 2, they stop at a restaurant (top-right anchor) before returning home.
	\begin{figure}
		\centering
		\includegraphics[width=0.25\linewidth]{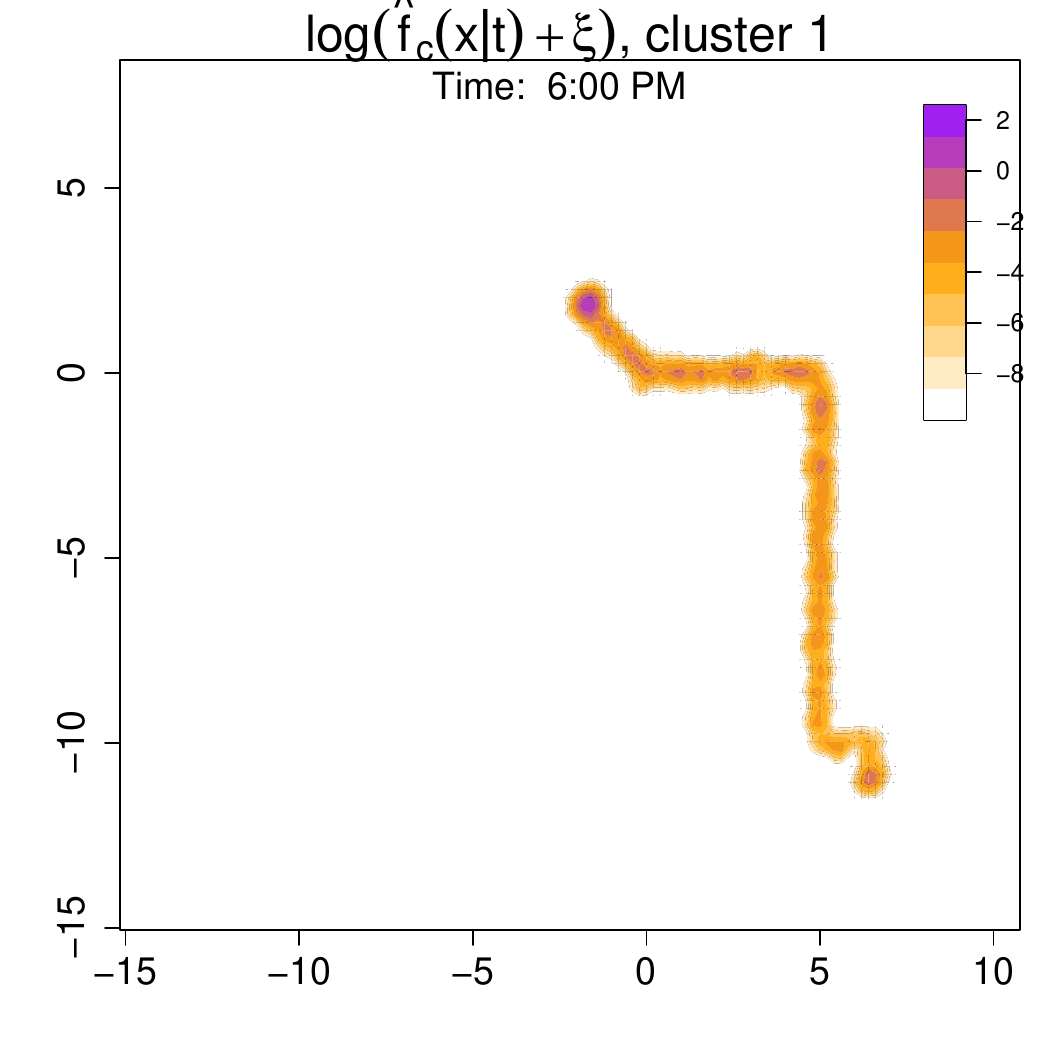}
		\includegraphics[width=0.25\linewidth]{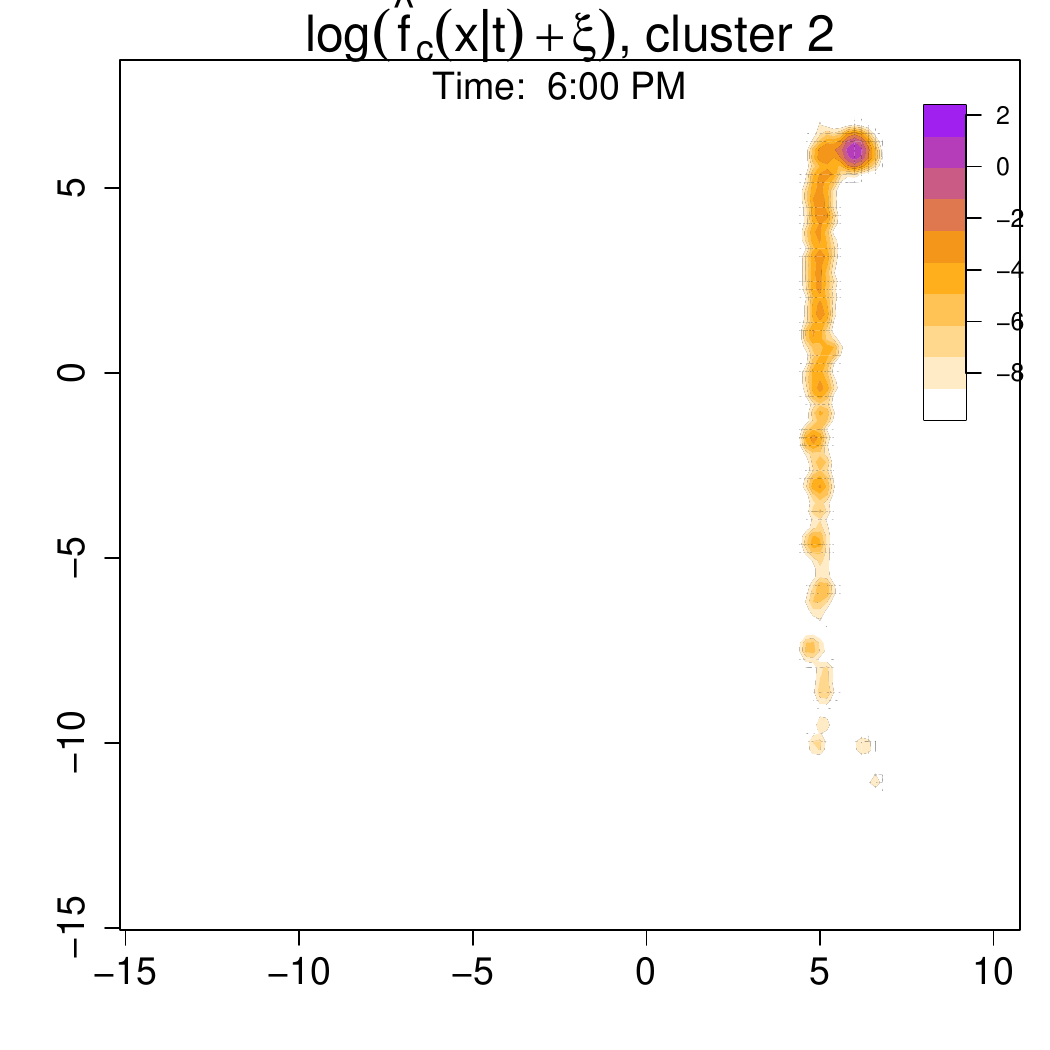}
		\caption{Estimated conditional density at \textbf{6:00~PM} for the two weekday clusters (Patterns 1 and 2). The cluster in the left panel shows a high density around home, while the right panel indicates a high density near the restaurant.}
		\label{fig:sim-con density-weekdays}
	\end{figure}
	Figure~\ref{fig:sim-con density-weekends} displays the conditional density at 12:00~PM for each weekend cluster. In Pattern 3, the individual is often near the supermarket around noon; in Pattern 4, they might be found at the beach; and in Pattern 5, they remain at home throughout the day.
	
	\begin{figure}
		\centering
		\includegraphics[width=0.25\linewidth]{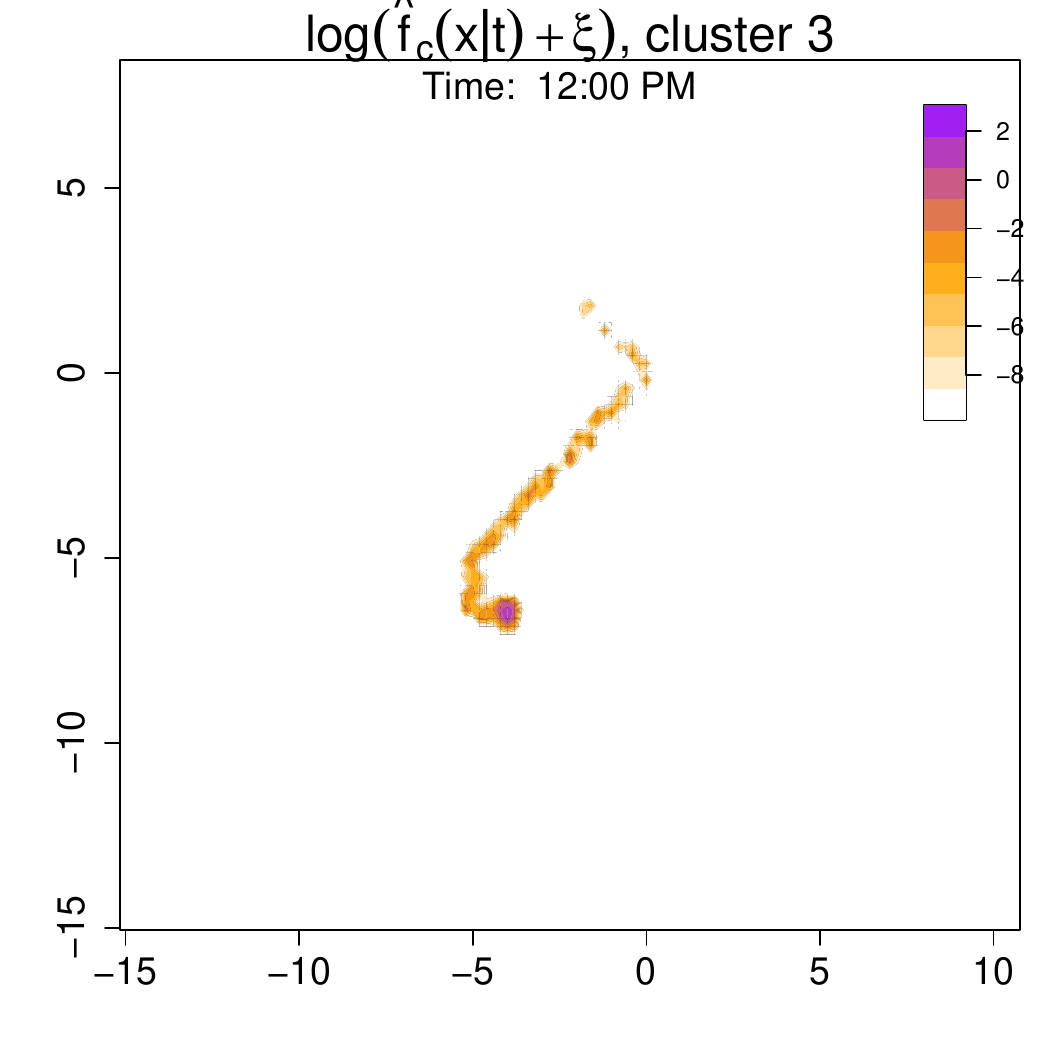}
		\includegraphics[width=0.25\linewidth]{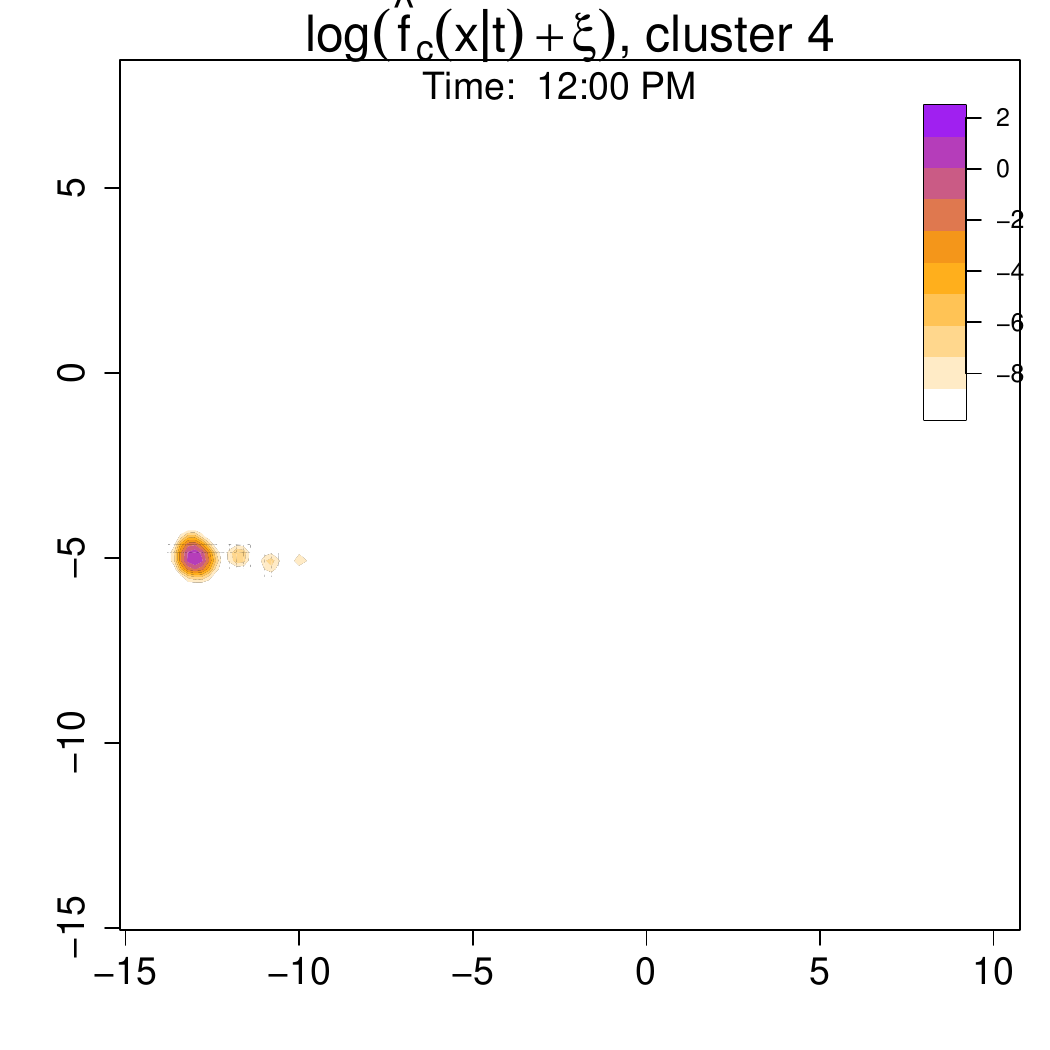}
		\includegraphics[width=0.25\linewidth]{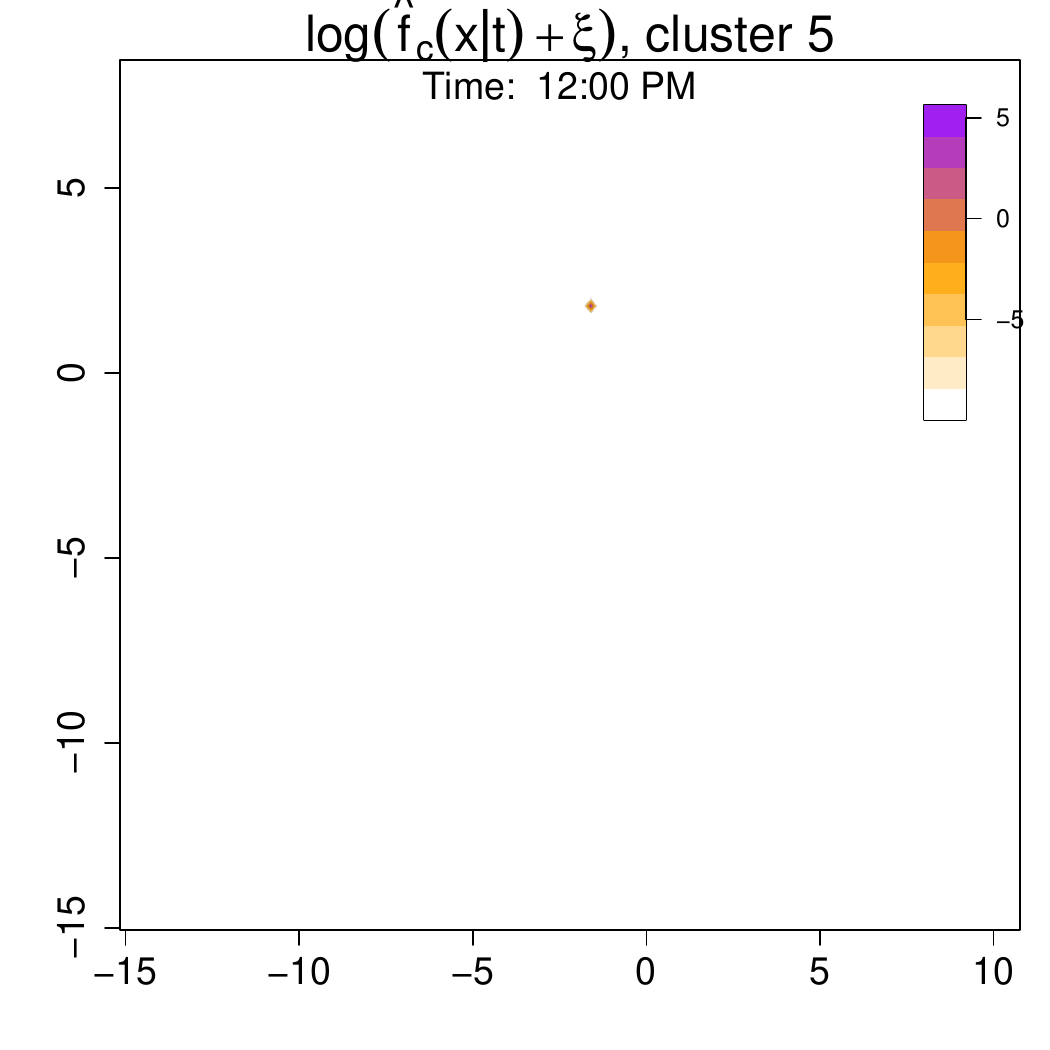}
		\caption{Estimated conditional density at \textbf{12:00~PM} for the weekend clusters (Patterns 3, 4, 5). From left to right: frequent midday location near the supermarket, near the beach, and at home.}
		\label{fig:sim-con density-weekends}
	\end{figure}
	
	\textbf{Centers on each cluster.}
	Given these day-level clusters, we also estimate a cluster-specific center \(\hat \mu_{g}(t)\) representing the average location for each cluster at time \(t\). Figure~\ref{fig:sim-centers} shows the computed centers at six time points (from 6:00~AM to 9:00~PM) for each cluster. The first (top-left) panel shows the combined data (all days), and the remaining panels correspond to each of the five clusters.
	
	\begin{figure}
		\centering
		\includegraphics[width=0.25\linewidth]{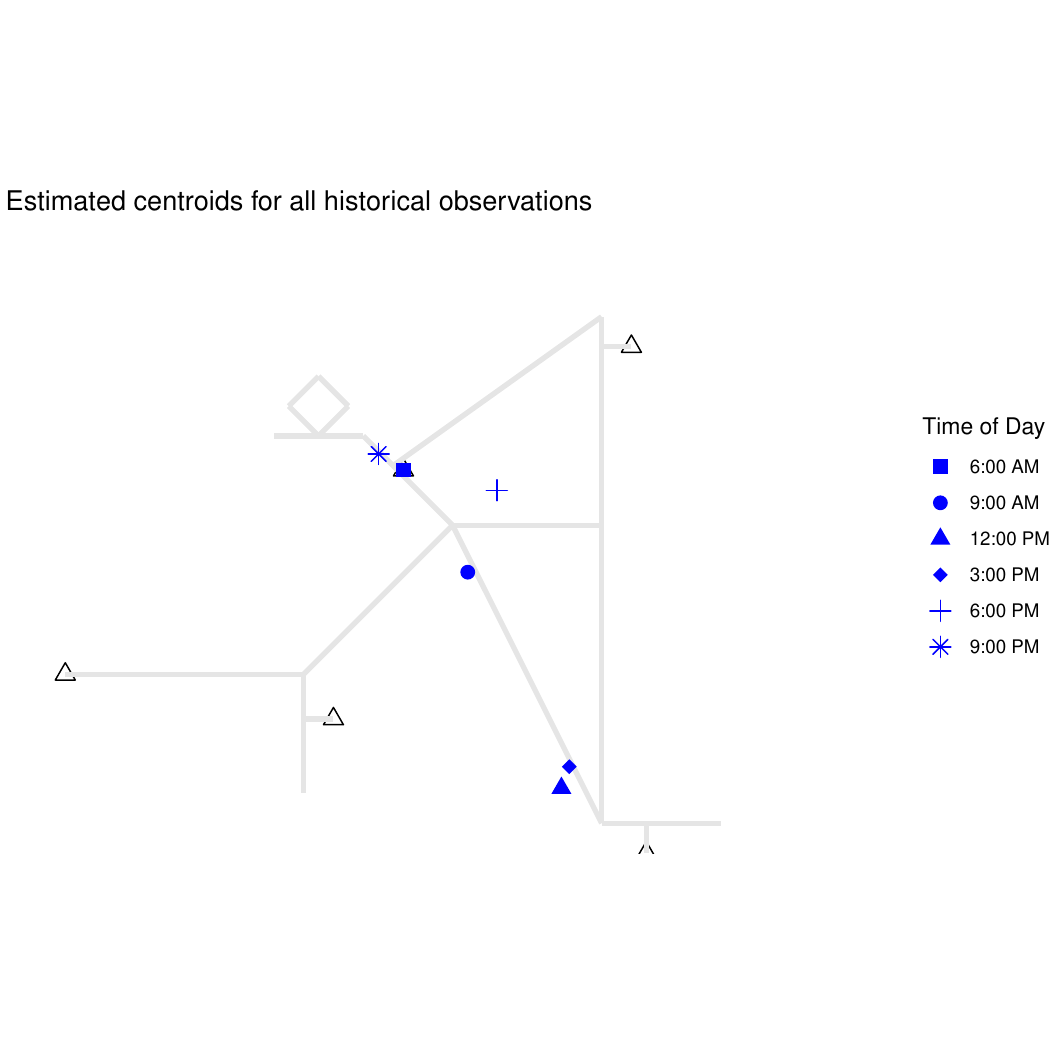}
		\includegraphics[width=0.25\linewidth]{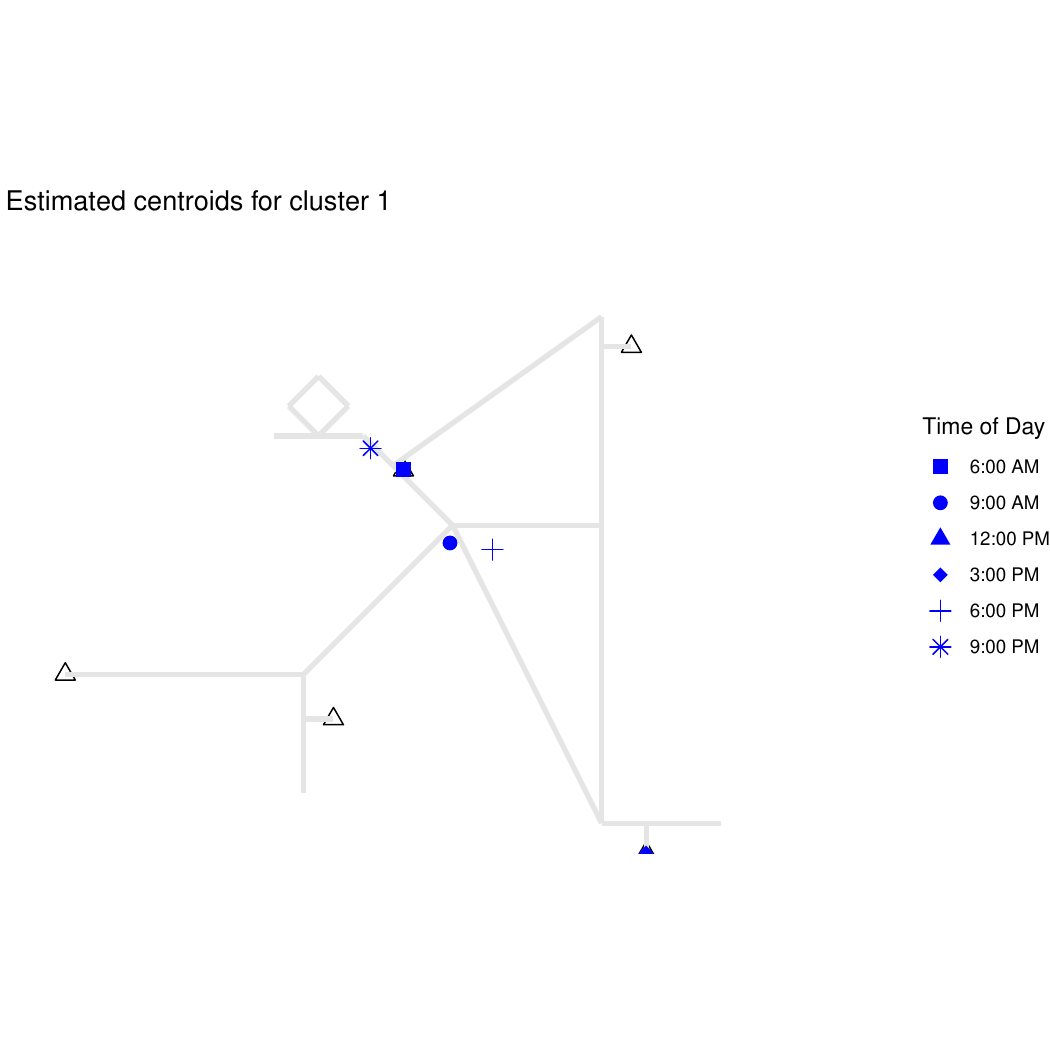}
		\includegraphics[width=0.25\linewidth]{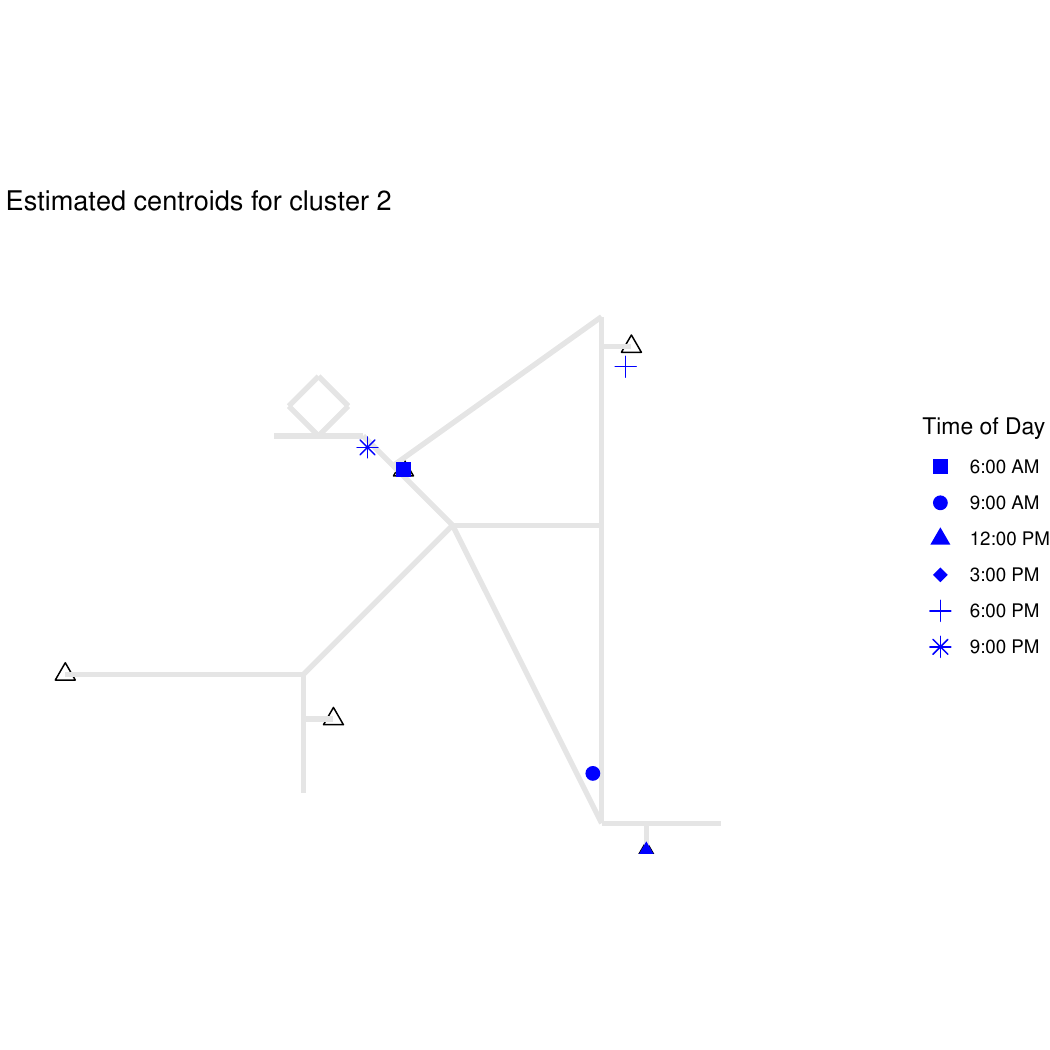}
		\includegraphics[width=0.25\linewidth]{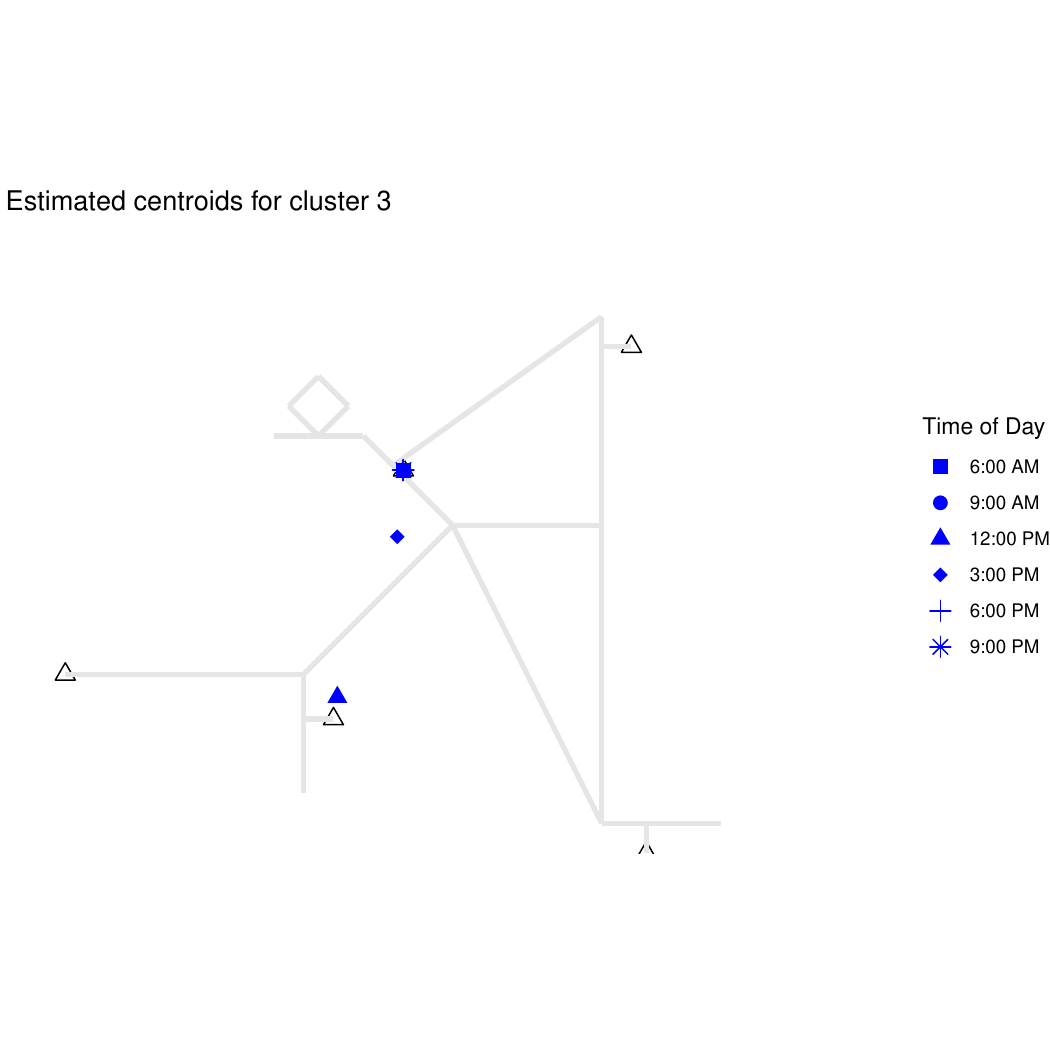}
		\includegraphics[width=0.25\linewidth]{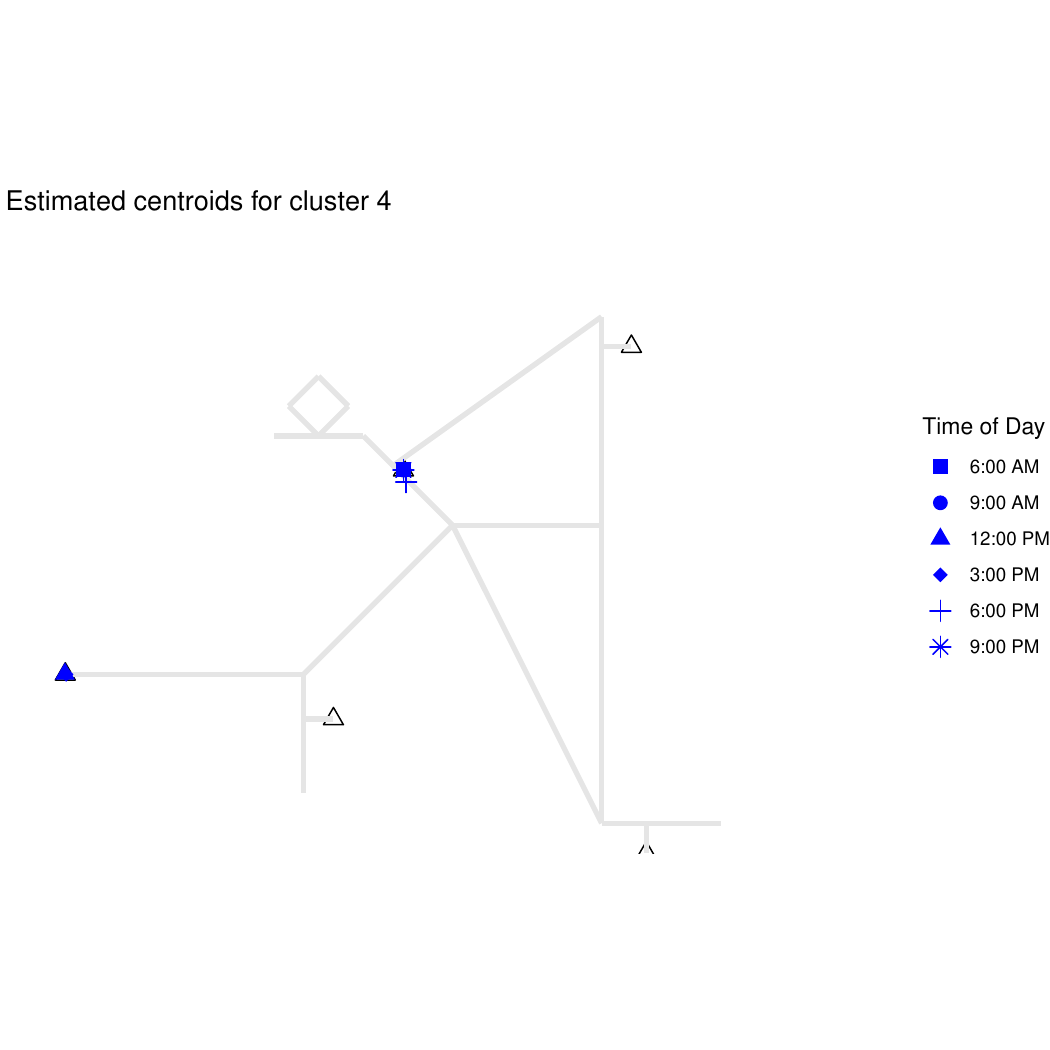}
		\includegraphics[width=0.25\linewidth]{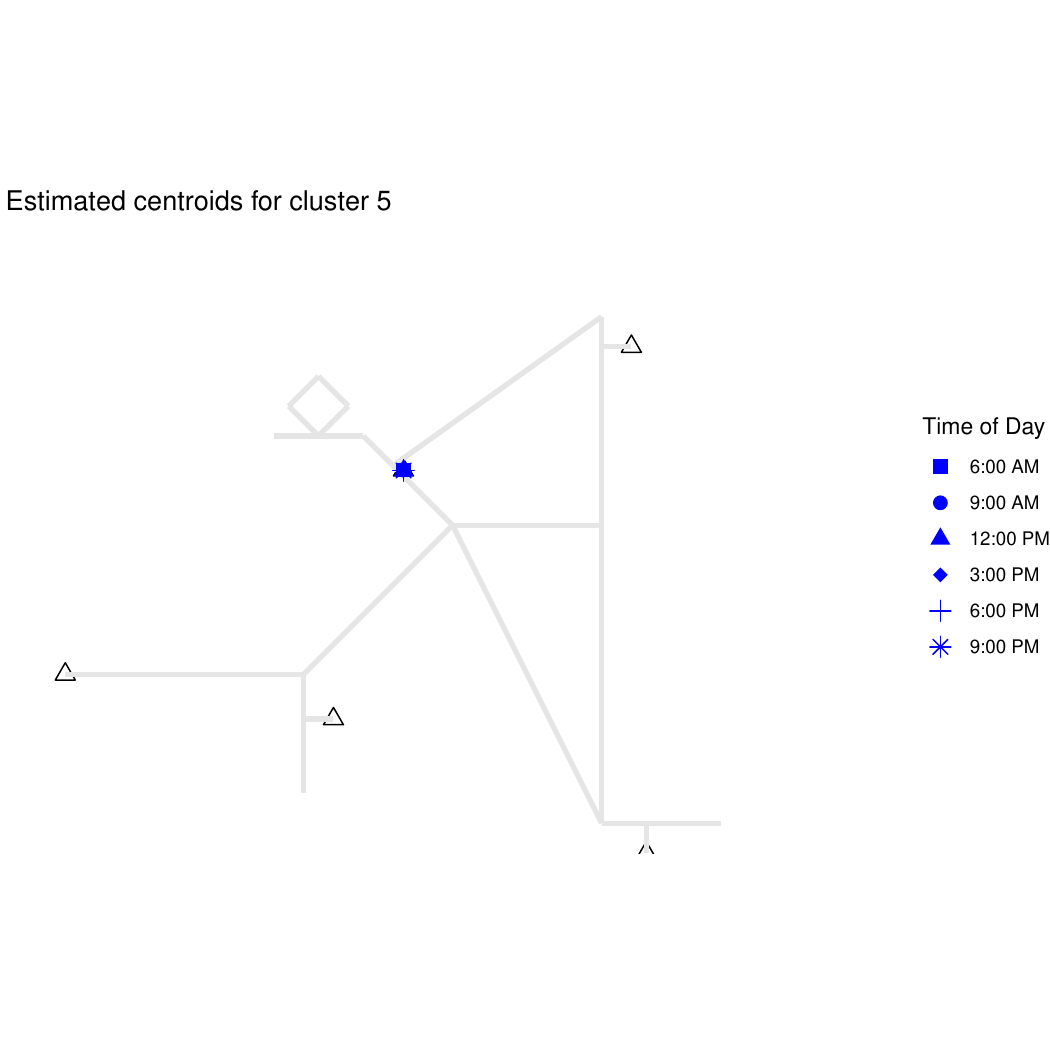}
		\caption{Estimated centers at six times of day (6:00~AM to 9:00~PM) across all days (\textbf{top-left} panel) and for each of the five activity clusters (subsequent panels).}
		\label{fig:sim-centers}
	\end{figure}
	On weekdays, for four of the six time points 6:00~AM, 12:00~PM, 3:00~PM, and 9:00~PM, both clusters have roughly the same center (reflecting time spent at home or in the office). However, at 9:00~AM, the Cluster 1 center is closer to home, whereas the Cluster 2 center is closer to the office, consistent with the expected earlier departure in the weekdays' morning under Pattern 2 in Table~\ref{Table:pattern_detail}.  
	At 6:00~PM, Cluster 1 centers near home, while Cluster 2 centers near the restaurant, mirroring the aforementioned difference in post-work behavior.
	
	Each weekend pattern displays a distinct midday center, aligning with the anchor points for each pattern on weekends: supermarket for Pattern 3, beach for Pattern 4, and home for Pattern 5.
	

	\section{Assumptions}	\label{sec::assum}
	
	\begin{assumption}[Property of the kernel function]
		\label{ass:K-basic}
		The 2-dimensional kernel function $K(x)$ satisfies the following properties:
		
		(a) $K(x)\geq 0$ for any $x\in \mathbb{R}^2$;
		
		(b) $K(x)$ is symmetric, i.e. for any $x_1,x_2\in \mathbb{R}^2$ satisfying $\|x_1\|_2 = \|x_2\|_2$, we have $K(x_1) = K(x_2)$;
		
		(c) $\int K(x) dx=1$;
		
		(d) $\int x^Tx K(x)dx<\infty$.
		
		(e) $K$ have Lipchitz property:  there exists constant $L_{K}$ such that, for any  $x_1, x_2\in\mathbb{R}^2$, we have
		\begin{align*}
			|K(x_1)-K(x_2)|\leq L_{K}\|x_1-x_2\|.
		\end{align*}
	\end{assumption}
	
	Assumption \ref{ass:K-basic} are standard assumptions on the kernel functions
	\citep{wasserman2006all}.
	The additional requirement is the Lipchitz condition (e) but it is a very mild condition.

	\begin{assumption}
		\label{ass:XR-2-derivative}
		Any random trajectory $S(t)$
		is Lipschitz, i.e., there exists a $\bar v < \infty$ such that $\|S(t_1)-S(t_2)\|\leq \bar v d_T(t_1,t_2)$
		for any $t_1,t_2$.
	\end{assumption}
	
	The Lipschitz condition in Assumption \ref{ass:XR-2-derivative}
	allows the case where the individual suddenly start moving.
	It also implies that the velocity of the individual is finite (the Lipschitz constant).

	\begin{assumption}
		\label{ass:measurement deri bounded}
		The density of measurement error $f_{\epsilon}$ is bounded in the sense that $\|f_\epsilon(\eta)\|$, $\|\nabla f_{\epsilon}(\eta)\|$, and $\|\nabla\nabla f_{\epsilon}(\eta)\|_2$ are uniformly bounded by a constant $M_{\epsilon}$ for every $\eta\in\mathbb{R}^2$.
	\end{assumption}
	The Assumption \ref{ass:K-basic} is the common assumption for kernel density estimation, and most common kernel functions can satisfy it. When the distribution of measurement error is fixed, then the Assumption \ref{ass:measurement deri bounded} is automatically satisfied. 
	
	By the convolution theory, 
	Assumptions \ref{ass:measurement deri bounded} implies that the GPS density  $f_{GPS}$
	has bounded second derivative.

	\begin{assumption}
		\label{ass:K1,K2}
		The kernel function $K_T$ is a symmetric PDF and $\int t^2 K(t) dt<\infty$ and 
		is Lipschitz. 
	\end{assumption}
	
	Assumption \ref{ass:K1,K2} is for the kernel function on the time.
	It will be used when we analyze the asymptotic properties of the conditional estimator.

	\section{Proofs}	\label{sec::proof}
	\subsection{Proof of Proposition \ref{prop::GPS}}
	\begin{proof}
		By definition of $f_{GPS}$ in \eqref{def:GPS density},
		\begin{align*}
			f_{GPS}(x) =  \lim_{r\rightarrow 0} \frac{P(X^* \in B(x,r))}{\pi r^2} 
			=  \lim_{r\rightarrow 0} \frac{P(S(U) + \epsilon \in B(x,r))}{\pi r^2}.
		\end{align*}    
		Since $U\sim {\sf Uniform}([0,1])$, we have
		\begin{align*}
			\lim_{r\rightarrow 0} \frac{P(S(U) + \epsilon \in B(x,r))}{\pi r^2}  &= \lim_{r\rightarrow 0} \frac{\int P(S(U) + \epsilon \in B(x,r)|U=t) dt}{\pi r^2} \\
			&= \lim_{r\rightarrow 0} \frac{\int_{t\in[0,1]} P(S(t) + \epsilon \in B(x,r)) dt}{\pi r^2} \\
			&=\lim_{r\rightarrow 0} \frac{\int_{t\in[0,1]} P(X^*_t\in B(x,r)) dt}{\pi r^2}.
		\end{align*}
		Then, by the dominated convergence theorem, 
		\begin{align*}
			\lim_{r\rightarrow 0} \frac{\int_{t\in[0,1]} P(X^*_t\in B(x,r)) dt}{\pi r^2} = \int_{t\in[0,1]}\lim_{r\rightarrow 0} \frac{ P(X^*_t\in B(x,r))}{\pi r^2} dt=  \int_0^1 f_{GPS}(x|t)dt.
		\end{align*}
		By definition of $f_{GPS,A}$ in \eqref{def:GPS density in area},
		\begin{align*}
			f_{GPS,A}(x) = \lim_{r\rightarrow 0} \frac{P(X^*_A \in B(x,r))}{\pi r^2} = \lim_{r\rightarrow 0} \frac{P(S(U_A)+\epsilon \in B(x,r))}{\pi r^2}.
		\end{align*}
		Since $U_A\sim {\sf Uniform}(A)$, then by the same reasoning,
		\begin{align*}
			\lim_{r\rightarrow 0} \frac{P(S(U_A)+\epsilon \in B(x,r))}{\pi r^2}  =  \lim_{r\rightarrow 0} \frac{\int_{t\in A} P(S(t) + \epsilon \in B(x,r)) \frac{1}{|A|}dt}{\pi r^2} =\lim_{r\rightarrow 0}  \frac{1}{|A|}\frac{\int_{t\in A} P(X^*_t\in B(x,r)) dt}{\pi r^2}.
		\end{align*}
		Similarly, by the dominated convergence theorem, we have
		\begin{align*}
			\lim_{r\rightarrow 0}  \frac{1}{|A|}\frac{\int_{t\in A} P(X^*_t\in B(x,r)) dt}{\pi r^2}  =  \int_{t\in A} \lim_{r\rightarrow 0}\frac{1}{|A|}\frac{ P(X^*_t\in B(x,r)) }{\pi r^2}dt = \frac{1}{|A|}\int_A f_{GPS}(x|t)dt.
		\end{align*}
	\end{proof}

	\subsection{Proof of Proposition \ref{prop::AP}}
	\begin{proof}
		By definition of $f_{GPS}$ in \eqref{def:GPS density},
		\begin{align*}
			f_{GPS}(a) = \lim_{r\rightarrow 0} \frac{P(X^* \in B(a,r))}{\pi r^2}= \lim_{r\rightarrow 0} \frac{P(S(U)+\epsilon \in B(a,r))}{\pi r^2}.
		\end{align*}
		Since $U\sim {\sf Uniform}([0,1])$
		\begin{align*}
			& P(S(U)+\epsilon \in B(a,r))\\
			\geq & P(S(U)+\epsilon\in B(a,r), S(U)=a) \\
			= & P(S(U)+\epsilon\in B(a,r)| S(U)=a)P(S(U)=a)\\
			= &  P(S(U)=a)P(\epsilon \in B(0,r))\\
			\geq & \rho_a P(\epsilon \in B(0,r)).
		\end{align*}
		Note that 
		\begin{align*}
			\lim_{r\rightarrow 0}\frac{P(\epsilon \in B(y,r))}{\pi r^2} = \frac{1}{2\pi\sigma^2} \exp\left\{-\frac{y^Ty}{2\sigma^2}\right\},
		\end{align*}
		for any $y \in \mathbb{R}^2$. Hence, $\lim_{r\rightarrow 0}\frac{P(\epsilon \in B(0,r))}{\pi r^2} = \frac{1}{2\pi\sigma^2}$. Then,
		\begin{align*}
			\lim_{r\rightarrow 0} \frac{P(S(U)+\epsilon \in B(a,r))}{\pi r^2} \geq  \frac{\rho_a}{2\pi\sigma^2}.
		\end{align*}
		
		For regular GPS density $f_{GPS}(x)$, we have 
		\begin{align*}
			\int_{B(a, r)} f_{GPS}(x)dx = P(S(U)+\epsilon\in B(a, r)) \geq  P(S(U)=a,S(U)+\epsilon\in B(a, r)).
		\end{align*}
		Moreover, the right-hand-side in the above inequality can further be bound by
		\begin{align*}
			&  P(S(U)=a,S(U)+\epsilon\in B(a, r)) \\
			=&  P(S(U)+\epsilon\in B(a, r)|S(U)=a)P(S(U)=a)\\
			\geq & \rho_a P(\epsilon\in B(0, r)|S(U)=a)\\
			=& {\rho_a}F_{\chi^2_2}\left(\frac{r^2}{\sigma^2}\right).
		\end{align*}
		Note that we use the fact that $\epsilon$ is an isotropic Gaussian with a variance $\sigma^2$, so 
		$$
		P(\epsilon \in B(0, r)) =P(\epsilon^2 \leq r^2) = P\left(\frac{\epsilon^2}{\sigma^2} \leq \frac{r^2}{\sigma^2}\right)=F_{\chi^2_2}\left(\frac{r^2}{\sigma^2}\right).
		$$
		
	\end{proof}
	\subsection{Proof of Proposition \ref{prop::HA}}
	
	\begin{proof}
		We can write $\int_{C\oplus r} f_{GPS}(x)dx$ as
		\begin{align*}
			&\int_{C\oplus r} f_{GPS}(x)dx= P(S(U)+\epsilon\in C\oplus r)\\
			\geq & P(S(U)\in C, S(U)+\epsilon\in C\oplus r) \\
			\geq & P(S(U)\in C,S(U)+\epsilon\in B(S(U), r)) \\
			= & P(S(U)\in C,\epsilon\in B(0, r)) \\
			= & P(S(U)\in C)P(\epsilon\in B(0, r)) \\
			\geq & \rho_C  F_{\chi^2_2}\left(\frac{r^2}{\sigma^2}\right).
		\end{align*}
	\end{proof}

	\subsection{Proof of Proposition \ref{prop::GA}}
	\begin{proof}
		A straight forward calculation shows that 
		\begin{align*}
			G_C  \leq\int_{C} f_{GPS}(x)dx  &= P(S(U)+\epsilon\in C)\\ 
			& = P(S(U)+\epsilon\in C, \epsilon\leq r) + P(S(U)+\epsilon\in C, \epsilon >r)\\
			&\leq P(S(U)+\epsilon\in C, \epsilon\leq r)  + P(\epsilon>r)\\
			&\leq P(S(U)+\epsilon\in C| \epsilon\leq r) P(\epsilon \leq r)  + 1-F_{\chi^2_2}\left(\frac{r^2}{\sigma^2}\right)\\
			&\leq P(S(U)+\epsilon\in C| \epsilon\leq r)  F_{\chi^2_2}\left(\frac{r^2}{\sigma^2}\right) + 1-F_{\chi^2_2}\left(\frac{r^2}{\sigma^2}\right).
		\end{align*}
		Given $\epsilon \leq r$, a necessary condition to $S(U)+\epsilon\in C$ is that $S(U)\in C\oplus r$.
		Thus, 
		$$
		P(S(U)+\epsilon\in C| \epsilon\leq r) \leq P(S(U)\in C\oplus r).
		$$
		Putting this into the long inequality shows that 
		$$
		G_C \leq P(S(U)\in C\oplus r)F_{\chi^2_2}\left(\frac{r^2}{\sigma^2}\right) + 1-F_{\chi^2_2}\left(\frac{r^2}{\sigma^2}\right),
		$$
		which is the desired result after rearrangements.

	\end{proof}
	
	
	\subsection{Proof of Theorem \ref{thm::TW}}
	Let $\hat f_{w,i}(x)$ is the estimated density of $i-$th day observations,  i.e.
	\begin{align*}
		\hat f_{w,i}(x) = \frac{1}{h^2}\sum_{i=1}^{m_{i}}\frac{(t_{i,j+1}-t_{i,j-1})}{2}K\left(\frac{x-X_{i,j}}{h}\right).
	\end{align*}
	
	Since the latent trajectories across different days are IID, 
	this estimator will be IID quantities in different days if the timestamps $\{t_{i,j}: j=1,\ldots, m_i\}$
	are the same for every $i$.
	Since the time-weighted estimator can be written as the mean of each day, i.e.,
	$$
	\hat f_w (x) = \frac{1}{n}\sum_{i=1}^n     \hat f_{w,i}(x) , 
	$$
	we can investigate the bias of $\hat f_w$
	by investigating the bias of $\hat f_{w,i}$ versus $f_{GPS}$.
	The following lemma will be useful for this purpose. 
	
	\begin{lemma}[Riemannian integration approximation of GPS density]\label{lemma:Riemannian approx}
		Assume assumptions ~\ref{ass:K-basic} to \ref{ass:measurement deri bounded}.
		For any $i\in\{1,...,n\}$ and a trajectory $s\in \mathcal{S}$, there are constants $C_{\bar v}, M_\epsilon$ such that
		$C_{\bar v}$ depends on the maximal velocity of $s$ and $M_\epsilon$ depends on the distribution of $\epsilon$ 
		with

			\begin{align}
				\mathbb{E}\left[\hat f_{w,i}(x)\mid S_i=s\right] - \mathbb{E}\left[\frac{1}{h^2}K\left(\frac{x-X^*}{h}\right)\mid S_i=s\right] 
				&\leq    C_{\bar v}\left(\frac{1}{nh^3}\sum_{j=0}^{m_i}(t_{i,j+1}-t_{i,j})^2\right)	\label{eq:Term2, decompose}
			\end{align}

		for every $x\in\mathbb{R}^2$. 
\end{lemma}


\begin{proof}[Proof of Lemma \ref{lemma:Riemannian approx}]
	Since $S_1,...,S_n\sim_{i.i.d} P_S$, without loss of generality, we only need to derive \eqref{eq:Term2, decompose} under $i=1$. Then, for $i=1$, we can decompose ~\eqref{eq:Term2, decompose} to the following two terms:
	\begin{equation}
		\begin{aligned}
			\mathbb{E}\left[\hat f_{w,1}(x)\mid S_1=s\right] &- \mathbb{E}\left[\frac{1}{h^2}K\left(\frac{x-X^*}{h}\right)\mid S_1=s\right]  \\
			&=\underbrace{\frac{1}{h^2}\sum_{j=1}^{m_{1}}\frac{(t_{1,j+1}-t_{1,j-1})}{2}\int_{\epsilon_{1,j}}K\left(\frac{x-s(t_{1,j})+\epsilon_{1,j}}{h}\right) dF_{\epsilon}(\epsilon_{1,j})}_{(I)}\\
			&-\underbrace{\frac{1}{h^2} \int_{\nu}\int_{u\in[0,1]} K\left(\frac{x-s(u)+\nu}{h}\right) dudF_{\epsilon}(\nu)}_{(II)} .
		\end{aligned}
		\label{eq: Lemma 11-1}
	\end{equation}
	
	To simplify in \eqref{eq: Lemma 11-1}, we use $I_h(t_{1,j};s,x)$ to denote $\int_{\epsilon_{1,j}}K\left(\frac{x-s(t_{1,j})+\epsilon_{1,j}}{h}\right) dF_{\epsilon}$, i.e. the expection of $K\left(\frac{x-X_{1,j}}{h}\right)$ at time $t_{1,j}$ given latent trajectory $S_1 = s$. 
	
	\textbf{Term (I): }
	Then, by $t_{1,j+1}-t_{1,j-1} = (t_{1,j+1}-t_{1,j}) + (t_{1,j}-t_{1,j-1}) = \int_{t_{1,j}}^{t_{1,j+1}}du+\int_{t_{1,j-1}}^{t_{1,j}}du$, Term (I) in \eqref{eq: Lemma 11-1} can be written as
	\begin{align}
		&\frac{1}{h^2}\sum_{j=1}^{m_{1}}\frac{(t_{1,j+1}-t_{1,j-1})}{2}I_h(t_{1,j};s,x) \nonumber\\
		=&\frac{(t_{1,1}-t_{1,0})}{2}I_h(t_{1,1};s,x)  + \frac{(t_{1,m_1+1}-t_{1,m})}{2}I_h(t_{1,m};s,x)  + \frac{1}{h^2}\sum_{j=1}^{m_{1}-1}\frac{(t_{1,j+1}-t_{1,j})}{2}I_h(t_{1,j};s,x).\nonumber
	\end{align}
	
	For the weights of first and last observation, note that,
	\begin{align*}
		\frac{(t_{1,1}-t_{1,0})}{2}=\frac{(t_{1,m_1+1}-t_{1,m_1})}{2}=\frac{1-(t_{1,m_1}-t_{1,1})}{2}=\frac{1}{2}\left[\int_{0}^{t_{1,1}}1du + \int_{t_{1,m_1}}^{1}1du\right].
	\end{align*}
	We have
	\begin{align}
		&\frac{1}{h^2}\sum_{j=1}^{m_{1}}\frac{(t_{1,j+1}-t_{1,j-1})}{2}I_h(t_{1,j};s,x) \nonumber\\
		= & \frac{1}{h^2}\sum_{j=1}^{m_{1}-1}\int_{t_{1,j}}^{t_{1,j+1}}I_h(t_{1,j};s,x)du +\frac{1}{2h^2}\left[\int_{0}^{t_{1,1}}I_h(t_{1,1};s,x)du + \int_{t_{1,m_1}}^{1}I_h(t_{1,1};s,x)du\right]+\nonumber\\
		&\frac{1}{2h^2}\left[\int_{0}^{t_{1,1}}I_h(t_{1,m_1};s,x)du + \int_{t_{1,m_1}}^{1}I_h(t_{1,m_1};s,x)du\right].\label{eq:term2: 1st term}
	\end{align}
	
	\textbf{Term (II): }
	Note that $I_h(u;s,x)=\int_{\epsilon}K\left(\frac{x-s(u)+\nu)}{h}\right)dF_{\epsilon}(\nu)$, then Term (II) in the \eqref{eq: Lemma 11-1} can be decomposed to
		\begin{align}
			&\frac{1}{h^2} \int_{\nu}\int_{u\in[0,1]} K\left(\frac{x- s(u)+ \nu}{h}\right) dudF_{\epsilon}(\nu)\nonumber\\
			=&\frac{1}{h^2} \int_{u\in[0,1]}I_h(u;s,x)du \nonumber\\
			=&  \frac{1}{h^2}\left[\int_{0}^{t_{1,1}}I_h(u;s,x)du +\int_{t_{1,1}}^{t_{1,m_1}}I_h(u;s,x)du +  \int_{t_{1,m_1}}^{1}I_h(u;s,x)du \right]\nonumber\\
			=& \frac{1}{h^2}\sum_{j=1}^{m_{1}-1} \int_{t_{1,j}}^{t_{1,j+1}}I_h(u;s,x)du+  \frac{1}{h^2}\left[\int_{0}^{t_{1,1}}I_h(u;s,x)du +\int_{t_{1,m_1}}^{1}I_h(u;s,x)du \right]. 		\label{eq:term2: 2nd term}
	\end{align}
	Namely, we decompose (II) into each interval $[t_{i,j}, t_{i,j+1}]$
	so that we can easily compare it with term (I).
	
	Combining \eqref{eq:term2: 1st term} and \eqref{eq:term2: 2nd term}, \eqref{eq: Lemma 11-1} can be written as:
	\begin{align}
		&\mathbb{E}\left[\hat f_{GPS,1}(x)\mid S_1=s\right] - \mathbb{E}\left[\frac{1}{h^2}K\left(\frac{x-X^*}{h}\right)\mid S_1=s\right] \nonumber \\
		=& \underbrace{\frac{1}{h^2}\sum_{j=1}^{m_{1}-1} \left[\int_{t_{1,j}}^{t_{1,j+1}}I_h(t_{1,j};s,x)-I_h(u;s,x)du+\right]}_{(III)}+\nonumber\\
		&\underbrace{\frac{1}{2h^2} \left[\int_{0}^{t_{1,1}}I_h(t_{1,1};s,x)-I_h(u;s,x)du+\int_{t_{1,m_1}}^{1}I_h(t_{1,m_1};s,x)-I_h(u;s,x)du\right]}_{(IV)}+\nonumber\\
		&\underbrace{\frac{1}{2h^2} \left[\int_{0}^{t_{1,1}}I_h(t_{1,m_1};s,x)-I_h(u;s,x)du+\int_{t_{1,m_1}}^{1}I_h(t_{1,1};s,x)-I_h(u;s,x)du\right]}_{(V)}\label{eq:Term2: 3terms}
	\end{align}
	In RHS of \eqref{eq:Term2: 3terms}, the term (III) represents the error caused by the approximation using GPS observations. The source of the term (IV) and (V) is due to our design of circular weights for the first and last observation in each day. 
	
	Before deriving the convergence rate of term (III), (IV) and (V) in \eqref{eq:Term2: 3terms}, note that, by Assumption \ref{ass:XR-2-derivative} (velocity assumption) and Assumption~\ref{ass:K-basic} (Lipchitz property of the kernel function $K$), for every $j\in \{1,...,m_1\}, u\in [0,1]$ the term $\left| K\left(\frac{x-s(t_{1,j})+\epsilon_{1,j}}{h}\right)-K\left(\frac{x-s(u)+ \epsilon_{1,j}}{h}\right)  \right|$ in the further analysis can be bounded in the following way:
	\begin{align}
		\left| K\left(\frac{x-s(t_{1,j})+ \epsilon_{1,j}}{h}\right)-K\left(\frac{x-s(u)+ \epsilon_{1,j}}{h}\right) \right| \leq  L_K\left\|\frac{s(t_{1,j})}{h}-\frac{s(u)}{h}  \right\| \leq \frac{L_K\bar v}{h}d_T(t_{1,j}, u), \label{eq: Riem-velocity0}
	\end{align}
where $d_T(t_{1,j}, u)= \min\left\{|t_{1,j}-u|, 1-|t_{1,j}-u|\right\}$
	
	\textbf{Term (III):}
	For the convergence rate of the term (III) in \eqref{eq:Term2: 3terms}, we find that
	\begin{equation}
		\begin{aligned}
			&\sum_{j=1}^{m_{1}-1} \left[\int_{t_{1,j}}^{t_{1,j+1}}I_h(u;s,x)-I_h(t_{1,j};s,x)du\right] \\
			=&\frac{1}{h^2}\sum_{j=1}^{m_{1}}\int_{t_{1,j}}^{t_{1,j+1}}\int_{\epsilon_{1,j}}K\left(\frac{x-s(t_{1,j})+\epsilon_{1,j}}{h}\right)dF_{\epsilon}(\epsilon_{1,j})du -\frac{1}{h^2}\sum_{j=1}^{m_{1}}\int_{t_{1,j}}^{t_{1,j+1}}\int_{\epsilon_{1,j}}K\left(\frac{x-s(u)
				+ \epsilon_{1,j}}{h}\right)dF_{\epsilon}(\epsilon_{1,j})du \\
			\leq & \frac{1}{h^2}\sum_{j=1}^{m_{1}}\int_{t_{1,j}}^{t_{1,j+1}}\int_{\epsilon_{1,j}}\left| K\left(\frac{x-s(t_{1,j})+\epsilon_{1,j}}{h}\right)-K\left(\frac{x-s(u)+ \epsilon_{1,j}}{h}\right)  \right| dF_{\epsilon}(\epsilon_{1,j})du  
		\end{aligned}
		\label{eq:K-K}
	\end{equation}
	In \eqref{eq:K-K},  for every $j\in\{1,...,m_i\}$ and $u\in[t_{1,j-1},t_{1,j}]$ corresponding to each integral of the summation, we have $d_T(u,t_{1,j})\leq t_{1,j+1}-t_{1,j}$
	Hence, by applying~\eqref{eq: Riem-velocity0}, the term (III) in \eqref{eq:Term2: 3terms} can be bounded by
	\begin{align}
		&\frac{1}{h^2}\frac{1}{T_1}\sum_{j=1}^{m_{1}-1} \left[\int_{t_{1,j}}^{t_{1,j+1}}I_h(u;s,x)-I_h(t_{1,j};s,x)du\right]\nonumber \\
		\leq &\frac{1}{h^3}\sum_{j=1}^{m_{1}-1}\int_{t_{1,j}}^{t_{1,j+1}}\int L_K\bar v d_T(t_{1,j},u)dF_{\epsilon}(\epsilon_{1,j}) du \quad\quad \text{[By \eqref{eq: Riem-velocity0}]}\nonumber\\
		\leq & \frac{1}{2h^3}\sum_{j=1}^{m_{1}-1} {\left(t_{1,j+1}-t_{1,j}\right)^2}L_K\bar v\quad\quad \text{[By }d_T(u,t_{1,j})\leq t_{1,j+1}-t_{1,j}]\nonumber\\
		\leq & C_{\bar v}\frac{\sum_{j=1}^{m_{1}-1} \left(t_{1,j+1}-t_{1,j}\right)^2}{h^3}, \label{eq: Riem-term1}
	\end{align}   
	for some constant $C_{\bar v}$.
	
	\textbf{Term (IV)}: For term (IV): 
	\begin{align}
		&\frac{1}{2h^2} \left[\int_{0}^{t_{1,1}}I_h(t_{1,1};s,x)-I_h(u;s,x)du+\int_{t_{1,m_1}}^{1}I_h(t_{1,m_1};s,x)-I_h(u;s,x)du\right]\nonumber \\
		=&\frac{1}{2h^2}\int_{0}^{t_{1,1}}\int_{\epsilon_{1,1}}K\left(\frac{x-s(t_{1,1})+\epsilon_{1,1}}{h}\right)-K\left(\frac{x-s(u)+ \epsilon_{1,j}}{h}\right)dF_{\epsilon}(\epsilon_{1,1})du+ \nonumber\\
	&\frac{1}{2h^2}\int_{t_{1,m_1}}^{1}\int_{\epsilon_{1,m_1}}K\left(\frac{x-s(t_{1,m_1})+\epsilon_{1,m_1}}{h}\right)-K\left(\frac{x-s(u)
				+ \epsilon_{1,m_1}}{h}\right)dF_{\epsilon}(\epsilon_{1,m_1})du \label{eq:decom of 1st and last-IV}
	\end{align}
	Since two parts in \eqref{eq:decom of 1st and last-IV} have same form, we only need to focus on the first term for RHS of \eqref{eq:decom of 1st and last-IV}. Recall $t_{1,0}=t_{1,m}-1$, then, for $u\in [0,t_{1,1}]$, $d_T(t_{1,1},u)\leq t_{1,1}\leq t_{1,1}-t_{1,0}$. Thus,
	\begin{align}
	&\frac{1}{2h^2}\int_{0}^{t_{1,1}}\int_{\epsilon_{1,1}}K\left(\frac{x-s(t_{1,1})+\epsilon_{1,1}}{h}\right)-K\left(\frac{x-s(u)
		+ \epsilon_{1,1}}{h}\right)dF_{\epsilon}(\epsilon_{1,1})du \nonumber\\
	\leq &\frac{1}{2h^2}\int_{0}^{t_{1,1}}\int_{\epsilon_{1,1}}\left|K\left(\frac{x-s(t_{1,1})+\epsilon_{1,1}}{h}\right)-K\left(\frac{x-s(u)
		+ \epsilon_{1,1}}{h}\right)\right|dF_{\epsilon}(\epsilon_{1,1})du\nonumber\\
		\leq &\frac{1}{2h^3}\int_{0}^{t_{1,1}}\int L_K\bar v d_T(t_{1,1},u)dF_{\epsilon}(\epsilon_{1,1}) du \quad\quad \text{[By \eqref{eq: Riem-velocity0}]}\nonumber\\
		\leq &C_{\bar v}\frac{\left(t_{1,1}-t_{1,0}\right)^2}{2h^3}.\label{eq: Riem-term2-1}
	\end{align}
	Similarly,  for the second term of RHS of \eqref{eq:decom of 1st and last-IV}, we have
	\begin{align}
		&\frac{1}{2h^2}\int_{t_{1,m_1}}^{1}\int_{\epsilon_{1,m_1}}K\left(\frac{x-s(t_{1,m_1})+\epsilon_{1,m_1}}{h}\right)-K\left(\frac{x-s(u)
			+ \epsilon_{1,m_1}}{h}\right)dF_{\epsilon}(\epsilon_{1,m_1})du \leq C_{\bar v}\frac{\left(t_{1,m_1+1}-t_{1,m_1}\right)^2}{2h^3}.\label{eq: Riem-term2-2}
	\end{align}
	
	\textbf{Term (V):} 
	For term (V): 
	\begin{align}
		&\frac{1}{2h^2} \left[\int_{0}^{t_{1,1}}I_h(t_{1,m_1};s,x)-I_h(u;s,x)du+\int_{t_{1,m_1}}^{1}I_h(t_{1,1};s,x)-I_h(u;s,x)du\right]\nonumber \\
		=&\frac{1}{2h^2}\int_{0}^{t_{1,1}}\int_{\epsilon_{1,m_1}}K\left(\frac{x-s(t_{1,m_1})+\epsilon_{1,m_1}}{h}\right)-K\left(\frac{x-s(u)+ \epsilon_{1,m_1}}{h}\right)dF_{\epsilon}(\epsilon_{1,m_1})du+ \nonumber\\
		&\frac{1}{2h^2}\int_{t_{1,m_1}}^{1}\int_{\epsilon_{1,1}}K\left(\frac{x-s(t_{1,1})+\epsilon_{1,1}}{h}\right)-K\left(\frac{x-s(u)
			+ \epsilon_{1,1}}{h}\right)dF_{\epsilon}(\epsilon_{1,1})du \label{eq:decom of 1st and last-V}
	\end{align}
		Similar as the procedure of deriving convergence rate of term (IV), the two parts in RHS of \eqref{eq:decom of 1st and last-V} have same form, we only need to focus on the first term for RHS of \eqref{eq:decom of 1st and last-V}, for $u\in [0,t_{1,1}]$, $d_T(t_{1,m_1},u)\leq 1-(t_{1,m_1}-t_{1,1})= t_{1,m_1+1}-t_{1,m_1}$. Thus,
	\begin{align}
		&\frac{1}{2h^2}\int_{0}^{t_{1,1}}\int_{\epsilon_{1,1}}K\left(\frac{x-s(t_{1,m_1})+\epsilon_{1,m_1}}{h}\right)-K\left(\frac{x-s(u)
			+ \epsilon_{1,m_1}}{h}\right)dF_{\epsilon}(\epsilon_{1,m_1})du \nonumber\\
		\leq &\frac{1}{2h^2}\int_{0}^{t_{1,1}}\int_{\epsilon_{1,1}}\left|K\left(\frac{x-s(t_{1,m_1})+\epsilon_{1,m_1}}{h}\right)-K\left(\frac{x-s(u)
			+ \epsilon_{1,m_1}}{h}\right)\right|dF_{\epsilon}(\epsilon_{1,m_1})du\nonumber\\
		\leq &\frac{1}{2h^3}\int_{0}^{t_{1,1}}\int L_K\bar v d_T(t_{1,m_1},u)dF_{\epsilon}(\epsilon_{1,1}) du \quad\quad \text{[By \eqref{eq: Riem-velocity0}]}\nonumber\\
		\leq &C_{\bar v}\frac{\left(t_{1,m_1+1}-t_{1,m_1}\right)^2}{2h^3}.\label{eq: Riem-term3-1}
	\end{align}
	Similarly, 
	\begin{align}
		\frac{1}{2h^2}\int_{t_{1,m_1}}^{1}\int_{\epsilon_{1,1}}K\left(\frac{x-s(t_{1,1})+\epsilon_{1,1}}{h}\right)-K\left(\frac{x-s(u)
			+ \epsilon_{1,1}}{h}\right)dF_{\epsilon}(\epsilon_{1,1})du\leq C_{\bar v}\frac{\left(t_{1,1}-t_{1,0}\right)^2}{2h^3} \label{eq: Riem-term3-2}
	\end{align}
	
	Hence, combining \eqref{eq: Riem-term1}, \eqref{eq: Riem-term2-1}, \eqref{eq: Riem-term2-2}  and \eqref{eq: Riem-term3-1}, \eqref{eq: Riem-term3-2}, for every $x\in\mathcal{W}$, there is constant $C_{\bar v}$ such that
	\begin{align*}
		\mathbb{E}\left[\hat f_{w,i}(x)\mid S=s\right] - \mathbb{E}\left[\frac{1}{h^2}K\left(\frac{x-X^*}{h}\right)\mid S=s\right]
		\leq   C_{\bar v}\left(\frac{1}{n h^3}\sum_{i=1}^{n}\sum_{j=0}^{m_k}(t_{i,j+1}-t_{i,j})^2\right)
	\end{align*}
	
\end{proof}


\begin{proof}
	For any $x$, $f_{GPS}(x)-\hat f_{w}(x)$ can be decomposed to the following three terms:
	\begin{align*}
		f_{GPS}(x)-\hat f_{w}(x) =& \underbrace{f_{GPS}(x)-\mathbb{E}\left[\frac{1}{h^2}K\left(\frac{x-X^*}{h}\right)\right]}_{(I)} + \underbrace{\mathbb{E}\left[\frac{1}{h^2}K\left(\frac{x-X^*}{h}\right)\right] - \mathbb{E}[\hat f_{w}(x)]}_{(II)} +\underbrace{\mathbb{E}[\hat f_{w}(x)]  - \hat f_{w}(x)}_{(III)}.
	\end{align*}

	{\bf Term (I):} $f_{GPS}(x)-\mathbb{E}\left[\frac{1}{h^2}K\left(\frac{x-X^*}{h}\right)\right]$;  For any $x\in\mathbb{R}^2$, we have
	\begin{align*}
		&f_{GPS}(x)-\mathbb{E}\left[\frac{1}{h^2}K\left(\frac{x-X^*}{h}\right)\right]\\
		=& \frac{1}{h^2} \int_{y\in\mathbb{R}^2} K\left(\frac{x-y}{h}\right)f_{GPS}(y)dy-f_{GPS}(x)\\
		=&\int_{z\in\mathbb{R}^2} K\left(z\right)f_{GPS}(x-hz)dz-f_{GPS}(x) 
	\end{align*}
	By Assumption \ref{ass:measurement deri bounded}, $f_{GPS}$ is twice differentiable,
	\begin{align*}
		\int_{z\in\mathbb{R}^2} K\left(z\right)f_{GPS}(x-hz)dz-f_{GPS}(x)  =& \int_{y\in\mathbb{R}^2} K\left(z\right)f_{GPS}(x)+hz^T\nabla f_{GPS}(x)\\
		&+h^2z^T\nabla\nabla f_{GPS}(x)zdz-f_{GPS}(x) +o(h^2).
	\end{align*}
	By Assumption \ref{ass:K-basic} and Assumption~\ref{ass:measurement deri bounded}, $\int_{z\in\mathbb{R}^2}K(z)dz= 1$, $\int_{z\in\mathbb{R}^2}zK(z)dz= 0$, $\nabla\nabla f_{GPS}$ is bounded. 
	Hence,
	\begin{align}
		\int_{z\in\mathbb{R}^2} K\left(z\right)f_{GPS}(x-hz)dz-f_{GPS}(x) = O(h^2).\label{eq: thm9: term1}
	\end{align}
	
	{\bf Term (II): $\mathbb{E}\left[\frac{1}{h^2}K\left(\frac{x-X^*}{h}\right)\right] - \mathbb{E}[\hat f_{w}(x)]$.} 
	Recall that $\hat f_{w,i}(x)$ to is the estimated density contributed by the observations in $i-$th day, i.e.
	\begin{align*}
		\hat f_{w,i}(x) = \frac{1}{h^2}\sum_{j=1}^{m_{i}}\frac{(t_{i,j+1}-t_{i,j-1})}{2}K\left(\frac{x-X_{i,j}}{h}\right).
	\end{align*}
	Then, we can decompose the expectation of the density estimator:
	\begin{align*}
		&\mathbb{E}\left[\hat f_{w}(x)\right] \\
		= &\frac{1}{2nh^2}\sum_{i=1}^{n}\mathbb{E}\left[\sum_{j=1}^{m_{i}}{(t_{i,j+1}-t_{i,j-1})}K\left(\frac{x-X_{i,j}}{h}\right) \right]\\
		= &\frac{1}{2nh^2}\sum_{i=1}^{n}\mathbb{E}\left\{\mathbb{E}\left[\sum_{j=1}^{m_{i}}{(t_{i,j+1}-t_{i,j-1})}K\left(\frac{x-X_{i,j}}{h}\right)\mid S_i \right]\right\}\\
		=&\frac{1}{2nh^2}\sum_{i=1}^{n}\sum_{j=1}^{m_{i}}{(t_{i,j+1}-t_{i,j-1})}\int_{s\in\mathcal{S}}\int_{\epsilon_{i,j}}K\left(\frac{x-s(t_{i,j}) +\epsilon_{i,j}}{h}\right) dF_{S}(s)dF_{\epsilon}(\epsilon_{i,j}) \\
		\triangleq &\frac{1}{n}\sum_{i=1}^{n}\int_{s\in\mathcal{S}}\mathbb{E}\left[\hat f_{w,i}(x)\mid S=s\right]dF_{S}(s),
	\end{align*}
	

		
		Note that 
		$$
		\mathbb{E}\left[\frac{1}{h^2}K\left(\frac{x-X^*}{h}\right)\right] =\int_{s\in\mathcal{S}}\mathbb{E}\left[\frac{1}{h^2}K\left(\frac{x-X^*}{h}\right)|S=s\right]dF_{S}(s).
		$$
		Hence, term (II) can be written as
		\begin{align*}
			\mathbb{E}[\hat f_{w}(x)]  - \mathbb{E}\left[\frac{1}{h^2}K\left(\frac{x-X^*}{h}\right)\right] = \frac{1}{n}\sum_{i=1}^{n}\int_{s\in\mathcal{S}}\mathbb{E}\left[\hat f_{w,i}(x)\mid S=s\right]-\mathbb{E}\left[\frac{1}{h^2}K\left(\frac{x-X^*}{h}\right)|S=s\right] dF_{S}(s).
		\end{align*}
		Then, by Lemma \ref{lemma:Riemannian approx}, for any $i\in\{1,...,n\}$ and $s\in \mathcal{S}$, there is uniform constant $C_{\bar v}$  such that 
		\begin{align*}
		\mathbb{E}\left[\hat f_{w,i}(x)\mid S=s\right] - \mathbb{E}\left[\frac{1}{h^2}K\left(\frac{x-X^*}{h}\right)\mid S=s\right]  \leq &C_{\bar v}\left(\frac{1}{nh^3}\sum_{i=1}^{n}\sum_{j=0}^{m_i}(t_{i,j+1}-t_{i,j})^2\right)
		\end{align*}
		for every $x\in\mathcal{W}$. Hence,
		\begin{align*}
			\mathbb{E}\left[\hat f_{w}(x)\right] - \mathbb{E}\left[\frac{1}{h^2}K\left(\frac{x-X^*}{h}\right)\right]  = O\left(\frac{1}{nh^3}\sum_{i=1}^{n}\sum_{j=0}^{m}(t_{i,j+1}-t_{i,j})^2\right).
		\end{align*}

		{\bf Term (III).}
		Note that the term (III) is the variance of $\hat f_{w}(x)$:
		\begin{align*}
			& \text{Var}[\hat f_{w}(x)] = \text{Var}\left[\frac{1}{nh^2}\sum_{i=1}^{n}\sum_{j=1}^{m_i} W_{i,j}K\left(\frac{x-S_i(t_{i,j})+\epsilon_{i,j}}{h}\right)\right] = \frac{1}{n^2h^4}\sum_{i=1}^{n}\text{Var}\left[\sum_{j=1}^{m_i}W_{i,j} K\left(\frac{x-S_i(t_{i,j})+\epsilon_{i,j}}{h}\right)\right].
		\end{align*}
		by the independence among $\{S_1,...,S_{n}\}$. To simplify, we only need to consider $i=1$.  Then we calculate 
		\begin{equation}
			\begin{aligned}
				&\text{Var}\left[\sum_{j=1}^{m_1} W_{i,j}K\left(\frac{x-S(t_{1,j})+\epsilon_{1,j}}{h}\right)\right] \\
				=& \mathbb{E}\left\{\text{Var}\left[\sum_{j=1}^{m_1} W_{i,j}K\left(\frac{x-S(t_{1,j})
					+\epsilon_{1,j}}{h}\right)|S_1\right]\right\}+ \text{Var}\left\{\mathbb{E}\left[\sum_{j=1}^{m_1} W_{i,j}K\left(\frac{x-S(t_{1,j})+\epsilon_{1,j}}{h}\right)|S_1\right]\right\}.
			\end{aligned}
			\label{eq:E+Var}
		\end{equation}
		For the first term of the RHS of \eqref{eq:E+Var}, note that $\epsilon_{1,1},\ldots, \epsilon_{1,m_1}$ are IID,
		\begin{align*}
			& \mathbb{E}\left\{\text{Var}\left[\sum_{j=1}^{m_1} W_{i,j}K\left(\frac{x-S(t_{1,j}))+\epsilon_{1,j}}{h}\right)|S_1\right]\right\} \\
			=& \mathbb{E}\left\{\sum_{i=1}^{m_1} {W_{i,j}^2}\text{Var}\left[K\left(\frac{x-S(t_{1,j})+\epsilon_{1,j}}{h}\right)|S_1\right]\right\} \\
			\leq & \mathbb{E}\left\{\sum_{j=1}^{m_1} {W_{i,j}^2}\mathbb{E}\left[K\left(\frac{x-S(t_{1,j})+\epsilon_{1,j}}{h}\right)^2|S_1\right]\right\} .
		\end{align*}
		Note that, for any $s\in\mathcal{S}$, $\mathbb{E}\left[K\left(\frac{x-S(t_{1,j})+\epsilon_{1,j}}{h}\right)^2|S=s\right]$ only depends on the measurement error that is independent at each timestamp. Then, 
		\begin{align*}
			&\mathbb{E}\left\{\sum_{j=1}^{m_1} {W_{i,j}^2}\mathbb{E}\left[K\left(\frac{x-S(t_{1,j})+\epsilon_{i,j}}{h}\right)^2|S\right]\right\}\\
			= & \mathbb{E}\left\{\sum_{j=1}^{m_1} {W_{i,j}^2}\int_{\epsilon_{1,j}}K\left(\frac{x-S(t_{1,j})+ \epsilon_{1,j}}{h}\right)^2f_{\epsilon}(\epsilon_{1,j})d\epsilon_{1,j}\right\}\\
			= &\int_{s}\sum_{j=1}^{m_1} {W_{i,j}^2}\int_{\epsilon_{1,j}}K\left(\frac{x-s(t_{1,j})+ \epsilon_{1,j}}{h}\right)^2f_{\epsilon}(\epsilon_{1,j})d\epsilon_{1,j}dF_S(s)\\
			= & h^2\sum_{j=1}^{m_1} {W_{i,j}^2}\int_{s}\int_{u_{1,j}}K\left(u_{1,j}\right)^2f_{\epsilon}(u_{1,j}h-x+s(t_{1,j}))du_{1,j}dF_S(s).
		\end{align*}
		By Assumption~\ref{ass:K-basic} and~\ref{ass:measurement deri bounded}, $f_{\epsilon}$ and $\int_{x\in \mathbb{R}^2}K(x)^2dx$ is bounded, we have
		\begin{align}
			& h^2\sum_{j=1}^{m_1} {W_{i,j}^2}\int_{s}\int_{\epsilon_{1,j}}K\left(u_{1,j}\right)^2f_{\epsilon}(u_{1,j}h-x+s(t_{1,j}))du_{1,j}dF_S(s)=O\left({h^2}\sum_{j=1}^{m_1}W_{i,j}^2\right).\label{eq:1st term in (29)}
		\end{align}
		For the second term in RHS of \eqref{eq:E+Var}, by
		\begin{equation}
			\begin{aligned}
				& \text{Var}\left\{\mathbb{E}\left[\sum_{j=1}^{m_1} W_{i,j}K\left(\frac{x-S(t_{1,j})+\epsilon_{i,j}}{h}\right)|S\right]\right\}\\
				\leq & \mathbb{E}\left\{\left\{\mathbb{E}\left[\sum_{j=1}^{m_1} W_{i,j}K\left(\frac{x-S(t_{1,j})+\epsilon_{i,j}}{h}\right)|S\right]\right\}^2\right\},
			\end{aligned}
			\label{eq:E(E)}
		\end{equation}
		we consider $\mathbb{E}\left[\sum_{j=1}^{m_1} W_{i,j}K\left(\frac{x-S(t_{1,j})+\epsilon_{i,j}}{h}\right)|S\right]$ in \eqref{eq:E(E)} first. For any $s\in \mathcal{S}$:
		\begin{align*}
			&  \mathbb{E}\left[\sum_{j=1}^{m_1} W_{i,j}K\left(\frac{x-S(t_{1,j})+\epsilon_{i,j}}{h}\right)|S=s\right]\\
			= & \sum_{j=1}^{m_1}W_{i,j} \int_{\epsilon_{1,j}}K\left(\frac{x-s(t_{1,j})+\epsilon_{1,j}}{h}\right)f_{\epsilon}(\epsilon_{1,j})d\epsilon_{1,j}\\
			= & h^2\sum_{j=1}^{m_1} W_{i,j} \int_{\epsilon_{1,j}}K\left(u_{1,j}\right)f_{\epsilon}(u_{1,j}h-x+s(t_{1,j}))du_{1,j}.
		\end{align*}
		For any $s\in \mathcal{S}$, $\int_{\epsilon_{1,j}}K\left(u_{1,j}\right)f_{\epsilon}(u_{1,j}h-x+s(t_{1,j}))du_{1,j}$ is upped-bounded by $\max_y f_{\epsilon}(y)$, then the RHS of \eqref{eq:E(E)} can be upper-bounded by
		\begin{align*}
			&\mathbb{E}\left\{\left\{\mathbb{E}\left[\sum_{j=1}^{m_1} W_{i,j}K\left(\frac{x-S(t_{1,j})+\epsilon_{i,j}}{h}\right)|S\right]\right\}^2\right\}\\
			= & \mathbb{E}\left[\left\{h^2\sum_{j=1}^{m_1} W_{i,j} \int_{\epsilon_{1,j}}K\left(u_{1,j}\right)f_{\epsilon}(u_{1,j}h-x+S(t_{1,j}))du_{1,j}\right]^2 \right\}\\
			\leq & h^4\mathbb{E}\left\{\left[\sum_{j=1}^{m_1} W_{i,j} \max_y f_{\epsilon}(y)\right]^2\right\}\leq {h^4}\max_y f_{\epsilon}(y)^2.
		\end{align*}
		By Assumption \ref{ass:measurement deri bounded}, $f_{\epsilon}$ is uniformly bounded by some constant. Hence,
		\begin{align}
			\text{Var}\left\{\mathbb{E}\left[\sum_{j=1}^{m_1} W_{i,j}K\left(\frac{x-S(t_{1,j})+\epsilon_{i,j}}{h}\right)|S\right]\right\} = O(h^4).\label{eq:2nd term in (29)}
		\end{align}
		By \eqref{eq:1st term in (29)} and \eqref{eq:2nd term in (29)}, we have
		\begin{align*}
			\text{Var}[\hat f_{w}(x)]  =& \frac{1}{4n^2h^4}\sum_{i=1}^{n}\text{Var}\left[\sum_{j=1}^{m_i} W_{i,j}K\left(\frac{x-S_i(t_{i,j})+\epsilon_{i,j}}{h}\right)\right]=O\left(\frac{1}{n^2h^2}\sum_{i=1}^{n}\sum_{j=1}^{m_i} W_{i,j}^2\right)+O\left(\frac{1}{n}\right).
		\end{align*}
		Hence, the convergence rate of \textbf{Term (III)} is
		\begin{align*}
			\mathbb{E}[\hat f_{w}(x)]  - \hat f_{w}(x)
			=& O_p\left(\frac{1}{nh}
			\sqrt{\sum_{i=1}^{n}\sum_{j=1}^{m_i} W_{i,j}^2}\right)+O_p\left(n^{-\frac{1}{2}}\right) 
		\end{align*}
		
		
		Combining the convergence rate of 3 terms derived above, we have
		\begin{align*}
			\hat f_{w}(x)-f_{w}(x) =& O(h^2)+  O\left(\frac{1}{nh^3}\sum_{i=1}^{n}\sum_{j=0}^{m}(t_{i,j+1}-t_{i,j})^2\right)+ O_p\left(\frac{1}{nh}
			\sqrt{\sum_{i=1}^{n}\sum_{j=1}^{m_i} W_{i,j}^2}\right)+O_p\left(n^{-\frac{1}{2}}\right).
		\end{align*}  
	\end{proof}

	\subsection{Proof of Proposition \ref{prop::TW::bw}}
	\begin{proof}
		From equation \eqref{eq::TW::mise},
		we have 
		\begin{align*}
			\frac{(S)}{(B)} = o(1)
			&\Longleftrightarrow \frac{m^2h^6}{nmh^2} = o(1)\\
			&\Longleftrightarrow  \frac{mh^4}{n} = o(1)\\
			&\Longleftrightarrow mh^4=o(n).
		\end{align*}
		
		Since the temporal resolution bias (B) is the dominating term, 
		we should choose $h$ to optimize with respect to this,
		which leads to 
		$h \asymp m^{-1/5}$
		when $mh^4 = o(n)$. 
		This leads to the inequality $m^{1/5}=o(n)$
		when we put $h$ into the inequality $mh^4 = o(n)$.

		On the other hand, (S) dominates (B)
		when $n=o(mh^4)$ and in this case, the optimal smoothing bandwidth is 
		$h\asymp (nm)^{-1/3}$.
		By putting this into $n=o(mh^4)$, we obtain the inequality $n=o(m^{1/5})$. 
		
		The convergence rates are obtained by simply plug-in these optimal smoothing bandwidth. 
		
		Note that at the critical point $n \asymp m^{1/5}$,
		we have
		$m^{1/5} \asymp (nm)^{1/3}$,
		so the two choices coincide. 
		
	\end{proof}
	
	\subsection{Proof of Theorem \ref{thm::ckde}}
	We use $\hat f_{i}(x|t)$ to denote the $i-$th day's contribution to $ \hat f(x|t)$, i.e.  
	\begin{align*}
		\hat f_{i}(x|t) = \frac{1}{h_X^2}\frac{\sum_{j=1}^{m_i} K\left(\frac{d_T(t_{i,j},t)}{h_T}\right)K\left(\frac{X_{i,j}-x}{h_X}\right)}{\sum_{j'=1}^{m_i} K\left(\frac{d_T(t_{i,j'},t)}{h_T}\right)}.
	\end{align*}
	Recall that $X_{i,j}=S_i(t_{i,j})+\epsilon_{i,j}$, then, for any $s\in\mathcal{S}$, we define $\hat f_{i}(x|t;s)$ in the following way:  
	\begin{align*}
		\hat f_{i}(x|t;s) = \frac{1}{h_X^2}\frac{\sum_{j=1}^{m_i} K\left(\frac{d_T(t_{i,j},t)}{h_T}\right)K\left(\frac{s(t_{i,j})+\epsilon_{i,j}-x}{h_X}\right)}{\sum_{j'=1}^{m_i} K\left(\frac{d_T(t_{i,j'},t)}{h_T}\right)}.
	\end{align*}
	For any $s\in\mathcal{S}$, we define $f_{GPS}(x|t; s)$ as
	\begin{align}
		f_{GPS}(x|t;s) &= \lim_{r\rightarrow 0} \frac{P(s(t)+\epsilon\in B(x,r))}{\pi r^2} = f_{\epsilon}(x-s(t))\label{def:con-den-GPS t,s}
	\end{align}
	
	With the above definitions,  $\hat f_{i}(x|t)$ and $f_{G P S}(x|t)$ can be written as
	\begin{align}
		\hat f_{i}(x|t)=\int_{s\in\mathcal{S}}\hat f_{i}(x|t; s)dF_{S}(s), \quad f_{G P S}(x|t) = \int_{s\in\mathcal{S}}f_{GPS}(x|t; s)dF_{S}(s).\label{def: cond-den single-day}
	\end{align}
	
	Our analysis begins with a lemma controlling the error
	for a given trajectory $s$.

	\begin{lemma}\label{thm:single day, conditional, deter}
		Under Assumption \ref{ass:XR-2-derivative} to Assumption \ref{ass:K1,K2}, for any $i\in \{1,...,n\}$ and $s\in\mathcal{S}$, we have
		\begin{align}
			\mathbb{E}\left( \hat f_{i}(x|t;s)\right) - f_{G P S}(x|t; s) = O(h_X^2)+O\left(\sum_{j=1}^{m_i} \frac{K\left(d_T(\frac{t_{i,j},t)}{h_T}\right)}{\sum_{j'=1}^{m_i} K\left(\frac{d_T(t_{i,j'},t)}{h_T}\right)} f_{\epsilon}(x-s(t_{i,j}))-f_{\epsilon}(x-s(t))\right),\label{eq:lemma-ckd:bias}
		\end{align}
		and
		\begin{align}
			\text{Var}\left(\hat f_{i}(x|t;s)\right) = \frac{M_{\epsilon}}{h_X^2}\sum_{j=1}^{m_i} \left[\frac{K\left(\frac{d_T(t_{i,j},t)}{h_T}\right)}{\sum_{j'=1}^{m_i} K\left(\frac{d_T(t_{i,j'},t)}{h_T}\right)}\right]^2 \int_{x\in \mathbb{R}^2}K_1^2(x)dx,\label{eq:lemma-ckd:var}
		\end{align}
		for any $t\in [0,1]$, and $x\in \mathcal{W}$ as $h_X\rightarrow 0, m_i\rightarrow \infty$ for any $i\in\{1,...,n\}$.
	\end{lemma}
	\begin{proof}
		By the independence among $S_1,...,S_{n}$, WLOG, we can just let $i=1$.
		
		\textbf{Bias:} We first derive the bias of $\hat f_{G P S,1}(x|t)$, for any $s\in \mathcal{S}$, by variable changing $y_{1,j}=\frac{s(t_{1,j})+\epsilon_{1,j}-x}{h_X}, j=1,...,m_1$, and \eqref{def:con-den-GPS t,s}, we have
		\begin{align*}
			&\mathbb{E}\left( \hat f_{1}(x|t;s)\right) - f_{G P S}(x|t; s) \\
			=& \mathbb{E}\left[\frac{1}{h_X^2}\frac{\sum_{j=1}^{m_1} K\left(\frac{d_T(t_{1,j},t)}{h_T}\right)K\left(\frac{s(t_{1,j})+\epsilon_{1,j}-x}{h_X}\right)}{\sum_{j'=1}^{m_1} K\left(\frac{d_T(t_{1,j'},t)}{h_T}\right)} \right]- f_{\epsilon}(x-s(t))  \\
			=& \frac{1}{h_X^2}\sum_{j=1}^{m_1} \frac{K\left(\frac{d_T(t_{1,j},t)}{h_T}\right)}{\sum_{j'=1}^{m_1} K\left(\frac{d_T(t_{1,j},t)}{h_T}\right)} \int_{\epsilon_{1,j}\in\mathbb{R}^2} K\left(\frac{s(t_{1,j})+\epsilon_{1,j}-x}{h_X}\right)f_{\epsilon}(\epsilon_{1,j})d\epsilon_{1,j} -f_{\epsilon}(x-s(t))\\
			=&\sum_{j=1}^{m_1} \frac{K\left(\frac{d_T(t_{1,j},t)}{h_T}\right)}{\sum_{j=1}^{m_1} K\left(\frac{d_T(t_{1,j},t)}{h_T}\right)} \int_{y_{1,j}\in\mathbb{R}^2} K_1\left(y_{1,j}\right)f_{\epsilon}(h_Xy_{1,j}+x-s(t_{1,j}))dy_{1,j} - f_{\epsilon}(x-s(t)).
		\end{align*}
		Then, by Assumption \ref{ass:measurement deri bounded}, we apply Taylor expansion on $f_{\epsilon}(h_Xy_{1,j}+x-s(t_{1,j}))$:
		\begin{align*}
			&\sum_{j=1}^{m_1} \frac{K_T\left(\frac{d_T(t_{1,j},t)}{h_T}\right)}{\sum_{j=1}^{m_1} K_T\left(\frac{d_T(t_{1,j},t)}{h_T}\right)} \int_{y_{1,j}\in\mathbb{R}^2} K_1\left(y_{1,j}\right)f_{\epsilon}(h_Xy_{1,j}+x-s(t_{1,j}))dy_{1,j} -f_{\epsilon}(x-s(t))\\
			=&\sum_{j=1}^{m_1} \frac{K\left(\frac{d_T(t_{1,j},t)}{h_T}\right)}{\sum_{j'=1}^{m_1} K\left(\frac{d_T(t_{1,j'},t)}{h_T}\right)} \int_{y_{1,j}\in\mathbb{R}^2} K_1\left(y_{1,j}\right) [f_{\epsilon}(x-s(t_{1,j}))+h_X\nabla f_{\epsilon}(x-s(t_{1,j}))^Ty_{1,j}+\\
			&h_X^2y_{1,j}^T\nabla\nabla f_{\epsilon}(x-s(t_{1,j}))y_{1,j}-f_{\epsilon}(x-s(t))]dy_{1,j}+o(h_X^2) .
		\end{align*}
		By Assumption \ref{ass:K1,K2}, $\int_{y\in\mathbb{R}^2} K(y)ydy=0$, then
		\begin{align*}
			&\sum_{j=1}^{m_1} \frac{K\left(\frac{d_T(t_{1,j},t)}{h_T}\right)}{\sum_{j'=1}^{m_1} K\left(\frac{d_T((t_{1,j'},t)}{h_T}\right)} \int_{y_{1,j}\in\mathbb{R}^2} K\left(y_{1,j}\right) [f_{\epsilon}(x-s(t_{1,j}))+h_X\nabla f_{\epsilon}(x-s(t_{1,j}))^Ty_{1,j}\\
			&+h_X^2y_{1,j}^T\nabla\nabla  f_{\epsilon}(x-s(t_{1,j}))y_{1,j}-f_{\epsilon}(x-s(t))]dy_{1,j}+o(h_X^2)\\
			=&\sum_{j=1}^{m_1} \frac{K\left(\frac{d_T(t_{1,j},t)}{h_T}\right)}{\sum_{j'=1}^{m_1} K\left(\frac{d_T(t_{1,j},t)}{h_T}\right)} \int_{y_{1,j}\in\mathbb{R}^2} K\left(y_{1,j}\right)[f_{\epsilon}(x-s(t_{1,j}))-f_{\epsilon}(x-s(t))+\\
			&h_X^2y_{1,j}^T\nabla \nabla f_{\epsilon}(x-s(t_{1,j}))y_{1,j}] dy_{1,j}+o(h_X^2) \\
			\leq & M_{\epsilon}h_X^2+\sum_{j=1}^{m_1} \frac{K\left(\frac{d_T(t_{1,j},t)}{h_T}\right)}{\sum_{j'=1}^{m_1} K\left(\frac{d_T(t_{1,j'},t)}{h_T}\right)} f_{\epsilon}(x-s(t_{1,j}))-f_{\epsilon}(x-s(t)).
		\end{align*}
		Hence,
		\begin{align*}
			\mathbb{E} \hat f_{i}(x|t;s) - f_{G P S}(x|t; s) = O(h_X^2)+O\left(\sum_{j=1}^{m_i} \frac{K\left(\frac{d_T(t_{i,j},t)}{h_T}\right)}{\sum_{j'=1}^{m_i} K\left(\frac{d_T(t_{i,j'},t)}{h_T}\right)} f_{\epsilon}(x-s(t_{i,j}))-f_{\epsilon}(x-s(t))\right).
		\end{align*}
		for any $i\in\{1,...,n\}$
		
		\textbf{Step 2 (Variance):} 
		As for the variance, for any $s\in\mathcal{S}$,
		\begin{align}
			\text{Var}\left(\hat f_{1}(x|t;s)\right) =& \text{Var}\left(\frac{1}{h_X^2}\frac{\sum_{j=1}^{m_1} K\left(\frac{d_T(t_{1,j},t)}{h_T}\right)K\left(\frac{s(t_{1,j})+\epsilon_{1,j}-x}{h_X}\right)}{\sum_{j'=1}^{m_1} K\left(\frac{d_T(t_{1,j},t)}{h_T}\right)} \right)\nonumber\\
			=&\text{Var}\left(\frac{1}{h_X^2}\frac{\sum_{j=1}^{m_1} K\left(\frac{d_T(t_{1,j},t)}{h_T}\right)K\left(\frac{s(t_{1,j})+\epsilon_{1,j}-x}{h_X}\right)}{\sum_{j'=1}^{m_1} K_T\left(\frac{d_T(t_{1,j'},t)}{h_T}\right)} \right)\nonumber\\
			=& \frac{1}{h_X^4}\sum_{j=1}^{m_1} \left[\frac{K\left(\frac{d_T(
					,t)}{h_T}\right)}{\sum_{j=1}^{m_1} K\left(\frac{d_T(t_{1,j},t)}{h_T}\right)}\right]^2 \text{Var}\left(K\left(\frac{s(t_{1,j})+\epsilon_{1,j}-x}{h_X}\right) \right),\label{eq:variance-condition-1}
		\end{align}
		In \eqref{eq:variance-condition-1}, by Assumption \ref{ass:measurement deri bounded}, $f_{\epsilon}$ is upper-bounded by some constant $M_{\epsilon}$, then for every $j\in\{1,...,m_1\}$
		\begin{align}
			\text{Var}\left(K\left(\frac{s(t_{1,j})+\epsilon_{1,j}-x}{h_X}\right)\right)\leq &\mathbb{E}_{\epsilon_{1,j}}\left[K\left(\frac{s(t_{1,j})+\epsilon_{1,j}-x}{h_X}\right)^2\right]\nonumber\\
			=&\int_{\epsilon_{1,j}\in \mathbb{R}^2}K\left(\frac{s(t_{1,j})+\epsilon_{1,j}-x}{h_X}\right)^2 f_{\epsilon}(\epsilon_{1,j}) d\epsilon_{1,j}\nonumber \\
			=& h_X^2 \int_{y_{1,j}\in \mathbb{R}^2}K\left(y_{1,j}\right)^2 f_{\epsilon}(h_Xy_{1,j}+x-s(t_{1,j})) dy_{1,j} \nonumber\\
			\leq & M_{\epsilon}h_X^2\int_{x\in \mathbb{R}^2}K^2(x)dx,\label{eq:lemma-ckd:var-one obs}
		\end{align}
		for any $s\in\mathcal{S}$, where the last inequality is derived by Assumption \ref{ass:measurement deri bounded} that $f_{\epsilon}$ is bounded. Plugging \eqref{eq:lemma-ckd:var-one obs} to \eqref{eq:variance-condition-1}, we obtain
		\begin{align*}
			\text{Var}\left(\hat f_{i}(x|t;s)\right) = \frac{M_{\epsilon}}{h_X^2}\sum_{j=1}^{m_i} \left[\frac{K_T\left(\frac{d_T(t_{i,j},t)}{h_T}\right)}{\sum_{j'=1}^{m_i} K_T\left(\frac{d_T(t_{i,j'},t)}{h_T}\right)}\right]^2 \int_{x\in \mathbb{R}^2}K_1^2(x)dx
		\end{align*}
		for any $s\in\mathcal{S}$.  
	\end{proof}

	Now we formally state our proof of Theorem \ref{thm::ckde}.
	
	\begin{proof}
		Recall that we use $\hat f_{i}(x|t)$ to denote the $i-$th day's contribution to $ \hat f(x|t)$, i.e.  
		\begin{align*}
			\hat f_{k}(x|t) = \frac{1}{h_X^2}\frac{\sum_{j=1}^{m_i} K\left(\frac{d_T(t_{i,j},t)}{h_T}\right)K\left(\frac{X_{i,j}-x}{h_X}\right)}{\sum_{j'=1}^{m_i} K\left(\frac{d_T(t_{i,j'},t)}{h_T}\right)}.
		\end{align*}
		
		Then, $ \hat f(x|t)$ can be written as 
		\begin{align*}
			\hat f(x|t)=&\frac{1}{h_X^2}\frac{\sum_{i=1}^{n} \frac{1}{m_i}\sum_{j=1}^{m_i} K\left(\frac{d_T(t_{i,j},t)}{h_T}\right)K\left(\frac{X_{i,j}-x}{h_X}\right)}{\sum_{i=1}^{n} \frac{1}{m_i}\sum_{j'=1}^{m_i} K\left(\frac{d_T(t_{i,j'},t)}{h_T}\right)}     = \sum_{i=1}^{n} \alpha_i(t) \hat f_{i}(x|t).
		\end{align*}
		where $\alpha_i(t)=\frac{\frac{1}{m_i}\sum_{j=1}^{m_i} K\left(\frac{d_T(t_{i,j},t)}{h_T}\right)}{\sum_{i=1}^{n} \frac{1}{m_i}\sum_{j'=1}^{m_i} K\left(\frac{d_T(t_{i,j'},t)}{h_T}\right)}$ represents the kernel-sum weight for each day. Note that $\sum_{i=1}^{n}\alpha_i(t)=1$ for every $t\in[0,1]$.
		
		\textbf{Bias:} The expectation of $\hat f(x|t)$ is
		\begin{align*}
			\mathbb{E}\left(\hat f(x|t)\right)= \sum_{i=1}^{n}\alpha_i(t)\mathbb{E}\left[\hat f_{i}(x|t)\right],
		\end{align*}
		where, for any $i\in\{1,...,n\}$:
		\begin{align*}
			&\mathbb{E}\left[\hat f_{i}(x|t)\right] = \int_{s\in\mathcal{S}} \mathbb{E}\left[\hat f_{i}(x|t)|S_i=s\right] dF_{S}(s).
		\end{align*}
		And $f_{G P S}(x|t) $ can be written as
		\begin{align*}
			f_{G P S}(x|t) = \int_{s\in\mathcal{S}}f_{GPS}(x|t; s)dF_{S}(s).
		\end{align*}
		
		By the \eqref{eq:lemma-ckd:bias} in Lemma \ref{thm:single day, conditional, deter}, for any $s\in\mathcal{S}$,
		\begin{align*}
			& \mathbb{E}\left[\hat f_{i}(x|t;s)\right]-f_{GPS}(x|t; s)= O(h_X^2)+O\left(\sum_{j=1}^{m_i} \frac{K\left(\frac{d_T(t_{i,j},t)}{h_T}\right)}{\sum_{j'=1}^{m_i} K\left(\frac{d_T(t_{i,j'},t)}{h_T}\right)} f_{\epsilon}(x-s((t_{i,j}))-f_{\epsilon}(x-s(t))\right).
		\end{align*}
		So,
		\begin{align*}
			&\mathbb{E}\left[\hat f(x|t)\right] - f_{GPS}(x|t)\\
			=&\sum_{i=1}^{n}\alpha_i(t)\int_{S\in\mathcal{s}} \mathbb{E}\left[\hat f_{i}(x|t)|S_i=s\right] -f_{GPS}(x|t; s)dF_{S}(s)\\
			=& O(h_X^2)+O\left(\sum_{i=1}^{n}\alpha_i(t)\int_{s\in\mathcal{S}}\left[\sum_{j=1}^{m_i} \frac{K\left(\frac{d_T(t_{i,j},t)}{h_T}\right)}{\sum_{j'=1}^{m_i} K\left(\frac{d_T(t_{i,j'},t)}{h_T}\right)} f_{\epsilon}(x-s(t_{i,j}))-f_{\epsilon}(x-s(t))\right]dF_{S}(s)\right).
		\end{align*}
		Furthermore, by Assumption \ref{ass:XR-2-derivative} and Assumption \ref{ass:measurement deri bounded}, note that there exists constant $L_{\epsilon}$ and $\bar v$ such that
		\begin{align*}
			f_{\epsilon}(x-s(t_{i,j}))-f_{\epsilon}(x-s(t)) \leq L_{\epsilon}|s(t_{i,j})-s(t)|\leq L_{\epsilon}\bar v d_T(t_{i,j},t),
		\end{align*}
		hold for any $t_{i,j}\in [0,1]$, $s\in\mathcal{S}$ and $x\in\mathcal{W}$. Then,
		\begin{align}
			\mathbb{E}\hat f(x|t)- f_{G P S}(x|t) = O(h_X^2)+O\left(\sum_{i=1}^{n}\alpha_i(t)\sum_{j=1}^{m_i} \frac{K_T\left(\frac{d_T(t_{i,j},t)}{h_T}\right)}{\sum_{j'=1}^{m_i} K_T\left(\frac{d_T(t_{i,j'},t)}{h_T}\right)}d_T(t_{i,j},t)\right) .\label{eq:bias-mul-con-deter1}
		\end{align}
		
		\textbf{Variance:} By the independence of mobility among each day, 
		\begin{align*}
			\text{Var}\left[\hat f(x|t)\right] =& \text{Var}\left[\sum_{i=1}^{n} \alpha_i(t) \hat f_{i}(x|t)\right]=\sum_{i=1}^{n} \alpha_i(t)^2 \text{Var}\left[\hat f_{i}(x|t)\right].
		\end{align*}
		By $\operatorname{Var}(Y)=\mathbb{E}(\operatorname{Var}(Y \mid X))+\operatorname{Var}(\mathbb{E}(Y \mid X))$, for every $i\in \{1,...,n\}$, we have
		\begin{align}
			&\text{Var}\left[\hat f_{i}(x|t)\right] =\mathbb{E}\left\{\text{Var}\left[\hat f_{i}(x|t)\mid S_i\right] \right\}+\text{Var}\left\{\mathbb{E}\left[\hat f_{i}(x|t)\mid S_i\right] \right\}.\label{eq:mul-cond-1}
		\end{align}
		WLOG, we can simply let $i=1$. 
		
		For the term $\mathbb{E}\left\{\text{Var}\left[\hat f_{1}(x|t)\mid S_1\right] \right\}$ in RHS of \eqref{eq:mul-cond-1}, for any $s\in \mathcal{S}$, by \eqref{eq:lemma-ckd:var} in Lemma \ref{thm:single day, conditional, deter}, we have
		\begin{align*}
			&\text{Var}\left[\hat f_{1}(x|t)\mid S_1=s\right] =\text{Var}\left[\hat f_{1}(x|t;s)\right]    \leq \frac{M_{\epsilon}}{h_X^2}\sum_{j=1}^{m_1} \left[\frac{K\left(\frac{d_T(t_{1,j},t)}{h_T}\right)}{\sum_{j'=1}^{m_1} K\left(\frac{d_T(t_{1,j'},t)}{h_T}\right)}\right]^2\int_{x\in \mathbb{R}^2}K^2(x)dx.
		\end{align*}
		So, by Assumption \ref{ass:K1,K2} (boundness of $\int_{x\in \mathbb{R}^2}K_1^2(x)dx$), we have
		\begin{equation}
			\begin{aligned}
				\sum_{i=1}^{n} \alpha_i(t)^2\mathbb{E}\left\{\text{Var}\left[\hat f_{i}(x|t)\mid S_k\right] \right\}&=\sum_{i=1}^{n} \alpha_i(t)^2\int_{z\in\mathcal{Z}} \text{Var}\left[\hat f_{i}(x|t;s)\right] dF_{S}(s)\\
				&=O\left(\frac{1}{h_X^2}\sum_{i=1}^{n} \alpha_i(t)^2\sum_{j=1}^{m_i} \left[\frac{K\left(\frac{d_T(t_{i,j},t)}{h_T}\right)}{\sum_{j'=1}^{m_i} K\left(\frac{d_T(t_{i,j'},t)}{h_T}\right)}\right]^2\right).
			\end{aligned}
			\label{eq:mul-Evar1}
		\end{equation}
		
		For the term $\text{Var}\left\{\mathbb{E}\left[\hat f_{1}(x|t)\mid S_1\right] \right\}$ in the RHS of \eqref{eq:mul-cond-1}, for any $s\in \mathcal{S}$, by Assumption \ref{ass:measurement deri bounded}, $f_{\epsilon}$ is bounded by some constant $M_{\epsilon}$, we have
		\begin{align}
			&\mathbb{E}\left[\hat f_{1}(x|t)\mid S_1=s\right] \nonumber\\
			=& \frac{1}{h_X^2}\sum_{j=1}^{m_1} \frac{K\left(\frac{d_T(t_{1,j},t)}{h_T}\right)}{\sum_{j'=1}^{m_1} K\left(\frac{d_T(t_{1,j'},t)}{h_T}\right)} \int_{\epsilon_{1,j}\in\mathbb{R}^2} K\left(\frac{s(t_{1,j})+\epsilon_{1,j}-x}{h_X}\right)f_{\epsilon}(\epsilon_{1,j})d\epsilon_{1,j} \nonumber \\
			=&\sum_{j=1}^{m_1} \frac{K\left(\frac{d_T(t_{1,j},t)}{h_T}\right)}{\sum_{j'=1}^{m_1} K\left(\frac{d_T(t_{1,j'},t)}{h_T}\right)} \int_{y_{1,j}\in\mathbb{R}^2} K\left(y_{1,j}\right)f_{\epsilon}(h_Xy_{1,j}+x-s(t_{1,j}))dy_{1,j} \nonumber \\
			\leq & M_{\epsilon}.\nonumber
		\end{align}
		Hence,
		\begin{align}
			\sum_{i=1}^{n} \alpha_i(t)^2\text{Var}\left\{\mathbb{E}\left[\hat f_{i}(x|t)\mid Z_i\right] \right\} = O\left({\sum_{i=1}^{n} \alpha_i(t)^2}\right).\label{eq:mul-varE1}
		\end{align}
		Combining \eqref{eq:mul-Evar1} and \eqref{eq:mul-varE1}, we have
		\begin{align}
			\hat f(x|t)-\mathbb{E}\left(\hat f(x|t)\right) =&O_p\left(\sqrt{\sum_{i=1}^{n} \alpha_i(t)^2}\right)+O_p\left( \frac{1}{h_X}\sqrt{\sum_{i=1}^{n} \alpha_i(t)^2\sum_{j=1}^{m_i} \left[\frac{K\left(\frac{d_T(t_{i,j},t)}{h_T}\right)}{\sum_{j'=1}^{m_k} K\left(\frac{d_T(t_{k,j'},t)}{h_T}\right)}\right]^2}\right).\label{eq:var-mul-con-deter1}
		\end{align}
		Combining \eqref{eq:bias-mul-con-deter1} and \eqref{eq:var-mul-con-deter1}, we get
		\begin{align*}
			\hat f(x|t)- f_{G P S}(x|t) = &O(h_X^2)+O\left(\sum_{i=1}^{n}\alpha_i(t)\sum_{j=1}^{m_i} \frac{K\left(\frac{d_T(t_{i,j},t)}{h_T}\right)}{\sum_{j'=1}^{m_i} K\left(\frac{d_T(t_{i,j'},t)}{h_T}\right)}d_T(t_{i,j},t)\right)+O_p\left(\sqrt{\sum_{i=1}^{n} \alpha_i(t)^2}\right)+\\
			&O_p\left( \frac{1}{h_X}\sqrt{\sum_{i=1}^{n} \alpha_i(t)^2\sum_{j=1}^{m_i} \left[\frac{K\left(\frac{d_T(t_{i,j},t)}{h_T}\right)}{\sum_{j'=1}^{m_i} K\left(\frac{d_T(t_{i,j'},t)}{h_T}\right)}\right]^2}\right).
		\end{align*}
		
	\end{proof}

\end{document}